\documentclass[a4paper,dvips]{book}

\usepackage[latin1]{inputenc}
\usepackage{amsthm,amssymb,amsmath}
\usepackage{wasysym} 
\usepackage{bm}
\usepackage{url}
\usepackage{slashed}
\usepackage[all]{xy}
\CompileMatrices

\usepackage{virginialake}

 \usepackage[%
bookmarks,%
colorlinks=false,%
pdftitle={Nested Sequents},%
pdfauthor={Kai Brunnler},%
pdfpagemode=UseOutlines]{hyperref}%

\newcommand{\ftextwidth}{.97\textwidth}

\newcommand{\vltrnote}[6]{\vltrl{#1}{}{#6}{#2}{#3}{#4}{#5}}
\newcommand{\Proofleaf}[2]%
{\vltr{#1}{#2}{\vlhy{}}{\vlhy{\qquad}}{\vlhy{}}}

\newcommand{\Proof}[2]%
{\vlderivation{
\vltr{#1}{#2}{\vlhy{}}{\vlhy{\qquad}}{\vlhy{}}}}

\begin{document}

\newtheoremstyle{kaithm}{\topsep}{\topsep}%
{}
{}
{\bfseries}
{}
{ }
{\thmname{#1}\thmnumber{ #2} \thmnote{(#3)}}

\theoremstyle{kaithm}
  \newtheorem{theorem}{Theorem}[chapter]
  \newtheorem{lemma}[theorem]{Lemma}
  \newtheorem{corollary}[theorem]{Corollary}
  \newtheorem{proposition}[theorem]{Proposition}
  \newtheorem{conjecture}[theorem]{Conjecture}
  \newtheorem{fact}[theorem]{Fact}
  \newtheorem{problem}[theorem]{Problem}
  \newtheorem{definition}[theorem]{Definition}
  \newtheorem{remark}[theorem]{Remark}
  \newtheorem{example}[theorem]{Example}
  \newtheorem{notation}[theorem]{Notation}


 \newcommand{\kaiinf}[3]{{#1} \;
 \begin{array}{l}
   {#3} \\\hline
   {#2}\\
 \end{array}
}

\newcommand{\cinf}[3]{{#1} \;
 \begin{array}{c}
   {#3} \\\hline
   {#2}\\
 \end{array}
}

\newcommand{\K}{{\sf K}}
\newcommand{\Kc}{{\ensuremath{\sf K_c}}}
\newcommand{\X}{{\sf X}}
\newcommand{\Y}{{\sf Y}}
\newcommand{\Z}{{\sf Z}}

\newcommand{\ctr}{{\sf ctr}}
\newcommand{\wk}{{\sf wk}}
\newcommand{\nec}{{\sf nec}}
\newcommand{\cut}{{\sf cut}}
\newcommand{\ycut}{{\text{\sf Y-cut}}}
\newcommand{\ystr}{{\text{\sf Y-str}}}
\newcommand{\Km}{{\ensuremath{\sf K_m}}}

\newcommand{\dia}[1]{\Diamond #1}
\newcommand{\diac}[1]{\Diamond #1_{\sf c}}
\newcommand{\str}[1]{[#1]}

\newcommand{\mcut}{{\text{\sf mcut}}}
\newcommand{\mboxx}{{\sf m}\Box}
\newcommand{\mand}{{\sf m}\vlan}
\newcommand{\fctr}{{\sf fctr}}

\newcommand{\diak}{{\ensuremath{\sf \Diamond{k}}}}
\newcommand{\diad}{{\ensuremath{\sf\Diamond{d}}}}
\newcommand{\diat}{{\ensuremath{\sf\Diamond{t}}}}
\newcommand{\diab}{{\ensuremath{\sf\Diamond{b}}}}
\newcommand{\diafour}{{\ensuremath{\sf\Diamond{4}}}}
\newcommand{\diafive}{{\ensuremath{\sf\Diamond{5}}}}
\newcommand{\diafivepa}{{\ensuremath{\sf\Diamond{5_1}}}}
\newcommand{\diafivepb}{{\ensuremath{\sf\Diamond{5_2}}}}
\newcommand{\diafivepc}{{\ensuremath{\sf\Diamond{5_3}}}}

\newcommand{\kc}{{\ensuremath{\sf \Diamond{k}_c}}}
\newcommand{\diadc}{{\ensuremath{\sf\Diamond{d}_c}}}
\newcommand{\diadccirc}{{\ensuremath{\sf\Diamond{d}_c^\circ}}}
\newcommand{\diatc}{{\ensuremath{\sf\Diamond{t}_c}}}
\newcommand{\diabc}{{\ensuremath{\sf\Diamond{b}_c}}}
\newcommand{\diafourc}{{\ensuremath{\sf\Diamond{4}_c}}}
\newcommand{\diafivec}{{\ensuremath{\sf\Diamond{5}_c}}}
\newcommand{\diafivepac}{{\ensuremath{\sf\Diamond{5_1}_c}}}
\newcommand{\diafivepbc}{{\ensuremath{\sf\Diamond{5_2}_c}}}
\newcommand{\diafivepcc}{{\ensuremath{\sf\Diamond{5_3}_c}}}

\newcommand{\diacrules}{\diadc, \diatc, \diabc, \diafourc, \diafivec}
\newcommand{\modax}{{\{\sf d,t,b,4,5\}}}
\newcommand{\modaxnod}{{\{\sf t,b,4,5\}}}

\newcommand{\strd}{{\ensuremath{\sf [d]}}}
\newcommand{\strt}{{\ensuremath{\sf [t]}}}
\newcommand{\strb}{{\ensuremath{\sf [b]}}}
\newcommand{\strfour}{{\ensuremath{\sf [4]}}}
\newcommand{\strfive}{{\ensuremath{\sf [5]}}}
\newcommand{\strfivepa}{{\ensuremath{\sf [5a]}}}
\newcommand{\strfivepb}{{\ensuremath{\sf [5b]}}}
\newcommand{\strfivepc}{{\ensuremath{\sf [5c]}}}

\newcommand{\med}{{\sf med}}
\newcommand{\D}{\mathcal{D}}
\newcommand{\PP}{\mathcal{P}}

\newcommand{\rel}{\rightarrow}
\newcommand{\invrel}{\leftarrow}
\newcommand{\symrel}{\leftrightarrow}

\newcommand{\M}{{\mathcal M} }
\newcommand{\Sys}{{\mathcal S} }

\newcommand{\set}{\mathit{set}}
\newcommand{\st}{\mathit{st}}
\newcommand{\context}{\{\;\}}
\newcommand{\depth}[1]{\mathit{depth}(#1)}
\newcommand{\prove}[1]{\mathit{prove}(#1)}
\newcommand{\tree}[1]{\mathit{tree}(#1)}

\newcommand{\inverse}[1]{\overline{#1}}
\newcommand{\sufo}[1]{{\mathit{sf}(#1)}}

\newcommand{\nl}{\hspace{0mm}\\}

\def\grammareq {\mathrel{\raise.4pt\hbox{::}{=}}}%
\newcommand{\imp}{\vljm}%
\newcommand{\cons}[1]{{\{#1\}}}%

\renewcommand{\neg}[1]{\bar{#1}}
\def\true {{\sf{t}}}%
\def\false {{\sf{f}}}%

\newcommand{\formula}[1]{{\underline{#1}_{\scriptscriptstyle \sf F}}}

\frontmatter

\begin{titlepage}
\hspace{0ex}
\vfill
\begin{center}
 \textbf{\Huge Nested Sequents}  \\[5ex]
\Large Habilitationsschrift\\[5ex]
{\bf Kai~Br\"unnler}\\[1ex]
Institut f\"ur Informatik und\\
angewandte Mathematik\\
Universit\"at Bern\\[1ex]
\today\\
\end{center}
\vfill
\end{titlepage}

\pagestyle{empty}

\pagebreak

\chapter*{Abstract}

We see how \emph{nested sequents}, a natural generalisation of
hypersequents, allow us to develop a systematic proof theory for modal
logics.  As opposed to other prominent formalisms, such as the display
calculus and labelled sequents, nested sequents stay inside the modal
language and allow for proof systems which enjoy the subformula
property in the literal sense.

In the first part we study a systematic set of nested sequent systems
for all normal modal logics formed by some combination of the axioms
for seriality, reflexivity, symmetry, transitivity and euclideanness.
We establish soundness and completeness and some of their good
properties, such as invertibility of all rules, admissibility of the structural
rules, termination of proof-search, as well as syntactic
cut-elimination.

In the second part we study the logic of common knowledge, a modal
logic with a fixpoint modality. We look at two infinitary proof
systems for this logic: an existing one based on ordinary sequents,
for which no syntactic cut-elimination procedure is known, and a new
one based on nested sequents. We see how nested sequents, in contrast
to ordinary sequents, allow for syntactic cut-elimination and thus
allow us to obtain an ordinal upper bound on the length of proofs.

\pagebreak

\tableofcontents

\mainmatter

\pagestyle{headings}

\chapter{Introduction}

{\bf The problem of the proof theory of modal logic.} The proof theory
of modal logic as developed in Gentzen's sequent calculus is widely
recognised as unsatisfactory: it provides systems only for a few modal
logics, and does so in a non-systematic way. To solve this problem,
many extensions of the sequent calculus have been proposed. The survey
by Wansing \cite{Wan02} discusses many of them. The three most
prominent formalisms seem to be the \emph{hypersequent calculus}, due
to Avron \cite{avron96method}, the \emph{display calculus} due to
Belnap \cite{Belnap82,Wan98}, and \emph{labelled sequent systems},
which have been introduced and studied by many researchers. The book
by Vigan{\`o} \cite{Vigano00a} and the article by Negri \cite{Neg05}
provide a recent account of labelled sequent systems where more
references can be found.

{\bf Hypersequents, display calculus, and labelled sequents.} The
relationship between these formalisms might be summarised as follows.
The hypersequent calculus is a comparatively gentle extension of the
sequent calculus, in particular it allows for a subformula property in
the literal sense. Both the display calculus and the labelled sequent
calculus are departing further from the ordinary sequent calculus, in
particular they only satisfy weaker forms of the subformula property.
On the other hand, both the display calculus and labelled systems are
more expressive than hypersequents.  They are known to capture all the
basic modal logics that we are going to consider here, which is not
true for hypersequents. In fact, the only modal logic captured so far
in the hypersequent calculus, that has not been captured in the
ordinary sequent calculus, is the modal logic {\sf S5}. In general,
there seems to be a tension between the desire to have a formalism
which is expressive and the desire to have a formalism in which
cut-free proofs are simple objects with a true subformula property.

{\bf Staying inside the modal language.} A hypersequent is a sequence
of ordinary sequents and can be read as a formula of modal logic: it
is a disjunction where all disjuncts are prefixed by a box modality. A
display sequent generally does not correspond to a modal formula: it
contains structural connectives which correspond to backward-looking
modalities, so connectives of tense logic. Similarly, a labelled
sequent does not correspond to a formula of modal logic: it contains
variables and an accessibility relation, so notions from predicate
logic. In this sense, hypersequents stay inside the modal language,
while display calculus and labelled sequents do not. In this work, we
develop a proof theory for modal logic which aims to be as systematic
and expressive as the display calculus and labelled sequents, but
stays within the modal language and allows for a true subformula
property, like hypersequents.

{\bf Nested sequents.} To that end, we use \emph{nested sequents},
which are essentially trees of sequents. They naturally generalise
both sequents (which are nested sequents of depth zero) and hypersequents
(which essentially are nested sequents of depth one). The notion of
nested sequent has been invented several times independently. Bull
\cite{Bull92} gives a proof system based on nested sequents for a
fragment of propositional dynamic logic with converse. Kashima
\cite{Kash94} gives proof systems for some tense logics and attributes
the idea to Sato \cite{Sato77}. Unaware of these works, the author
introduced the same notion of nested sequent under the name \emph{deep
  sequent} in \cite{Bru06ds}. Poggiolesi introduced again the same
notion but with a rather different notation under the name
\emph{tree-hypersequent} \cite{Pog09}. Nested sequents are also used
by Gor{\'e} et al.\ to give a proof system for bi-intuitionistic logic
which is suitable for proof-search \cite{GorPosTiu08}.

{\bf Deep inference.} Nested sequents are tree-like structures with
formulas occurring deeply inside of them.  The proof systems
introduced in this work crucially rely on being able to apply
inference rules to all formulas, including those deeply inside.  The
general idea of applying rules deeply has been proposed several times
in different forms and for different purposes. Sch\"utte already used
it in the 1950s in order to obtain systems without contraction and
weakening, which he considered more elegant \cite{SchuSKP}.  Guglielmi
developed a formalism which is centered around applying rules deeply
and which replaces the traditional tree-format of sequent calculus
proofs by a linear format \cite{Gugl07}. This solved the problem of finding
a proof-theoretic system for a certain substructural logic which
cannot be captured in the sequent calculus. The name of this formalism
used to be \emph{calculus of structures} but is now simply \emph{deep
  inference}. Deep inference systems then have also been developed for
some modal logics \cite{HeinStew05,StSt,Sto06,GoreTiu07}. The design
of the proof systems in this work is inspired by deep inference. We
will see the precise connection between nested sequent systems and
deep inference systems later.

\textbf{The big picture.} This work is a case study in designing
proof-theoretic systems for non-classical logics. It is an instance of
the widely-known phenomenon that the notion of sequent, so the
structural level of the proof system, has to be extended in order to
accommodate certain logics. Our methodology here is that the structural
level is not extended by arbitrary structural connectives, but only
by those from the logic. As we do this, sequents become nested
structures and so more formula-like. It then turns out that we need
to allow inference rules to apply inside of these nested structures in
order to obtain complete cut-free proof systems. There are many other
instances of this phenomenon.  The \emph{logic of bunched}
implications by Pym \cite{Pym02} is a substructural logic which has
both a multiplicative and an additive conjunction. The proof systems
for this logic have two corresponding structural connectives, which
can be nested. Logics with non-associative conjunction also naturally lead
to sequents which are nested structures, for example the
non-associative Lambek calculus which can be found in the handbook
article by Moortgat \cite{Moortgat97}. Another example are the proof
systems for logics with \emph{adjoint modalities}, certain epistemic
logics for reasoning about information in a multi-agent system, by
Dyckhoff and Sadrzadeh \cite{RoySadr09}.

\textbf{The plan.} In the following there are two chapters which are
independent.  In the first chapter we study nested sequent systems for
all normal modal logics formed by some combination of the axioms for
seriality, reflexivity, symmetry, transitivity and euclideanness. We
establish soundness and completeness and some of their good
properties, such as invertibility, admissibility of the structural
rules, termination of proof-search, as well as syntactic
cut-elimination. This chapter contains work from
\cite{Bru06ds,Bru08dsl} and also from \cite{BrSt09} which is joint
work with Lutz Stra{\ss}burger.

In the second chapter we study the logic of common knowledge, a modal
logic with a fixpoint modality. We look at two infinitary proof
systems for this logic: an existing one based on ordinary sequents,
for which no syntactic cut-elimination procedure is known, and a new
one based on nested sequents. We see how nested sequents, in contrast
to ordinary sequents, allow for syntactic cut-elimination and thus
allow us to obtain an ordinal upper bound on the length of proofs.  This
chapter contains work from \cite{BruStu07,BrSt09} which are joint
work with Thomas Studer.

{\bf Acknowledgements.}  This work benefited from discussions with Roy
Dyckhoff, Rajeev Gor{\'e}, Gerhard J\"ager, Roman Kuznets, Richard
McKinley, Dieter Probst, Thomas Strahm, Lutz Stra{\ss}burger,
Thomas Studer and Alwen Tiu. Special thanks go to Alessio Guglielmi
for his constant support and for his LaTeX macros.

\chapter{Systems for Basic Normal Modal Logics}

In this chapter we consider modal logics formed from the least normal
modal logic $\K$ by adding axioms from the set $\modax$ which is shown
in Figure~\ref{fig:framecond}. This gives rise to the modal logics
shown in Figure~\ref{fig:cube}. In the first section we consider
sequent systems in which modal axioms are turned into logical rules,
namely rules for the $\Diamond$-connective. For each modal logic we
find a corresponding cut-free sequent system which is sound and
complete for this logic. However, some modal logics can be axiomatised
in different ways, for example {\sf S5} can be axiomatised as
$\K+\{{\sf t,b,4}\}$ and as $\K+\{{\sf t,5}\}$. Without cut, some of
these axiomatisations turn out to be incomplete. For those cut-free
systems which are complete we give a syntactic cut-elimination
procedure, in the course of which we discover certain structural modal
rules. In the second section we then study sequent systems where modal
axioms are formalised not by using logical rules, but by using the
structural modal rules we just found. This turns out to yield cut-free
systems where each possible way of axiomatising a modal logic is
complete.

At the end of the chapter we discuss some related formalisms.

\begin{figure}
\centering
\fbox{\parbox{\ftextwidth}{
\centering
\vspace{1ex}
  \begin{tabular}{lll}
    {\sf k}:  no condition & $\top$ & $\Box(A \vlor B) \imp (\Box A \vlor \Diamond B)$\\
    {\sf d}:  serial & $\forall s\exists t. \; s \rel t$ & $\Box A \imp \Diamond A$\\
    {\sf t}:  reflexive & $\forall s. \; s \rel s$ & $A \imp \Diamond A$\\
    {\sf b}:  symmetric & $\forall st. \; s \rel t \, \imp \, t \rel s$ & $A \imp \Box \Diamond A$\\
    {\sf 4}:   transitive & $\forall stu. \; s \rel t \, \vlan \, t
    \rel u \, \imp \, s \rel u$ & $\Box A \imp \Box\Box A$\\
    {\sf 5}:   euclidean \quad &  $\forall stu. \; s \rel t \, \vlan \, s
    \rel u \, \imp \, t \rel u$ \quad & $\Diamond A \imp \Box \Diamond A$\\
  \end{tabular}}
\vspace{1ex}
}
  \caption{Frame conditions and modal axioms}
  \label{fig:framecond}
\end{figure}

\begin{figure}
  \centering
\[
  \xymatrix@R-3ex{ {} & *{\circ} \save
    []+<-1.5ex,+1.5ex>*\txt{\scriptsize \sf S4}\restore \ar@{-}[rrrr]
    \ar@{-}[dd] \ar@{-}[ld] & {} & {} & {} & *{\circ}
    \save []+<1.5ex,+1.5ex>*\txt{\scriptsize \sf S5}\restore \ar@{-}[ld] \ar@{-}[ddddd] \\
    *{\circ} \save []+<-1.5ex,+1.5ex>*\txt{\scriptsize \sf T}\restore
    \ar@{-}[rrrr] \ar@{-}[ddd] & {} & {} & {} &
    *{\circ} \save []+<-1.5ex,+1.5ex>*\txt{\scriptsize \sf TB}\restore \ar@{-}[ddd] & {} \\
    {} & *{\circ} \save []+<-1.5ex,+1.5ex>*\txt{\scriptsize \sf
      D4}\restore \ar@{-}[ddd] \ar@{-}[rr] & {} & *{\circ} \save
    []+<2ex,-1.5ex>*\txt{\scriptsize \sf D45}\restore
    \ar@{-}@/_/[uurr] & {} & {} \\
    {} & {} & *{\circ} \save []+<2ex,-1.5ex>*\txt{\scriptsize \sf
      D5}\restore
    \ar@{-}[ur] & {} & {} & {} \\
    *{\circ} \save []+<-1.5ex,0ex>*\txt{\scriptsize \sf D}\restore
    \ar@{-}[ddd] \ar@{-}[ruu] \ar@{-}[urr] \ar@{-}[rrrr]& {} & {} & {}
    & *{\circ} \save []+<2.2ex,0ex>*\txt{\scriptsize \sf DB}\restore
    \ar@{-}[ddd] & {} \\
    {} & *{\circ} \save []+<-1.5ex,+1.5ex>*\txt{\scriptsize \sf
      K4}\restore \ar@{-}[ldd] \ar@{-}[rr]& {} & *{\circ} \save
    []+<1.5ex,-1.5ex>*\txt{\scriptsize \sf K45}\restore \ar@{-}[uuu]
    \ar@{-}[rr] & {} &
    *{\circ} \save []+<2.5ex,-1ex>*\txt{\scriptsize \sf KB5}\restore \ar@{-}[ldd] \\
    {} & {} & *{\circ} \save []+<+1.5ex,-1.5ex>*\txt{\scriptsize \sf
      K5}\restore \ar@{-}[ru] \ar@{-}[uuu] & {} & {} & {} \\
    *{\circ} \save []+<-1.5ex,-1.5ex>*\txt{\scriptsize \sf K}\restore
    \ar@{-}[rrrr] \ar@{-}[rru] & {} & {} & {}
    & *{\circ} \save []+<1.5ex,-1.5ex>*\txt{\scriptsize \sf KB}\restore & {} \\
  }
  \]
  \caption{The modal ``cube'' \cite{GarStanEncy}}
\label{fig:cube}
\end{figure}

\section{Modal Axioms as Logical Rules}

The plan of this section is as follows: we first introduce the sequent
systems and then we see that they are sound and complete for the
respective Kripke semantics.  After that we see the syntactic
cut-elimination procedure.

\subsection{The Sequent Systems}

{\bf Formulas.} Propositions $p$ and their negations $\neg p$ are
\emph{atoms}, with $\neg{\neg p}$ defined to be $p$. Atoms are denoted
by $a, b, c, d$. \emph{Formulas}, denoted by $A,B,C,D$ are given by
the grammar
\[
A \grammareq p \mid \neg p \mid (A \vlor A) \mid (A \vlan A) \mid
\Diamond A \mid \Box A \quad .
\]

Given a formula $A$, its \emph{negation} $\neg A$ is defined as usual
using the De Morgan laws, $A \supset B$ is defined as $\neg A \vlor B$
and $\bot$ and $\top$ are defined as $p \vlan \neg p$ and $p \vlor
\neg p$, respectively, for some proposition $p$. Binary connectives
are left-associative: $A \vlor B \vlor C$ denotes $((A \vlor B)
\vlor C)$, for example.

\textbf{Nested sequents.} The set of \emph{nested sequents} is
inductively defined as follows: 
\begin{enumerate}
\item a finite multiset of formulas is a nested sequent,
\item  the multiset union of two nested sequents is a nested sequent,
\item if $\Gamma$ is a nested sequent then the singleton multiset
  containing $[\Gamma]$ is a nested sequent.
\end{enumerate}
In the following a \emph{sequent} is a nested sequent.  Sequents are
denoted by $\Gamma$,$\Delta$,$\Lambda$,$\Pi$ and $\Sigma$. We adopt
the usual notational conventions for sequents, in particular the comma
in the expression $\Gamma, \Delta$ is multiset union and there is no
distinction between a singleton multiset and its element. A sequent of
the form $[\Gamma]$ is also called a \emph{boxed sequent}. Clearly, a
sequent is always a multiset of formulas and boxed sequents,
so it is of the form
\[
A_1,\dots,A_m,[\Delta_1],\dots,[\Delta_n] \rlap{\quad .}
\]
We assume a fixed arbitrary linear order on formulas and another fixed
arbitrary linear order on boxed sequents.  The \emph{corresponding
  formula} of a sequent $\Gamma$, denoted $\formula{\Gamma}$,
is defined as follows: the corresponding formula of a sequent as given
above is $\bot$ if $m=n=0$ and otherwise it is
\[
A_1 \vlor \dots \vlor A_m 
\vlor \Box(\formula{\Delta_1}) 
\vlor \dots \vlor \Box(\formula{\Delta_n})
\rlap{\quad ,}
\]
where formulas and boxed sequents are list according to the fixed
orders.  Often we do not distinguish between a sequent and its
corresponding formula, for example a model of a sequent is a model of
its corresponding formula. A sequent $\Gamma$ has a
\emph{corresponding tree}, denoted $\tree{\Gamma}$, whose nodes are
marked with multisets of formulas. The corresponding tree of the above
sequent is
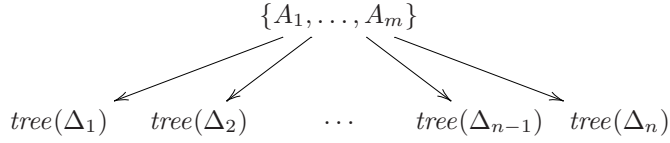
\[
\vcenter{\xymatrix@!C=1cm{{}  & {} & {\{A_1,\dots,A_m\}} \ar[dll] \ar[dl] \ar[dr] \ar[drr]\\
    {\tree{\Delta_1}} & {\tree{\Delta_2}} & {\dots} &
    {\tree{\Delta_{n-1}}}& {\tree{\Delta_n}}}} \qquad .
\]
Often we do not distinguish between a sequent and its corresponding
tree, for example the \emph{root} of a sequent is the root of its
corresponding tree.

{\bf Sequent contexts, unary.} Informally, a context is a sequent with
holes.  We will mostly encounter sequents with just one hole. To mark
the place of a hole in a sequent we use the symbol $\context$, called
the \emph{hole}. We inductively define the set of \emph{unary
  contexts}: 
\begin{enumerate}
\item the multiset containing a single hole is a unary context,
\item the multiset
  union of a sequent and a unary context is a unary context, and
\item given a unary context $\cal C$, the multiset containing a single
  occurrence of $[\cal C]$ is a unary context.
\end{enumerate}
Unary contexts are
denoted by $\Gamma\context, \Delta\context$ and so on. The multiset
containing a single hole is also called the \emph{empty context}. Our
conventions for writing sequents also apply to sequent contexts, in
particular comma denotes multiset union. The \emph{depth} of a unary
context $\Gamma\context$, denoted $\depth{\Gamma\context}$ is defined
as follows:
\begin{enumerate}
\item $\depth{\context} = 0$
\item $\depth{\Gamma,\Delta\context} = \depth{\Delta\context}$
\item  $\depth{[\Delta\context]} = \depth{\Delta\context} +1$ \quad .
\end{enumerate}
Given a unary context $\Gamma\context$ and a sequent $\Delta$ we can
obtain the sequent $\Gamma\cons{\Delta}$ by filling the hole in
$\Gamma\context$ with $\Delta$. Formally, $\Gamma\cons{\Delta}$ is
defined inductively as follows: 
\begin{enumerate}
\item if $\,\Gamma\context=\context\,$ then $\,\Gamma\cons{\Delta}=\Delta$,
\item  if $\,\Gamma\context=\Gamma_1,\Gamma_2\context\,$ then
  $\,\Gamma\cons{\Delta}=\Gamma_1,\Gamma_2\cons{\Delta}\,$ and
\item  if $\,\Gamma\context=[\Gamma_1\context]\,$ then
  $\,\Gamma\cons{\Delta}=[\Gamma_1\cons{\Delta}]$\quad.
\end{enumerate}
\begin{example}
  Given the unary context $\Gamma\context=A,[[B],\context]$ and the
  sequent $\Delta=C,[D]$ we can obtain the sequent
  \[
  \Gamma\cons{\Delta}=A,[[B],C,[D]] \quad .
  \]
\end{example}

{\bf Sequent contexts, generally.} We want to allow multiple holes in
a context and we want to allow filling holes with contexts, not just
sequents. This is conceptually straightforward and formally somewhat
technical, so the reader is invited to skip to Example~\ref{excon}. To
keep track of the order of holes we index them with a number $i>0$ as
in $\context_i$.  Later the indices will never be shown since holes in
a context are of course naturally ordered when written down on paper.
We inductively define the set of \emph{precontexts}: 
\begin{enumerate}
\item a multiset containing a single hole $\context_i$ with $i>0$ is a precontext,
\item a  multiset containing a single formula is a precontext,
\item the multiset  union of two precontexts is a precontext, and
\item given a precontext
  $\cal C$, the multiset containing a single occurrence of $[\cal C]$
  is a precontext.
\end{enumerate}
The \emph{arity} of a context is the number of holes occurring in it. A
\emph{sequent context}, or just \emph{context}, is a precontext of
arity $n$ such that for each $i\leq n$ the hole $\context_i$ occurs
exactly once in it. Notice that sequents are exactly the contexts of
arity zero and, disregarding the index on the hole, unary contexts are
exactly the contexts of arity one.  A context of arity $n$ is denoted
by
\[
\Gamma\underbrace{\context\dots\context}_{n-\text{times}} \quad .
\]

\newcommand{\C}{{\cal C}}

Given an $n$-ary context
$\Gamma\context\dots\context$ and $n$ contexts ${\cal
  C}_{1},\dots,{\cal C}_{n}$ we can obtain the context
\[
\Gamma\cons{{\cal C}_{1}}\dots\cons{{\cal C}_{n}}
\]
by filling the holes in $\Gamma\context\dots\context$ with ${\cal
  C}_{1},\dots,{\cal C}_{n}$. Formally, to define this we first need
an auxiliary definition adjusting indices of holes. Given a precontext
$\C$, let $\C^{\mathord{+}j}$ be the precontext obtained from it by
replacing each hole $\context_i$ by $\context_{i+j}$. Given a
precontext $\C$ and contexts $\C_1,\dots,\C_n$ we now inductively
define $\cons{{\C}_{1}}\dots\cons{{\C}_{n}}$ as follows, where $a_j$
is the arity of $\C_j$:
\begin{enumerate}
\item if $\,\C=\context_i\,$ then
  $\,\C\cons{{\C}_{1}}\dots\cons{{\C}_{n}}={\C}_{i}^{+\sum_{j<i}a_j}\,$,
\item if $\,\C=\C',\C''\,$ then $\,\C\cons{{\C}_{1}}\dots\cons{{\C}_{n}}=
  \C'\cons{{\C}_{1}}\dots\cons{{\C}_{n}},
  \C''\cons{{\C}_{1}}\dots\cons{{\C}_{n}}\,$
  and
\item   if $\,\C=[\C']\,$ then $\,\C\cons{{\C}_{1}}\dots\cons{{\C}_{n}}=
  [\C'\cons{{\C}_{1}}\dots\cons{{\C}_{n}}]$\quad.
\end{enumerate}
Clearly, $\C\cons{{\C}_{1}}\dots\cons{{\C}_{n}}$ is a context if
$\C$ and $\C_1,\dots,\C_n$ are contexts. We leave out replacements of
holes by holes, so by convention we write $\Gamma\cons{\C_1}\context$
instead of $\Gamma\cons{\C_1}\cons{\C_2}$ if ${\C}_2$ is a hole.

\begin{example}\label{excon}
  Given the binary context
  $\Gamma\context\context=A,[[B],\context],\context$ and the unary
  context $\Delta\context=C,[\context]$ we can obtain the binary
  context
  \[
  \Gamma\cons{\Delta\context}\context=A,[[B],C,[\context]],\context
  \quad,
  \]
  where we omitted the indices of holes since in all contexts the
  holes are ordered from left to right as shown.
\end{example}


{\bf Inference rules, derivations and proofs.} In the following
instance of an inference rule $\rho$
\[
\cinf{\rho}{\Delta}{\Gamma_1 \quad
  \dots \quad \Gamma_n}
\]
we call $\Gamma_1\dots\Gamma_n$ its \emph{premises} and $\Delta$ its
\emph{conclusion}.  We write $\rho^n$ to denote $n$ instances of
$\rho$ and $\rho^*$ to denote an unspecified number of instances of
$\rho$. A \emph{system}, denoted by $\Sys$, is a set of inference
rules.  A \emph{derivation} in a system $\Sys$ is a finite tree whose
nodes are labelled with sequents and which is built according to the
inference rules from $\Sys$. The sequent at the root is the
\emph{conclusion} and the sequents at the leaves are the
\emph{premises} of the derivation. Derivations are denoted by $\D$. A
derivation $\D$ with conclusion $\Gamma$ in system $\Sys$ is sometimes
shown as
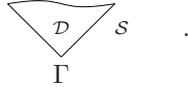
\[
\vlderivation{
 \vltrnote{\D}{\Gamma}{
\vlhy {}}{\vlhy {\qquad}}{\vlhy {}}{\Sys}}
\qquad .
\]
The depth of a derivation $\D$ is denoted by $|\D|$.  Note that the
depth of a derivation, which is a tree, has nothing to do with the
depth of the sequents in it, which are also trees. A \emph{proof} of a
sequent $\Gamma$ in a system is a derivation in this system with
conclusion $\Gamma$  where all premises are instances of the axiom
$\Gamma\cons{p,\neg p}$.  Proofs are denoted by $\PP$.

{\bf The sequent systems.} Figure~\ref{fig:kx} shows the set of rules
from which we form our deductive systems. \emph{System {\K}} is the
set of rules $\{\vlan,\vlor,\Box,{\kc}\}$. We will look at extensions
of System $\K$ with any combination of the rules $\diacrules$.  Each
rule name $\diac{\rho}$ in $\X$ has a corresponding frame condition
and modal Hilbert-style axiom $\rho$ as shown in
Figure~\ref{fig:framecond}. The subscript {\sf c} denotes that a rule
has a built-in contraction. We also consider rules without built-in
contraction. They have the same names but without the subscript and
are shown in Figure~\ref{fig:dia}. The purpose of the built-in
contraction is to make all rules invertible and to make contraction
admissible. Given a set of names of modal axioms $\X
\subseteq \{{\sf d,t,b,4,5}\}$, $\dia{\X}$ is the set of rule names
$\{\dia{\rho}\,|\, \rho \in \X \}$, and $\diac{\X}$ is the set of rule names
$\{\diac{\rho}\,|\, \rho \in \X \}$.

The $\diafivec$-rule is a bit special since it uses a
binary context. It can actually be decomposed into three rules that
use unary contexts, as we will see.  However, we prefer the
presentation as a single rule.  The rule is best understood as
allowing to do the following: when going from premise to conclusion,
take some formula $\Diamond A$, which is not at the root, and copy it
to any other place in the sequent.

\begin{figure}
\fbox{\parbox{\ftextwidth}{
  \centering
  \[ \Gamma\cons{p,\neg p} \qquad \cinf{\vlan}{\Gamma\cons{A \vlan
      B}}{\Gamma\cons{A} \quad \Gamma\cons{B}} \qquad
  \cinf{\vlor}{\Gamma\cons{A \vlor B}}{\Gamma\cons{A,B}}
\]
\[
  \kaiinf{\Box}{\Gamma\cons{\Box A}}{\Gamma\cons{[A]}} \qquad
\kaiinf{\kc}{\Gamma\cons{\Diamond A, [\Delta]}}{\Gamma\cons{\Diamond
    A, [\Delta, A]}}
\]

\[
\kaiinf{\diadc}{\Gamma\cons{\Diamond
    A}}{\Gamma\cons{\Diamond A,[A]}}  \qquad
\kaiinf{\diatc}{\Gamma\cons{\Diamond
    A}}{\Gamma\cons{\Diamond A,A}} \qquad
\kaiinf{\diabc}{\Gamma\cons{[\Delta, \Diamond A]}}{\Gamma\cons{[\Delta,
    \Diamond A], A}}
\]

\[
\kaiinf{\diafourc}{\Gamma\cons{\Diamond A,
    [\Delta]}}{\Gamma\cons{\Diamond A, [\Delta, \Diamond A]}} \qquad\quad
\kaiinf{\diafivec}{\Gamma\cons{\Diamond
    A}\cons{\emptyset}}{\Gamma\cons{\Diamond A}\cons{\Diamond
A}}{\quad \mbox{$\depth{\Gamma\context\cons{\emptyset}}>0$}}
\]
}}
  \caption{System \K+$\{\diacrules\}$ }
  \label{fig:kx}
\end{figure}

\begin{figure}
  \centering
\[
\cinf{\nec}{[\Gamma]}{\Gamma} \qquad\quad
\cinf{\wk}{\Gamma\cons{\Delta}}{\Gamma\cons{\emptyset}} \quad\qquad
\cinf{\ctr}{\Gamma\cons{\Delta}}{\Gamma\cons{\Delta,\Delta}}\quad\qquad
  \cinf{{\cut}}{\Gamma\cons{\emptyset}}{\Gamma\cons{A} \qquad
    \Gamma\cons{\neg A}}
\]
  \caption{Necessitation, weakening, contraction and cut}
  \label{fig:admissible}
\end{figure}

\begin{example}
  Here is an example of a proof in system $\K$, namely of some
instance of the $\sf k$-axiom:
\[
\vlderivation{
\vlin{=}{}{\Box(a \vlor b) \imp (\Box a \vlor \Diamond b)}{
\vlin{\vlor^2}{}{\Diamond(\neg a \vlan \neg b) \vlor (\Box a \vlor \Diamond b) }{
\vlin{\Box}{}{\Diamond(\neg a \vlan \neg b) , \Box a , \Diamond b}{
\vlin{\kc}{}{\Diamond(\neg a \vlan \neg b) , [a] , \Diamond b}{
\vliin{\vlan}{}
{\Diamond(\neg a \vlan \neg b) , [a, \neg a \vlan \neg b] , \Diamond
b}
{\vlhy{\Diamond(\neg a \vlan \neg b) , [a, \neg a] , \Diamond
b}}
{\vlin{\kc}{}{\Diamond(\neg a \vlan \neg b) , [a, \neg b] , \Diamond b}{
\vlhy{\Diamond(\neg a \vlan \neg b) , [a, \neg b, b] , \Diamond b}}}}}}}}
\]
\end{example}

{\bf Admissibility, derivability and invertibility.} We write $\Sys
\vdash \Gamma$ if there is a proof of $\Gamma$ in system $\Sys$. An
inference rule $\rho$ is \emph{(depth-preserving) admissible} for a
system $\Sys$ if for each proof in $\Sys \cup \{\rho\}$ there is a
proof of in $\Sys$ with the same conclusion (and with at most the same
depth).  An inference rule $\rho$ is \emph{derivable} for a system $\Sys$ if
for each instance of $\rho$ there is a derivation $\D$ in $\Sys$ with
the same conclusion and such that each premise of $\D$ is a premise of
the given instance of $\rho$.

For each rule $\rho$ there is its
\emph{inverse}, denoted by $\inverse{\rho}$, which is obtained by
exchanging premise and conclusion. The $\inverse{\vlan}$-rule allows
both $\Gamma\cons{A}$ and $\Gamma\cons{B}$ as conclusions of
$\Gamma\cons{A\vlan B}$. An inference rule $\rho$ is
\emph{(depth-preserving) invertible} for a system $\Sys$ if
$\inverse{\rho}$ is (depth-preserving) admissible for $\Sys$.

The rules shown in Figure~\ref{fig:admissible} turn out to be
admissible. We will now show this for the first three rules, for the
cut rule it will be shown later.

\begin{lemma}[Admissibility of structural rules and invertibility]\label{l:stradm} 
For each system $\K+\diac{\X}$ with $\X \subseteq \modax$ the following hold:\\
(i)  The rules necessitation, weakening and contraction are depth-preserving admissible.\\
(ii)  All its rules are depth-preserving invertible.\\
\end{lemma}

\begin{proof}
  The admissibility of necessitation and weakening follows from a
  routine induction on the depth of the proof. The same works for the
  invertibility of the $\vlan, \vlor$ and $\Box$-rules in (ii). The
  inverses of all other rules are just weakenings. For the
  admissibility of contraction we also proceed by induction on the
  depth of the proof tree, using depth-preserving invertibility of the
  rules. The cases are easy for the propositional rules and for the
  $\Box,{\diadc,\diatc}$-rules. For the {\kc}-rule we consider the formula
  $\Diamond A$ from its conclusion $\Gamma\cons{\Diamond A, [\Delta]}$
  and its position inside the premise of contraction
  $\Lambda\cons{\Sigma,\Sigma}$. We have the cases 1) $\Diamond A$ is
  inside $\Sigma$ or 2) $\Diamond A$ is inside $\Lambda\context$. We
  have three subcases for case 1: 1.1) $[\Delta]$ inside
  $\Lambda\context$, 1.2) $[\Delta]$ inside $\Sigma$, 1.3)
  $\Sigma,\Sigma$ inside $[\Delta]$. There are two subcases of case 2:
  2.1) $[\Delta]$ inside $\Lambda\context$ and 2.2) $[\Delta]$ inside
  $\Sigma$.  All cases are either simpler than or similar to case 1.2,
  which is as follows:
\[
\vlderivation{
\vlin{{\ctr}}{}{\Lambda'\cons{\Diamond A,\Sigma', [\Delta]}}{ 
\vlin{{\kc}}{}{\Lambda'\cons{\Diamond A, \Sigma', [\Delta],\Sigma',
        [\Delta]}}{
\vlhy{ \Lambda'\cons{\Diamond A,\Sigma', [\Delta,A], \Sigma',
[\Delta]} }}}}
 \quad \leadsto \quad 
\vlderivation{
   \vlin{{\kc}}{}{\Lambda'\cons{\Diamond A,
           \Sigma', [\Delta]}}{ 
\vlin{\ctr}{}{\Lambda'\cons{\Diamond A,
           \Sigma', [\Delta,A]}}{
       \vlin{\inverse{{\kc}}}{}{\Lambda'\cons{\Diamond A,
           \Sigma', [\Delta,A], \Sigma', [\Delta,A]}}{
\vlhy{ \Lambda'\cons{\Diamond A,\Sigma', [\Delta,A], \Sigma',
[\Delta]} }}}} }
 \quad ,
\]
where the instance of $\inverse{{\kc}}$ in the proof on the right is
removed because it is depth-preserving admissible and the instance of
contraction is removed by the induction hypothesis. The case for the
{\diafourc}-rule works the same way.

For the {\diabc}-rule we make a case analysis based on the position of
$[\Delta,\Diamond A]$ from its conclusion
$\Gamma\cons{[\Delta,\Diamond A]}$ inside the premise of contraction
$\Lambda\cons{\Sigma,\Sigma}$. We have three cases: 1)
$[\Delta,\Diamond A]$ inside $\Lambda\context$, 2) $[\Delta,\Diamond
A]$ in $\Sigma$ and 3) $\Sigma,\Sigma$ inside $[\Delta,\Diamond A]$.
Case 3 has two subcases: either $\Diamond A \in \Sigma$ or not. All
cases are trivial except for case 2 where invertibility of the $\diabc$-rule is used.

For the {\diafivec} rule we make a case analysis based on the positions of
the sequent occurrences $\Diamond A$ and $\emptyset$ from its
conclusion $\Gamma\cons{\Diamond A}\cons{\emptyset}$ inside the
premise of contraction $\Lambda\cons{\Sigma,\Sigma}$. We have two
cases: 1) $\emptyset$ inside $\Lambda\context$, 2) $\emptyset$ inside
$\Sigma$. The first case is trivial, in the second we have two
subcases: 1) $\Diamond A$ inside $\Lambda\context$ and 2) $\Diamond A$
inside $\Sigma$. Case 2.1 is similar to case 2.2 which is as follows:
\[
\vlderivation{\vlin{{\ctr}}{}{\Lambda\cons{\Sigma\cons{\Diamond
        A}\cons{\emptyset}}}{ \vlin{{\diafivec}}{}
    {\Lambda\cons{\Sigma\cons{\Diamond
          A}\cons{\emptyset},\Sigma\cons{\Diamond
          A}\cons{\emptyset}}}{ \vlhy{\Lambda\cons{\Sigma\cons{\Diamond
            A}\cons{\emptyset},\Sigma\cons{\Diamond A}\cons{\Diamond
            A}} }}}} 
\quad \leadsto \quad 
\vlderivation{
\vlin{{\diafivec}}{}{\Lambda\cons{\Sigma\cons{\Diamond A}\cons{\emptyset}} }{
    \vlin{\ctr}{}{ \Lambda\cons{\Sigma\cons{\Diamond A}\cons{\Diamond
          A}} }{ \vlin{\inverse{{\diafivec}}}{}{\Lambda\cons{\Sigma\cons{\Diamond A}\cons{\Diamond
            A}, \Sigma\cons{\Diamond A}\cons{\Diamond A}}
      }{\vlhy{\Lambda\cons{\Sigma\cons{\Diamond A}\cons{\emptyset},
            \Sigma\cons{\Diamond A}\cons{\Diamond A}} }}}} } \quad .
\]

\end{proof}

\begin{figure}
  \centering
\fbox{
\parbox{\ftextwidth}{
\medskip
\[
\cinf{{\diak}}{\Gamma\cons{\Diamond A, [\Delta]}}{\Gamma\cons{[A,
\Delta]}} \qquad 
\cinf{{\diad}}{\Gamma\cons{\Diamond A}}{\Gamma\cons{[A]}} 
\qquad 
\cinf{{\diat}}{\Gamma\cons{\Diamond      A}}{\Gamma\cons{A}}
\qquad 
\cinf{{\diab}}{\Gamma\cons{[\Delta, \Diamond A]}}{\Gamma\cons{[\Delta], A}}
\]

\[
\cinf{{\diafour}}{\Gamma\cons{\Diamond A,
    [\Delta]}}{\Gamma\cons{[\Delta, \Diamond A]}}
\qquad\quad \cinf{{\diafive}}{\Gamma\cons{\Diamond
    A}\cons{\emptyset}}{\Gamma\cons{\emptyset}\cons{\Diamond
    A}}{\quad \mbox{$\depth{\Gamma\context\cons{\emptyset}}>0$}}
\]
}}
\caption{Diamond rules without built-in contraction}
\label{fig:dia}
\end{figure}

By using weakening admissibility, we easily get the following proposition.

\begin{proposition}[Relation between the $\dia{}$-rules and the $\diac{}$-rules]\nl
  For each $\rho \in \{{\sf k,d,t,b,4,5}\}$ we have that\\
  (i) the rule $\dia{\rho}$ is derivable for $\{\diac{\rho},\wk\}$ and
 admissible for system $\K+\diac{\rho}$,\\
  (ii) the rule $\diac{\rho}$ is derivable for $\{\dia{\rho},\ctr\}$.
\end{proposition}

\subsection{Soundness}

To prove soundness, we first need some standard definitions for Kripke
semantics.

\begin{definition}[frames, models, validity]
  A \emph{frame} is a pair $(S,\rel)$ of a nonempty set $S$ of
  \emph{states} and a binary relation $\rel$ on it.  A \emph{model}
  $\M$ is a triple $(S,\rel,V)$ where $(S,\rel)$ is a frame and $V$ is
  a a mapping which assigns a subset of $S$ to each proposition, and
  which is called \emph{valuation}. A model $\M$ as given above
  induces a relation $\models$ between states and formulas which is
  defined as usual. In particular we have $s \models p$ iff $s \in
  V(p)$, $s \models \neg p$ iff $s \not\in V(p)$, $s \models A \vlor B$
  iff $s \models A$ or $s \models B$, $s \models A \vlan B$ iff $s
  \models A$ and $s \models B$, $s \models \Diamond A$ iff there is a
  state $t$ such that $s \rel t$ and $t \models A$, and $s \models
  \Box A$ iff for all $t$ if $s \rel t$ then $t \models A$. Further, a
  formula $A$ is \emph{valid in a model} $\M$, denoted $\M \models A$,
  if for all states s of $\M$ we have $s \models A$. A formula $A$ is
  \emph{valid in a frame} $(S,\rel)$, denoted $(S,\rel) \models A$, if
  for all valuations $V$ we have $(S,\rel,V) \models A$. A formula is
  \emph{valid} if it is valid in all frames. For a set of $\X$ of rule
  names or names of modal axioms we call a frame an \emph{$\X$-frame}
  if it satisfies all the frame conditions corresponding to the names
  in $\X$.  A formula is \emph{{\X}-valid} if it is valid in all
  $\X$-frames.
\end{definition}

The {\diafivec}-rule requires some care when proving its soundness
because it is defined in terms of a binary context. We first show
how it is derivable for three rules which, modulo built-in
contraction, are special cases of the {\diafivec}-rule. The soundness of
these rules is then easy to establish.

\begin{lemma}[Decompose \diafivec]\label{l:split5}
  The $\diafivec$-rule is derivable for $\{
  \diafivepa,\diafivepb,\diafivepc,\ctr\}$, where
  {\diafivepa, \diafivepb, \diafivepc} are the rules
\[
\kaiinf{{\diafivepa}}{\Gamma\cons{[\Delta, \Diamond
    A]}}{\Gamma\cons{[\Delta],\Diamond A}} \quad,\quad 
\kaiinf{{\diafivepb}}{\Gamma\cons{[\Delta, \Diamond
    A],[\Lambda]}}{\Gamma\cons{[\Delta],[\Lambda,\Diamond A]}}
\quad,\quad 
\kaiinf{{\diafivepc}}{\Gamma\cons{[\Delta, \Diamond A,
    [\Lambda]]}}{\Gamma\cons{[\Delta, [\Lambda, \Diamond A]]}} \quad .
\]
\end{lemma}
\begin{proof}
  Seen bottom-up, the $\diafivec$-rule allows to put a formula $\Diamond
  A$ which occurs at a node different from the root into an arbitrary
  node. We can use contraction to duplicate $\Diamond A$ and move one
  copy to the root and also to some child of the root by $\diafivepa$. By
  $\diafivepb$ we can move it to any child of the root and by $\diafivepc$
  into any descendant of a child of the root.
\end{proof}

\begin{lemma}[Deep inference is sound]\label{l:deepsound}
  Let $\X \subseteq \modax$, $\Gamma\context$ be a
  context and $A,B$ be formulas. If the formula $A \supset B$ is
  {\X}-valid then $\Gamma\cons{A} \supset \Gamma\cons{B}$ is
  {\X}-valid.
\end{lemma}
\begin{proof}
  By induction on the depth of $\Gamma\context$. We use the soundness
  of some Hilbert-style axiomatisation of {\K+\X}. To show the
  validity of
\[
(\Gamma_1,[\Gamma_2\cons{A}]) \supset (\Gamma_1,[\Gamma_2\cons{B}])
\]
  we use the induction hypothesis to get $\Gamma_2\cons{A} \supset
  \Gamma_2\cons{B}$, necessitation to get $\Box(\Gamma_2\cons{A} \supset
  \Gamma_2\cons{B})$, the {\sf k}-axiom to get $\Box(\Gamma_2\cons{A}) \supset
  \Box(\Gamma_2\cons{B})$, and finally propositional reasoning to get $\Gamma_1,[\Gamma_2\cons{A}] \supset \Gamma_1,[\Gamma_2\cons{B}]$.

\end{proof}

\begin{theorem}[Soundness] Let $\Gamma, \Delta$ and
  $\Gamma_1,\dots,\Gamma_n$ be sequents. Then the following hold:\\
  (i) For any rule $\rho \in \K$ if $\,\cinf{\rho}{\Delta}{\Gamma_1 \;
    \dots \; \Gamma_n}\,$ then $\,\Gamma_1 \vlan \dots \vlan
  \Gamma_n \supset \Delta$
  is valid. \\
  (ii) For any rule $\rho \in \modax$ 
if  $\,\cinf{\diac{\rho}}{\Delta}{\Gamma}\,$ then 
$\Gamma \supset \Delta$ is $\{\rho\}$-valid.  \\
  (iii) For any $\X \subseteq \modax$ if $\,\K+\diac{\X} \vdash
  \Gamma$ then $\Gamma$ is $\X$-valid.
\end{theorem}

\begin{proof}
  The axiom is valid in all frames which follows from an induction on
  the depth of $\Gamma\context$ where necessitation is used in the
  induction step.  Thus (i) and (ii) imply (iii). Most cases of (i)
  are trivial, for the $\vlan$-rule it follows from an induction on
  the context and uses the implication $\Box A \vlan \Box B \supset
  \Box(A \vlan B)$.  Lemma~\ref{l:deepsound}~(Deep inference is
  sound) used together with the $\sf k$-axiom yields that the premise
  of the $\kc$-rule implies its conclusion.  The cases from (ii) for
  the $\{{\diadc,\diatc,\diabc,\diafourc}\}$-rules are similar to the
  $\kc$-rule, using the corresponding modal axiom.

  For the soundness of the {\diafivec}-rule we use
  Lemma~\ref{l:split5}~(Decompose \diafivec) and show soundness of the
  rules $\diafivepa,\diafivepb,\diafivepc$.  For $\diafivepc$ we show
  that a euclidean countermodel for the conclusion is also a
  countermodel for the premise, the other cases are similar. A
  countermodel for $[\Delta,\Diamond A,[\Lambda]]$ has to contain
  states $s \rel t \rel u$ such that $t \not\models \Delta$, $u
  \not\models \Lambda$ and $v \not\models A$ for any $v$ with $t \rel
  v$. We need to show that for any $w$ with $u \rel w$ we have $w
  \not\models A$. By euclideanness we obtain, in this order: $t \rel
  t, u \rel t, t \rel w$. Thus $w \not\models A$.

\end{proof}

\subsection{Completeness}

The current set of modal rules does not allow for a modular
completeness result of the form ``\emph{if $\Gamma$ is $\X$-valid then
  $\K+\diac{\X} \vdash \Gamma$}''. It is easy to check
that some of our systems are incomplete.

\begin{fact}[Incompleteness]\label{f:45needed} For any propositional
  variable $p$ we have that the formula $\Box p\imp \Box\Box p$ holds
  in any $\{{\sf t,5} \}$-frame and the formula
  $\Diamond p\imp\Box\Diamond p$ holds in any $\{{\sf b,4} \}$-frame, but:\\
  \begin{tabular}{ll}
    (i)   & $\K+\{\diatc,\diafivec\} \; \nvdash \; \Box p\supset
    \Box\Box p$ \quad \mbox{and}\\
    (ii) & $\K+\{\diabc,\diafourc\} \; \nvdash \; \Diamond p\supset
    \Box\Diamond p$
    \quad .
  \end{tabular}
\end{fact}

However, while not every combination of modal rules is sound and
complete for the respective set of frames, we can define a condition
on rule combinations which ensures that they are complete.

\begin{definition}[45-closed]
  Let $\X\subseteq\{{\sf d,t,b,4,5}\}$. The set $\X$ is
  \emph{45-closed} if for $\rho \in \{{\sf 4,5}\}$ we have that if all
  $\X$-frames satisfy $\rho$ then $\rho \in \X$.  Both of the sets
  $\{{\sf t,5}\}$ and $\{{\sf b,4}\}$ are not 45-closed, for example,
  while both $\{{\sf t,4,5}\}$ and $\{{\sf b,4,5}\}$ are. A set of
  modal rules is \emph{45-closed} if its underlying set of names of
  modal axioms is 45-closed.
\end{definition}

The completeness result we are about to prove holds for 45-closed
$\X$. It is easy to check that for each set of frames which can be
characterised by our five axioms there is a combination of modal rules
which is 45-closed and thus is also sound and complete.  In order to
prove our completeness result, we first need some preliminary
definitions which will help us to extract a tree-like Kripke model
from a sequent.

\begin{definition}[subtree of a sequent]
  A sequent $\Delta$ is an \emph{immediate subtree} of
  a sequent $\Gamma$ if there is a sequent $\Lambda$ such that $\Gamma
  = \Lambda, [\Delta]$. It is a \emph{proper subtree} if it is an
  immediate subtree either of $\Gamma$ or of a proper subtree of
  $\Gamma$, and it is a \emph{subtree} if it is either a proper
  subtree of $\Gamma$ or $\Delta=\Gamma$. The set of all subtrees of
  $\Gamma$ is denoted by $\st(\Gamma)$. A formula $A$ is \emph{in} a
  sequent $\Gamma$ if $A \in \Gamma$ and it is \emph{inside} $\Gamma$
  if there is a subtree $\Delta$ of $\Gamma$ such that $A \in \Delta$.
\end{definition}

Our sequents are based on multisets. We need a way to stop proof
search once their underlying sets remain the same, so we need the
following notion:

\begin{definition}[set sequent]
  The \emph{set sequent} of the sequent 
\[
  A_1,\dots,A_m,[\Delta_1],\dots,[\Delta_n] 
\]
is the underlying set of
\[
  A_1,\dots,A_m,[\Lambda_1],\dots,[\Lambda_n] \quad ,
\]
  where $\Lambda_1\dots\Lambda_n$ are the set sequents of
  $\Delta_1\dots\Delta_n$. Clearly the set sequent of a given sequent
  is again a sequent since a set is a multiset.
\end{definition}

We will not directly prove completeness of the systems $\K+\diac{\X}$,
but of different, equivalent systems $(\K+\diac{\X})^\circ$ that we define
now. For each rule $\rho$ we define a rule $\rho^\circ$ which keeps
the main formula from the conclusion. For most rules $\rho=\rho^\circ$
except for the following rules:
\[
\cinf{\vlan^\circ}{\Gamma\cons{A \vlan B}}{\Gamma\cons{A\vlan
      B, A} \quad \Gamma\cons{A \vlan B, B}}
\qquad
\kaiinf{\vlor^\circ}{\Gamma\cons{A \vlor B}}{\Gamma\cons{A\vlor B,A,B}}
\]

\[
\kaiinf{\Box^\circ}{\Gamma\cons{\Box A}}{\Gamma\cons{\Box A, [A]}}
\quad\parbox{8.5cm}{where in the conclusion the node of the active
  formula does not have a child node which contains $A$}
\]

\[
\kaiinf{{\diadccirc}}{\Gamma\cons{\Diamond A}}{\Gamma\cons{\Diamond A,
    [A]}} \quad\parbox{8.5cm}{where in the conclusion the node of the
  active formula does not have a child node.}
\]
In addition, each rule $\rho^\circ$ carries the proviso that for all
of its premises the set sequent is different from the set sequent of
the conclusion.  Given a system $\Sys$ the system $\Sys^\circ$ is
obtained by replacing each rule $\rho \in \Sys$ by $\rho^\circ$.
Systems $\Sys$ and $\Sys^\circ$ will turn out to be equivalent, as we
will know after the completeness theorem. For now we just prove one
direction of the equivalence.

\begin{lemma}[$\Sys^\circ$ into $\Sys$]\label{l:equiv}
  For all $\X \subseteq \modax$ and for all sequents
  $\Gamma$ we have that $(\K+\diac{\X})^\circ \vdash \Gamma$
  \;implies\; $\K+\diac{\X} \vdash \Gamma$.
\end{lemma}
\begin{proof}
  By a standard induction on the proof tree, using contraction and
  weakening admissibility for $\K+\diac{\X}$.

\end{proof}

In order to prove completeness we need some closures of relations.

\begin{definition}[some closures of relations]
  Let $\rightarrow$ be a binary relation on a set $S$. Then
  $\leftarrow$ denotes its inverse, $\leftrightarrow$ its symmetric
  closure, $\rightarrow^+$ its transitive closure and $\rightarrow^*$
  its reflexive-transitive closure. For $\X \subseteq \modaxnod$
  $\rightarrow^\X$ denotes the smallest relation that includes
  $\rightarrow$ and has the properties in $\X$. The same conventions
  are used for different arrows that denote relations, such as
  $\Rightarrow$, the inverse of which is $\Leftarrow$, and so on.
\end{definition}

We will see shortly that $\rightarrow^\X$ is well-defined. First we
need to characterise the euclidean and the transitive-euclidean
closure of a relation.

\begin{definition}[(transitive-)euclidean connection]
  Let $\rightarrow$ be a binary relation on a set $S$ and let $s,t \in
  S$. A \emph{euclidean connection} for $\rightarrow$ from $s$ to $t$
  is a nonempty sequence $s_1\dots s_n$ of elements of $S$ such that
  we have
\[
s \leftarrow s_1 \leftrightarrow s_2 \leftrightarrow \dots
\leftrightarrow s_n \rightarrow t \qquad .
\]
A \emph{transitive-euclidean connection} is defined likewise but such that
\[
s = s_1 \leftrightarrow s_2 \leftrightarrow \dots
\leftrightarrow s_n \rightarrow t \qquad .
\]
We write $s \rightarrow_{(4)5} t$ if there is a (transitive-)euclidean
connection for $\rel$ from $s$ to $t$.
\end{definition}

\begin{lemma}[$\rightarrow^\X$ is well-defined]\label{l:45}
  Let $\rightarrow$ be a binary relation on a set $S$.
  Then the following hold:\\
  (i) For all $\X \subseteq \modaxnod$ the relation $\rightarrow^\X$ is well-defined.\\
  (ii) The relation $\rightarrow \cup  \rightarrow_5$  is the least euclidean relation that contains $\rightarrow$. \\
  (iii) The relation $\rightarrow_{45}$ is the least transitive and
  euclidean relation that contains $\rightarrow$.
\end{lemma}

\begin{proof}
  (i) is easy to check except for the cases for $\{{\sf 5}\}$ and
  $\{{\sf 4,5}\}$, which follow from (ii) and (iii).

  (ii) Euclideanness is easy to check. For leastness we show that any
  euclidean relation $\Rightarrow$ that includes $\rel$ also includes
  $\rel_5$. If $s \rel_5 t$ then $s{\Rightarrow}_5 t$. We show
  $s{\Rightarrow}_5 t$ for a euclidean connection of length $n$
  implies $s {\Rightarrow} t$ by induction on $n$. Assume there is an
  $s_i$ in the euclidean connection such that $s_{i-1} {\Rightarrow} s_i
  {\Leftarrow} s_{i+1}$.  Then we have two smaller euclidean
  connections to which we apply the induction hypothesis and obtain $s
  {\Rightarrow} t$ by euclideanness.  If there is no such $s_i$ then the
  euclidean connection looks as follows:
\[
s = s_0 {\Leftarrow} s_1 {\Leftarrow} \dots {\Leftarrow} s_j
{\Rightarrow} \dots {\Rightarrow} s_n {\Rightarrow} s_{n+1}=t
\rlap{\quad ,}
\]
and by euclideanness we have $s_{j-1} {\Rightarrow} s_{j+1}$ and thus
removing $s_j$ yields a smaller euclidean connection from $s$ to $t$
which by induction hypothesis implies $s {\Rightarrow} t$.

(iii) Euclideanness and transitivity are easy to check. For leastness
we show that any transitive-euclidean relation ${\Rightarrow}$ that
includes $\rel$ also includes $\rel_{45}$. If $s\rel_{45} t$ then
$s{\Rightarrow}_{45} t$. If there is no $s_i$ in the
transitive-euclidean such that $s_i {\Leftarrow} s_{i+1}$, then $s
{\Rightarrow} t$ follows by transitivity.  Otherwise, choose the first
such $s_i$.  We have a euclidean connection from $s_i$ to $t$, thus
similarly to (ii) obtain $s_i {\Rightarrow} t$ and by transitivity $s
{\Rightarrow} s_i$ and $s {\Rightarrow} t$.

\end{proof}

\begin{definition}[serial closure]
Let $\rightarrow$ be a binary relation on a set $S$. Its
\emph{serial closure}, denoted $\rightarrow^{\sf d}$, is obtained
from $\rightarrow$ by adding $s \rightarrow s$ for each $s \in S$ which
violates seriality. For $\X \subseteq \{{\sf t,b,4,5}\}$ the relation
$\rightarrow^{\sf X \cup \{{d}\}}$ is defined as $(\rightarrow^{\sf
X})^{\sf d}$.
\end{definition}

\begin{lemma}[Serial closure preserves frame conditions]\label{l:d}
  Let $\rightarrow$ be a binary relation on a set $S$. If
  $\rightarrow$ satisfies a frame condition in $\{{\sf t,b,4,5}\}$
  then $\rightarrow^{\sf d}$ also satisfies that frame condition.
\end{lemma}
\begin{proof}
  For reflexivity this is clear since a reflexive relation is its own
  serial closure. For symmetry this is clear since only loops are
  added, which are their own inverses. For transitivity, assume that
  we have $s\rightarrow^{\sf d} t$ and $t \rightarrow^{\sf d} u$. If
  either $s=t$ or $t=u$ then we have $s\rightarrow^{\sf d} u$. So
  assume $s\neq t$ and $t\neq u$. Then $s\rightarrow t$ and $t
  \rightarrow u$ and by transitivity of $\rightarrow$ we get $s
  \rightarrow u$ and thus $s  \rightarrow^{\sf d} u$.

  For euclideanness, assume that $s \rightarrow^{\sf d} t$ and $s
  \rightarrow^{\sf d} u$. We need to show that $t \rightarrow^{\sf d}
  u$.  If $s=t$ then we are done, so assume $s\neq t$ which implies $s
  \rightarrow t$. Since $s\rightarrow^{\sf d} u$ and since $s$ does
  not violate seriality we have $s\rightarrow u$. By euclideanness of
  $\rightarrow$ we obtain $t \rightarrow u$ and thus $t
  \rightarrow^{\sf d} u$.  
\end{proof}

\begin{definition}[cyclic, finished, $\prove{\Gamma,\X}$]
  A leaf of a sequent is \emph{cyclic} if there is an inner node in
  the sequent that carries the same set of formulas. A node in a
  sequent is \emph{finished} for a system $\Sys$ if no rule from
  $\Sys$ applies to a formula in this node. A sequent is
  \emph{finished} for a system $\Sys$ if all its nodes are either
  finished for $\Sys$ or cyclic. We define a procedure
  $\prove{\Gamma,\X}$, which takes a sequent $\Gamma$ and a set $\X
  \subseteq \modax$ and builds a derivation tree for
  $\Gamma$ by applying rules from $(\K+\diac{\X})^\circ$ to non-axiomatic and
  unfinished derivation leaves in a bottom-up fashion. It is shown in
  Figure~\ref{fig:prove}.  If $\prove{\Gamma,\X}$ terminates and all
  derivation leaves are axiomatic then it \emph{succeeds} and if it
  terminates and there is a non-axiomatic derivation leaf then it
  \emph{fails}.
\end{definition}

  \begin{figure}
    \centering
\fbox{
\parbox{\ftextwidth}{
\begin{description}
    \item[{\bf Repeat}]
    \item[\quad (step 1)] Apply the rules in $((\K+\diac{\X})
      \setminus \{\Box,\diadc\})^\circ$ as long as possible.
    \item[\quad (step 2)] Wherever possible, apply the rules in
      $(\{\Box\} \cup (\diac{\X} \cap \{\diadc\}))^\circ$ once.
    \item[{{\bf Until} each non-axiomatic derivation leaf is
        finished.}]
    \end{description}
}}
    \caption{The algorithm $\prove{\Gamma,\X}$}
\label{fig:prove}
\end{figure}

\begin{definition}[size of a sequent, $\sufo{\Gamma}$]
  The \emph{size} of a sequent is the number of nodes of its
  corresponding tree. The set of subformulas of a sequent $\Gamma$,
  denoted $\sufo{\Gamma}$ is the set of all subformulas of all
  formulas which are element of some node of the sequent.  
\end{definition}

\begin{lemma}[Termination]\label{l:terminates}
  For all sets $\X \subseteq \{{\sf d,t,b,4,5}\}$ and for all sequents
  $\Gamma$ the procedure $\prove{\Gamma,\X}$ terminates after at most
  $2^{|\sufo{\Gamma}|}$ iterations (of the repeat-until-loop).  
\end{lemma}
\begin{proof}
  Consider a sequence of sequents along a given branch of the
  derivation starting from the root.  A rule application in step 1
  does not create new nodes in the sequent and causes the set of
  formulas at some node in the sequent to strictly grow. By the
  subformula property only finitely many formulas can occur in a node,
  so step 1 terminates. If after step 1 there is an unfinished leaf in
  a sequent then the size of the sequent strictly grows in step 2.
  Since there are only $2^{|\sufo{\Gamma}|}$ different sets of
  formulas that can occur each unfinished sequent leaf has to be
  cyclic before $2^{|\sufo{\Gamma}|}$ iterations. Then the sequent
  will be finished if it is not axiomatic, and thus the algorithm
  terminates.

\end{proof}

\begin{theorem}[Completeness]
\label{t:completeness}
For all 45-closed sets $\X \subseteq\modax$ and for all sequents
$\Gamma$ the following hold:\\
(i) If $\Gamma$ is $\X$-valid then  $\K+\diac{\X} \vdash \Gamma$.\\
(ii) If $\prove{\Gamma,\X}$ fails then there is a finite $\X$-frame
in which $\Gamma$ is not valid.
\end{theorem}
\begin{proof}
  The contrapositive of (i) follows from (ii): if $\K+\diac{\X}
  \nvdash \Gamma$ then by Lemma~\ref{l:equiv}~($\Sys^\circ$ into
  $\Sys$) also $(\K+\diac{\X})^\circ \nvdash \Gamma$ and thus in
  particular $\prove{\Gamma,\X}$ cannot yield a proof and by
  Lemma~\ref{l:terminates}~(Termination) has to fail. Thus by (ii)
  $\Gamma$ is not $\X$-valid.  For (ii) we define a model $\M$ on an
  $\X$-frame for which we prove that it is a countermodel for
  $\Gamma$.  Let $\Gamma^*$ be the set sequent of the non-axiomatic
  finished sequent obtained.  Let $Y$ be the set of all cyclic leaves
  in $\Gamma^*$.  Let $S = \st(\Gamma^*) \setminus Y$.  Let $f: Y
  \rightarrow S$ be some function which maps a cyclic leaf to a
  sequent in $S$ whose root carries the same set of formulas and
  extend $f$ to $\st(\Gamma^*)$ by the identity on $S$.  Define a
  binary relation $\rel$ on $S$ such that $\Delta \rel \Lambda$ iff
  either 1) $\Lambda$ is an immediate subtree of $\Delta$ or 2)
  $\Delta$ has an immediate subtree $\Sigma \in Y$ and
  $f(\Sigma)=\Lambda$. Let $V(p) = \{\Delta \in S \,|\, \neg p \in
  \Delta\}$. Let $\M = (S,\rel^\X,V)$. We prove three claims about
  $\M$, each claim depending on the next.  Since all rules seen
  top-down preserve countermodels Claim 1 implies that $\M \not\models
  \Gamma$.

  \textbf{Claim 1} \; For each sequent $\Delta \in \st(\Gamma^*)$ we have
  that $\M, f(\Delta) \not\models \Delta$.

  By induction on the depth of $\Delta$. For depth zero this follows
  from Claim 2 and the fact that a formula is in $\Delta$ iff it is in
  $f(\Delta)$. So let 
\[
\Delta = A_1, \dots,A_m, [\Delta_1], \dots,
  [\Delta_n] \quad \mbox{ and } \quad n>0 \quad .
\]  Then $f(\Delta)=\Delta$. We have $\M,
  f(\Delta) \not\models A_i$ for all $i\leq m$ by Claim 2 and $\M,
  \Delta \not\models [\Delta_i]$ because $\Delta \rel f(\Delta_i)$ and
  by induction hypothesis $\M , f(\Delta_i) \not\models \Delta_i$.

  \textbf{Claim 2} \; For each sequent $\Delta \in S$ and for each
  formula $A \in \Delta$ we have that $\M, \Delta \not\models A$.

  By induction on the depth of $A$. For atoms it is clear from the
  definition of $\M$ and the fact that $\Gamma^*$ is not axiomatic.
  For the propositional connectives it is clear from the shape of the
  $\vlan, \vlor$-rules. If $A = \Box B$ then by the $\Box$-rule we
  have some $[\Lambda] \in \Delta$ with $B \in \Lambda$. By induction
  hypothesis we have $\M, \Lambda \not\models B$ and thus $\M, \Delta
  \not\models \Box B$. If $A=\Diamond B$ then by Claim 3 we have $B
  \in \Lambda$ for all $\Lambda$ with $\Delta \rel^\X \Lambda$, and
  thus $\M, \Lambda \not\models B$. Thus $\M,\Delta \not\models
  \Diamond B$.

  \textbf{Claim 3} \; For all sequents $\Delta,\Lambda\in S$ with
  $\Delta \rel^\X \Lambda$ and for each formula $A$ it holds that if
  $\Diamond A \in \Delta$ then $A \in \Lambda$.

  We make a case analysis on $\X$. Note that each modal logic has
  exactly one 45-complete axiomatisation, with the exception of {\sf
    S5}, which has two.

  {\bf K} $\X=\emptyset:$ By the definition of $\rel$ there is an
  immediate subtree of $\Delta$ whose root node carries the same set
  of formulas as the root node of $\Lambda$. By the $\kc$-rule we
  have $A$ in (the root node of) all immediate subtrees of $\Delta$.

  {\bf T} $\X=\{{\sf t}\} $ : $\Delta \rel^{\{\sf t\}} \Lambda$ iff $\Delta
  \rel \Lambda$ or $\Delta = \Lambda$. In the second case $A \in
  \Lambda$ follows from the $\diatc$-rule.

  {\bf KB} $\X=\{{\sf b}\}$: $\Delta \rel^{\{\sf b\}} \Lambda$ iff $\Delta
  \rel \Lambda$ or $\Lambda \rel \Delta$. In the second case $A \in
  \Lambda$ follows by the $\diabc$-rule.

  {\bf K4} $\X=\{{\sf 4}\}$: $\Delta \rel^{\{\sf 4\}} \Lambda$ iff there is a
  sequence
  \[\Delta = \Delta_0 \rightarrow \Delta_1 \rightarrow
  \Delta_2 \rightarrow \dots \rightarrow \Delta_n=\Lambda\;,
  \]
  with $n \geq 1$. An induction on $i$ gives us that $\Diamond A \in
  \Delta_i$ for $0 \leq i\leq n$ by using the $\diafourc$-rule. By the
  $\kc$-rule it follows that $A \in \Delta_n$.

  {\bf K5} $\X=\{{\sf 5}\}$: By Lemma~\ref{l:45}~($\rightarrow^\X$ is
  well-defined) we have $\Delta \rel^{\{\sf 5\}} \Lambda$ iff $\Delta
  \rel \Lambda$ or there is a euclidean connection from $\Delta$ to
  $\Lambda$. In the second case there are sequents $\Pi, \Sigma$ such
  that $\Delta \invrel \Pi$ and $\Sigma \rel \Lambda$. Thus there is
  an immediate subtree $\Delta'$ of $\Pi$ with the same formulas as
  $\Delta$ and an immediate subtree $\Lambda'$ of $\Sigma$ with the
  same formulas as $\Lambda$.  Since $\Diamond A \in \Delta$ we have
  $\Diamond A \in \Delta'$ and since $\Delta'\neq\Gamma^*$ by the
  $\diafivec$-rule we have $\Diamond A \in \Sigma$. Thus by the
  $\kc$-rule we have $A$ in $\Lambda'$ and thus in $\Lambda$.

  {\bf K45} $\X=\{{\sf 4,5}\}$: By Lemma~\ref{l:45}~($\rightarrow^\X$
  is well-defined) we have $\Delta \rel^{\{\sf 4,5\}} \Lambda$ iff
  $\Delta \rel \Lambda$ or there is a transitive-euclidean connection
  from $\Delta$ to $\Lambda$. In the second case there is a sequent
  $\Sigma$ such that $\Sigma\rel\Lambda$ and thus an immediate subtree
  $\Lambda'$ of $\Sigma$ with the same formulas as $\Lambda$. Since
  $\Diamond A \in \Delta$, by the $\diafivec$- and $\diafourc$-rules
  we have $\Diamond A$ in every subtree of $\Gamma^*$ and thus also in
  $\Sigma$, and by the $\kc$-rule we have $A$ in $\Lambda'$ and thus
  in $\Lambda$. (It is sufficient to have the {\diafivepac}-rule
  instead of the {\diafivec}-rule for all $\X$ which contain $\sf 4$.)

  {\bf KB5} $\X=\{{\sf b,4,5}\}$: $\Delta \rel^{\{\sf b,4,5\}}
  \Lambda$ iff $\Delta \symrel^+ \Lambda$. Thus there is a sequent
  $\Sigma$ such that either $\Sigma \rel \Lambda$ or $\Sigma \invrel
  \Lambda$. Rule $\sf 4,5$ imply that $\Diamond A$ is in every subtree
  of $\Gamma^*$ and thus in particular in $\Sigma$. We have $A \in
  \Lambda$ in the first case by the $\kc$-rule and in the second
  case by the $\diabc$-rule.

  {\bf KTB} $\X=\{{\sf b,t}\}$: $\Delta \rel^{\{\sf b,t\}} \Lambda$
  iff $\Delta \rel \Lambda$ or $\Delta \invrel \Lambda$ or $\Delta =
  \Lambda$. In these cases $A \in \Lambda$ respectively follows from
  the $\kc$- or $\diabc$- or $\diatc$-rule.

  {\bf S4} $\X=\{{\sf t,4}\}$: $\Delta \rel^{\{\sf t,4\}} \Lambda$ iff
  $\Delta \rel^+ \Lambda$ or $\Delta = \Lambda$. In the first case $A
  \in \Lambda$ follows from the rules $\diafourc$ and $\kc$ and in the
  second case from the $\diatc$-rule.

  {\bf S5}(1) $\X=\{{\sf t,4,5}\}$: $\Delta \rel^{\{\sf t,4,5\}}
  \Lambda$ iff $\Delta \symrel^* \Lambda$. We have $\Diamond A$ in all
  subtrees of $\Gamma^*$ by the rules $\diafourc,\diafivec$ and thus also $A$ by
  the $\diatc$-rule.

  {\bf S5}(2) $\X=\{{\sf d,b,4,5}\}$: $\Delta \rel^{\{\sf d,b,4,5\}}
  \Lambda$ iff $\Delta \symrel^* \Lambda$. We have $\Diamond A$ in all
  subtrees of $\Gamma^*$ by the rules $\diafourc,\diafivec$ and thus
also $\Diamond A \in \Lambda$. By the $\diadc$-rule the root of
$\Lambda$ has a child node. By the $\diafourc$-rule $\Diamond A$ is in
this child node and by the $\diabc$-rule $A \in \Lambda$.

  {\bf KD,KDB,KD4,KD5,KD45\,} The argument for all these cases is
  similar to the same system without {\sf d}. Take the corresponding
  $\X$, then $\Delta \rightarrow^{\X\cup\{{\sf d}\}} \Lambda$ iff
  $\Delta \rightarrow^{\X} \Lambda$ or ($\Delta=\Lambda$ and there is
  no $\Delta'$ with $\Delta\rightarrow^{\X}\Delta'$). In the second
  case, due to the ${\diadc}$-rule, there is no formula $\Diamond A$ in
  $\Delta$ and thus our claim is trivially true.
\end{proof}

Notice that each class of frames that can be characterised by our
modal axioms can also be characterised by a 45-closed set of axioms.
The restriction to 45-complete sets of rule names in the completeness
theorem is thus irrelevant for the two following corollaries.

\begin{corollary}[Finite Model Property]
For all $\X \subseteq \{{\sf d,t,b,4,5}\}$ it holds that if a formula is not 
$\X$-valid then there is a finite $\X$-frame in which it is not valid.
\end{corollary}
\begin{proof}
  Immediate from part (ii) of the completeness theorem.
\end{proof}

\begin{corollary}[Decidability]
For all $\X \subseteq \{{\sf d,t,b,4,5}\}$ it is decidable whether a
formula is $\X$-valid.
\end{corollary}
\begin{proof}
  By the termination lemma and part (ii) of the completeness theorem.
\end{proof}

\subsection{Syntactic Cut-Elimination}

While cut admissibility is an easy corollary of the completeness
theorem, it is still interesting to provide a nontrivial procedure
which removes cuts from a proof. The existence of a step-by-step cut
elimination procedure shows a certain symmetry, a certain good design
of the inference rules. Also, it can serve as a starting point for a
computational interpretation, maybe along the lines of \cite{Mas96}.

We now see a cut-elimination procedure which follows the lines of the
one for system {\sf G3} for first-order predicate logic, see for
example \cite{TroSch96}. The interesting twist is that the modalities
require some form of multicut, similar to Gentzen's original
procedure, even though contraction is admissible. We first need some
standard definitions.

\begin{definition}[depth of a formula]
  The \emph{depth} of a formula $A$, denoted by $\depth{A}$, is defined
  as usual:
\[
  \begin{array}{l}
  \depth{p} = \depth{\neg p} = 0\\
  \depth{\Box A} = \depth{\Diamond A} = \depth{A}+1\\
  \depth{A \vlan B} = \depth{A \vlor B} ={\mathit
    max}({\depth{A},\depth{B}}) +1 \quad .\\
\end{array}
\]
\end{definition}

\begin{definition}[cut rank, cut-rank-preserving]
  Given an instance of the $\cut$ rule as shown in
  Figure~\ref{fig:admissible}, its \emph{cut formula} is $A$ and its
  \emph{cut rank} is one plus the depth of its cut formula. For $r\geq
  0$ we define the rule $\cut_r$ which is cut with at most rank $r$.
  The \emph{cut rank} of a derivation is the supremum of the cut ranks
  of its cuts. A rule is \emph{cut-rank (and depth-) preserving
    admissible} for a system $\cal S$ if for all $r\geq 0$ the rule is
  (depth-preserving) admissible for ${\cal S}+\cut_r$. A rule is
  \emph{cut-rank (and depth-) preserving invertible} for a system
  $\cal S$ if its inverse is \emph{cut-rank (and depth-) preserving
    admissible} for $\cal S$.
\end{definition}

The problem with proving cut-elimination in the presence of the rules
{\diafourc} and {\diafivec} is that these rules, seen upwards, do not
decompose their main formula $\Diamond A$. If that formula happens to
be the cut formula, then we cannot form a new derivation by appealing to an
induction hypothesis based on a lower rank. We thus generalise the
cut-rule to incorporate instances of rules {\diafourc} and
{\diafivec}. This leads to the following definition.

\begin{definition}[\ycut]
  Let $\cons{\Delta}^n$ denote
$\underbrace{\cons{\Delta}\dots\cons{\Delta}}_{n-\text{times}}$ .
 For $\Y \subseteq \{{\sf 4,5}\}$ and $n\geq 0$ we define the rule
\[
\vliinf{\ycut}{}
{\Gamma\cons{\emptyset}\cons{\emptyset}^n}
{\Gamma\cons{\Box A}\cons{\emptyset}^n}
{\Gamma\cons{\Diamond \neg A}\cons{\Diamond \neg A}^n}
\] with the proviso that there is a derivation from
$\Gamma\cons{\Diamond \neg A}\cons{\Diamond \neg A}^n$ to
$\Gamma\cons{\Diamond \neg A}\cons{\emptyset}^n$ in system $\Y$.
\end{definition}

\begin{fact}[Properties of \ycut] Consider an instance of $\,\ycut$ as above.\nl If
  $\,\Y=\emptyset$ then it is an instance of $\,\cut$, so $n=0$.\\
  If $\,\Y=\{{\sf 4}\}$ then $\Gamma\context\context^n$ is of the form
  $\Gamma_1\cons{\context,\Gamma_2\context^n}$.\\
  If $\,\Y=\{{\sf 5}\}$ and $n>0$ then the first hole is inside a box, so
  $\depth{\Gamma\context\cons{\emptyset}^n}>0$.\\
  (If $\,\Y=\{{\sf 4,5}\}$ then nothing can be said about the context
  since the proviso is trivially fulfilled.)
\end{fact}

{\bf Structural modal rules.} The rules which are shown in
Figure~\ref{f:admmodstr} are called \emph{structural modal rules}.
They are structural in the sense of not affecting connectives of
formulas. The modal rules $\diac{\X}$ are all \emph{$\Diamond$-rules},
in the sense that the active formula in the conclusion has $\Diamond$
as main connective. Given a set $\X$ of names of modal axioms,
$\str{\X}$ is defined as $\{\str{\rho}\,|\, \rho \in \X\}$. The
structural modal rules have the obvious corresponding frame
conditions. 

We need the admissibility of these structural modal rules for our
cut-elimination procedure. In some sense, they are the result of
``reflecting'' the corresponding diamond-rule at the cut. This comment
will hopefully become more clear after the reduction lemma. The
structural modal rules are cut-rank preserving admissible, as we will
see.

The case of the seriality is a bit different from the other rules. The
rule $\strd$ is admissible, but we cannot show this in the presence of
cut. Consider the problematic case where {\strd} cannot be pushed
above {\cut}:
\[
\vlderivation{
\vlin{\strd}{}{\Gamma\cons{\emptyset}}{
\vliin{\cut}{}{\Gamma\cons{[\emptyset]}}
{\vlhy{\Gamma\cons{[A]}}}
{\vlhy{\Gamma\cons{[\neg A]}}}
}}
\]

So we cannot use $\strd$-admissibility in the cut-elimination proof.
Our solution is to eliminate cut in the presence of $\strd$ and only
afterwards replace $\strd$ by $\diadc$. This means that in the
following we always have to consider the possible presence of the
$\strd$-rule.

\begin{figure}
\fbox{
    \parbox{\ftextwidth}{
\medskip 
\[
      \cinf{\strd}{\Gamma\cons{\emptyset}}{\Gamma\cons{[\emptyset]}}
      \qquad \cinf {{\strt}} {\Gamma\cons{\Delta}}
      {\Gamma\cons{[\Delta]}} \qquad \cinf {\strb}
      {\Gamma\cons{[\Delta],\Sigma}} {\Gamma\cons{[\Delta,[\Sigma]]}}
      \] 

      \[ \qquad\qquad\qquad\qquad\cinf {\strfour}
      {\Gamma\cons{[[\Delta],\Sigma]}} {\Gamma\cons{[\Delta],[\Sigma]}} \qquad
      \cinf {\strfive} {\Gamma\cons{\emptyset}\cons{[\Delta]}}
      {\Gamma\cons{[\Delta]}\cons{\emptyset}}{\quad
        \depth{\Gamma\context\cons{\emptyset}}>0}
\]

}}
    \caption{Modal structural rules}
\label{f:admmodstr}
\end{figure}

Before we eliminate the cut we need to make sure that
contraction and weakening can be eliminated without increasing the cut
rank. We just strengthen Lemma~\ref{l:stradm} (Admissibility of
structural rules and invertibility) accordingly to get the following
lemma.

\begin{lemma}[Cut-rank preserving admissibility of structural rules, invertibility]\label{l:stradmcut}
  Let $\X=\modax$. For each system $\K+\Y$ with 
$\Y \subseteq \diac{\X} \cup \{\strd\}$ the following hold:\\
  (i) The rules $\nec,\wk$ and $\ctr$ are depth- and
  cut-rank  preserving admissible.\\
  (ii) All its rules are depth- and cut-rank preserving invertible.
\end{lemma}
\begin{proof}
  The proof is just like the one for
  Lemma~\ref{l:stradm}~(Admissibility of structural rules and
  invertibility) except that we also consider $\cut_r$ and $\strd$. In
  proving contraction admissibility there is one more case which is
  mildly interesting and which is handled as follows:
\[
\vlderivation{
\vlin{\ctr}{}{\Gamma\cons{\Delta\cons{\emptyset}}}
  {\vliin{\cut_r}{}{\Gamma\cons{\Delta\cons{\emptyset},\Delta\cons{\emptyset}}}
    {\vlhy{\Gamma\cons{\Delta\cons{A},\Delta\cons{\emptyset}}}}
{\vlhy{\Gamma\cons{\Delta\cons{\neg A},\Delta\cons{\emptyset}}}}}}
\qquad\leadsto
\]
\[
\vlderivation{
\vliin{\cut_r}{}{\Gamma\cons{\Delta\cons{\emptyset}}}{
\vlin{\ctr}{}{\Gamma\cons{\Delta\cons{A}}}{
\vlin{\wk}{}{\Gamma\cons{\Delta\cons{A},\Delta\cons{A}}}{
\vlhy{\Gamma\cons{\Delta\cons{A},\Delta\cons{\emptyset}}}}}}{
\vlin{\ctr}{}{\Gamma\cons{\Delta\cons{\neg
          A}}}{
\vlin{\wk}{}{\Gamma\cons{\Delta\cons{\neg
            A},\Delta\cons{\neg
            A}}}{
\vlhy{\Gamma\cons{\Delta\cons{\neg
A},\Delta\cons{\emptyset}}}}}}} \qquad .
\]

\end{proof}

\begin{lemma}[Admissibility of the modal structural rules]\label{l:strmodadm} \hspace{0mm}\\
  (i) Let $\X$ be a 45-closed subset of $\{{\sf t,b,4,5}\}$ and let
  $\rho \in \X$.  Then the rule $\str{\rho}$ is cut-rank preserving
  admissible for system $\K+\diac{\X}$ and also
  for system $\K+\diac{\X}+\strd$.\\
  (ii) Let $\X$ be a 45-closed subset of $\{{\sf d,t,b,4,5}\}$ and let
  ${\sf d}\in\X$. Then the rule $\strd$ is admissible for system
  $\K+\diac{\X}$.
\end{lemma}
\begin{proof}
  For (i) the proof works by an outer induction on the number of
  instances of $\str{\rho}$ in a given proof, eliminating topmost
  instances first, and an inner induction on the depth of the proof
  above such a topmost instance. For each rule $\str{\rho}$ with $\rho
  \in \X$ we make a case analysis on the rule $\sigma$ above
  $\str\rho$. The induction base and the cases where $\sigma$ is among
  the rules $\vlor, \vlan, \Box, \cut_r, {\strd}$ and $\diatc$ are trivial. We
  use cut-rank preserving admissibility of contraction and weakening
  provided by the previous lemma without explicitly mentioning
  it.

{ $\str\rho={\strt}:$}
\[
\vlderivation{
\vlin{\strt}{}{\Gamma\cons{\Diamond A, \Delta}}
{\vlin{{\kc}}{}{\Gamma\cons{\Diamond A, [\Delta]}}
{\vlhy{\Gamma\cons{\Diamond A, [A, \Delta]}}}}}
\quad \leadsto\quad
 \vlderivation{\vlin{{\diatc}}{}{\Gamma\cons{\Diamond A, \Delta}}
{\vlin{\strt}{}{\Gamma\cons{\Diamond A, A, \Delta}}
{\vlhy{\Gamma\cons{\Diamond A, [A, \Delta]}}}}}
\]

The case for $\sigma={\diabc}$ is similar.

\[
\vlderivation{\vlin{\strt}{}{\Gamma\cons{\Diamond A, \Delta}} {
\vlin{{\diafourc}}{}{\Gamma\cons{\Diamond A, [\Delta]}}
    {\vlhy{\Gamma\cons{\Diamond A, [\Diamond A, \Delta]}}}}}
\quad \leadsto\quad
\vlderivation{\vlin{\ctr}{}{\Gamma\cons{\Diamond A, \Delta}}
  {\vlin{\strt}{}{\Gamma\cons{\Diamond A, \Diamond A, \Delta}}
    {\vlhy{\Gamma\cons{\Diamond A, [\Diamond A, \Delta]}}}}}
\]

For $\sigma={\diafivec}$ the case is trivial unless the diamond formula in
its conclusion is at depth 1. Then there are two cases, either the
{\diafivec}-rule moves the formula to somewhere outside the box that is
removed by $\strt$ or somewhere inside it. The second case is
similar to the first, which is as follows, where $\rho^*$ denotes
several applications of $\rho$:

\[
\vlderivation{\vlin{\strt}{}{\Diamond A, \Delta, \Sigma\cons{\emptyset}}
{\vlin{{\diafivec}}{}{[\Diamond A], \Delta, \Sigma\cons{\emptyset}}
{\vlhy{[\Diamond A], \Delta, \Sigma\cons{\Diamond A}}}}}
\quad \leadsto\quad
 \vlderivation{\vlin{{\diafourc^*},\wk^*,\ctr^*}{}
{\Diamond A, \Delta, \Sigma\cons{\emptyset}}
{\vlin{\strt}{}{ \Diamond A, \Delta, \Sigma\cons{\Diamond A}}
{\vlhy{[\Diamond A], \Delta, \Sigma\cons{\Diamond A}}}}}
\]

{ $\str\rho={\strb}:$}

\[
\vlderivation{\vlin{\strb}{}{\Gamma\cons{\Sigma, [\Delta, \Diamond A]}}
{\vlin{{\kc}}{}{\Gamma\cons{[\Delta, \Diamond A, [\Sigma]]}}
{\vlhy{\Gamma\cons{[\Delta, \Diamond A, [A, \Sigma]]}}}}}
\quad \leadsto\quad
 \vlderivation{\vlin{{\diabc}}{}{\Gamma\cons{\Sigma, [\Delta, \Diamond A]}}
{\vlin{\strb}{}{ \Gamma\cons{\Sigma, A, [\Delta, \Diamond A]} }
{\vlhy{\Gamma\cons{[\Delta, \Diamond A, [A, \Sigma]]}}}}}
\]

\[
\vlderivation{\vlin{\strb}{}{\Gamma\cons{\Diamond A, \Sigma, [\Delta]}}
{\vlin{{\diabc}}{}{\Gamma\cons{[\Delta, [\Diamond A, \Sigma]]}}
{\vlhy{\Gamma\cons{[\Delta, A, [\Diamond A, \Sigma]]}}}}}
\quad \leadsto\quad
 \vlderivation{\vlin{{\kc}}{}{\Gamma\cons{\Diamond A, \Sigma, [\Delta]}}
{\vlin{\strb}{}{\Gamma\cons{\Diamond A, \Sigma, [A, \Delta]}}
{\vlhy{\Gamma\cons{[\Delta, A, [\Diamond A, \Sigma]]}}}}}
\]

\[
\vlderivation{\vlin{\strb}{}{\Gamma\cons{\Sigma,[\Diamond A, \Delta]}}
{\vlin{{\diafourc}}{}{\Gamma\cons{[\Diamond A, \Delta,[\Sigma]]}}
{\vlhy{\Gamma\cons{[\Diamond A, \Delta,[\Diamond A,\Sigma]]}}}}}
\quad \leadsto\quad
 \vlderivation{\vlin{{\diafivec}}{}{\Gamma\cons{\Sigma,[\Diamond A, \Delta]}}
{\vlin{\strb}{}{\Gamma\cons{\Sigma,\Diamond A, [\Diamond A, \Delta]}}
{\vlhy{\Gamma\cons{[\Diamond A, \Delta,[\Diamond A,\Sigma]]}}}}}
\]

For $\sigma={\diafivec}$ the case is trivial unless the diamond
formula in its conclusion is at depth 2 and in the inner box in the
premise of $\strb$.  Then there are three similar cases of which we
just see the following one:
\[
\vlderivation{\vlin{\strb}{}{\Diamond A, \Delta, [\Sigma],\Gamma\cons{\emptyset}}
{\vlin{{\diafivec}}{}{[\Sigma,[\Diamond A, \Delta]],\Gamma\cons{\emptyset}}
{\vlhy{[\Sigma,[\Diamond A, \Delta]],\Gamma\cons{\Diamond A}}}}}
\quad \leadsto\quad
 \vlderivation{\vlin{{\diafourc}^*,\wk^*,\ctr^*}{}
{\Diamond A, \Delta, [\Sigma],\Gamma\cons{\emptyset}}
{\vlin{\strb}{}{\Diamond A, \Delta, [\Sigma],\Gamma\cons{\Diamond A}}
{\vlhy{[\Sigma,[\Diamond A, \Delta]],\Gamma\cons{\Diamond A}}}}}
\]

{ $\str\rho={\strfour}:$}

\[
\vlderivation{\vlin{\strfour}{}{\Gamma\cons{\Diamond A, [[\Delta],\Sigma]}}
{\vlin{{\kc}}{}{\Gamma\cons{\Diamond A, [\Delta],[\Sigma]}}
{\vlhy{\Gamma\cons{\Diamond A, [A, \Delta],[\Sigma]}}}}}
\quad \leadsto\quad
 \vlderivation{\vlin{{\diafourc}}{}{\Gamma\cons{\Diamond A, [[\Delta],\Sigma]}}
{\vlin{\wk, {\kc}}{}{\Gamma\cons{\Diamond A, [\Diamond A, [\Delta],\Sigma]}}
{\vlin{\strfour}{}{\Gamma\cons{\Diamond A, [[A,\Delta],\Sigma]}}
{\vlhy{\Gamma\cons{\Diamond A, [A, \Delta],[\Sigma]}}}}}}
\]

The case for $\sigma={\diafourc}$ is similar and the case for
$\sigma={\diafivec}$ is trivial.

\[
\vlderivation{\vlin{\strfour}{}{\Gamma\cons{[[\Diamond A, \Delta],\Sigma]}}
{\vlin{{\diabc}}{}{\Gamma\cons{[\Diamond A, \Delta],[\Sigma]}}
{\vlhy{\Gamma\cons{A, [\Diamond A, \Delta],[\Sigma]}}}}}
\quad \leadsto\quad
 \vlderivation{\vlin{{\diafivec}}{}{\Gamma\cons{[[\Diamond A, \Delta],\Sigma]}}
{\vlin{\wk, {\diabc}}{}{\Gamma\cons{[\Diamond A, [\Diamond A, \Delta],\Sigma]}}
{\vlin{\strfour}{}{\Gamma\cons{A, [[\Diamond A, \Delta],\Sigma]}}
{\vlhy{\Gamma\cons{A, [\Diamond A, \Delta],[\Sigma]}}}}}}
\]
{ $\str\rho={\strfive}:$}
\[
\vlderivation{\vlin{\strfive}{}
{\Gamma\cons{\Diamond A}\cons{[\Delta]}}
{\vlin{{\kc}}{}{\Gamma\cons{\Diamond A,[\Delta]}\cons{\emptyset}}
{\vlhy{\Gamma\cons{\Diamond A,[A,\Delta]}\cons{\emptyset}}}}}
\quad \leadsto\quad
 \vlderivation{\vlin{{\diafivec}}{}{\Gamma\cons{\Diamond A}\cons{[\Delta]}}
{\vlin{\wk,{\kc}}{}{\Gamma\cons{\Diamond A}\cons{\Diamond A, [\Delta]}}
{\vlin{\strfive}{}{\Gamma\cons{\Diamond A}\cons{[A,\Delta]}}
{\vlhy{\Gamma\cons{\Diamond A,[A,\Delta]}\cons{\emptyset}}}}}}
\]

The case for $\sigma={\diafourc}$ is similar and the case for $\sigma=
{\diafivec}$ is trivial. For $\sigma={\diabc}$ we have:
\[
\vlderivation{\vlin{\strfive}{}{\Gamma\cons{[\Sigma]}\cons{[\Diamond A, \Delta]}}
{\vlin{{\diabc}}{}{\Gamma\cons{[[\Diamond A, \Delta],\Sigma]}\cons{\emptyset}}
{\vlhy{\Gamma\cons{[A,[\Diamond A, \Delta],\Sigma]}\cons{\emptyset}}}}}
\quad \leadsto\quad
 \vlderivation{\vlin{{\diafivec}}{}{\Gamma\cons{[\Sigma]}\cons{[\Diamond A, \Delta]}}
{\vlin{\wk,{\kc}}{}{\Gamma\cons{\Diamond A,[\Sigma]}\cons{[\Diamond A, \Delta]}}
{\vlin{\strfive}{}{\Gamma\cons{[A,\Sigma]}\cons{[\Diamond A, \Delta]}}
{\vlhy{\Gamma\cons{[A,[\Diamond A, \Delta],\Sigma]}\cons{\emptyset}}}}}}
\]

The proof for (ii) is similar to the  one for (i), except that we
exclude $\sigma=\cut_r$. The case $\sigma={\diabc}$ is trivial.

{ $\str\rho={\strd}:$}

\[
\vlderivation{\vlin{\strd}{}{\Gamma\cons{\Diamond A}}{
\vlin{{\kc}}{}{\Gamma\cons{\Diamond A,
        [\emptyset]}}{\vlhy{\Gamma\cons{\Diamond A, [A]}}}}}
\quad\leadsto\quad \vlderivation{\vlin{{\diadc}}{}{\Gamma\cons{\Diamond
      A}}{\vlhy{\Gamma\cons{\Diamond A,[A]}}}}
\]

\[
\vlderivation{\vlin{\strd}{}{\Gamma\cons{\Diamond A}}{
\vlin{{\diafourc}}{}{\Gamma\cons{\Diamond A,[\emptyset]}}{
\vlhy{\Gamma\cons{\Diamond A,[\Diamond A]}}}}}
\quad\leadsto\quad
\vlderivation{\vlin{{\diadc}}{}{\Gamma\cons{\Diamond A}}{
\vlin{{\diafourc}}{}{\Gamma\cons{\Diamond A,[A]}}{
\vlin{\wk}{}{\Gamma\cons{\Diamond
A,[\Diamond A,A]}}{\vlhy{\Gamma\cons{\Diamond A,[\Diamond A]}}}}}}
\]

\[
\vlderivation{\vlin{\strd}{}{\Gamma\cons{\Diamond
      A}\cons{\emptyset}}{
\vlin{{\diafivec}}{}{\Gamma\cons{\Diamond
        A}\cons{[\emptyset]} }{\vlhy{ \Gamma\cons{\Diamond
          A}\cons{[\Diamond A]} }}}} 
\quad\leadsto\quad
\vlderivation{\vlin{{\diafivec}}{}{\Gamma\cons{\Diamond
      A}\cons{\emptyset}}{
\vlin{{\diadc}}{}{\Gamma\cons{\Diamond
      A}\cons{\Diamond A}}{
\vlin{{\diafivec}}{}{\Gamma\cons{\Diamond
      A}\cons{\Diamond A, [A]}}{
\vlin{\wk^2}{}{\Gamma\cons{\Diamond
      A}\cons{\Diamond A, [\Diamond A, A]}}{\vlhy{\Gamma\cons{\Diamond
              A}\cons{[\Diamond A]}}}}}}}
\]

\end{proof}

To keep the cut-elimination procedure short and uniform, we define a
structural rule which moves a box inside a sequent from one place to
another. Notice that the conditions on the context in the proviso
exactly match the conditions in the $\ycut$-rule:

\begin{definition}[{\ystr}-rule]
  For $\Y \subseteq \{{\sf 4,5}\}$ we define a rule
\[
\kaiinf{\ystr}{\Gamma\cons{\emptyset}\cons{[\Delta]}}{\Gamma\cons{[\Delta]}\cons{\emptyset}}
\]
with the proviso that:\\
if $\Y=\emptyset$ then $\Gamma\context\context$ is of the form
$\Gamma'\cons{\context,\context}$,\\
if $\Y=\{{\sf 4}\}$ then $\Gamma\context\context$ is
of the form $\Gamma_1\cons{\context,\Gamma_2\context}$, and\\
if   $\Y=\{{\sf 5}\}$ then $\depth{\Gamma\context\cons{\emptyset}} >
0$.\\ (This means there is no proviso for the case $\Y=\{{\sf 4,5}\}$.)
\end{definition}

\begin{lemma}[Admissibility of {\ystr}]\label{l:yadm}
  For 45-closed $\X \subseteq \{{\sf \strd,t,b,4,5}\}$ and for
  $\Y\subseteq \{{\sf 4,5}\}$ the rule $\ystr$ is cut-rank preserving
  admissible for system $\K+\X$ if $\Y\subseteq \X$.
\end{lemma}
\begin{proof}
  For $\Y=\emptyset$ that is trivial. For $\Y=\{{\sf 4}\}$ the rule is
  derivable as follows:
\[
\vlderivation{\vlin{\ctr}{}{\Gamma\cons{\Sigma\cons{[\Delta]}}}{
\vlin{\wk}{}{\Gamma\cons{\Sigma\cons{[\Delta]},\Sigma\cons{[\Delta]}}}{
\vlin{\wk^*}{}{\Gamma\cons{\Sigma\cons{[\Delta]},\Sigma\cons{\emptyset}}}{
\vlin{\strfour^*}{}{\Gamma\cons{[\dots[\Delta]\dots],\Sigma\cons{\emptyset}}}{
\vlhy{\Gamma\cons{[\Delta], \Sigma\cons{\emptyset}} }}}}}} \qquad ,
\]
and thus admissible by Lemma~\ref{l:stradmcut}~(Cut-rank preserving
admissibility of structural rules) and
Lemma~\ref{l:strmodadm}~(Admissibility of the modal structural rules).
For $\Y=\{{\sf 5}\}$ the rule coincides with $\strfive$ and is thus
admissible by Lemma~\ref{l:strmodadm}. For $\Y=\{{\sf 4,5}\}$ an
instance of the rule is either an instance of the $\ystr$-rule for
$\Y=\{{\sf 4}\}$ or $\Y=\{{\sf 5}\}$ and thus admissible as in the
previous two cases.  \end{proof}

\begin{lemma}[Reduction Lemma]\label{l:rl} 
  Let $\X$ be a 45-closed subset of $\{{\sf t,b,4,5}\}$, let
  $\Y$ be a subset of $\{{\sf 4,5}\} \cap \X$ and let either
  $\Z=\diac{\X}$ or $\Z=\diac{\X}+\strd$. Further, let $r >0$ and $n \geq 0$.\\
  (i) If there is a proof
\[
\vlderivation{\vliin{\cut_{r+1}}{} {\Gamma\cons{\emptyset}}
  {\Proofleaf{\PP_1}{\Gamma\cons{A}}}
  {\Proofleaf{\PP_2}{\Gamma\cons{\neg A}}}}
\]
with $\PP_1$ and $\PP_2$ in $\K+\Z+\cut_r\,$  then  $\,\K+\Z+\cut_r \vdash \Gamma\cons{\emptyset}\,$.\\
(ii)  If there is a proof
\[
\vlderivation{\vliin {\ycut_{r+1}}{}
  {\Gamma\cons{\emptyset}\cons{\emptyset}^n}
  {\Proofleaf{\PP_1}{\Gamma\cons{\Box A}\cons{\emptyset}^n}}
  {\Proofleaf{\PP_2}{\Gamma\cons{\Diamond \neg A}\cons{\Diamond \neg A}^n}}}
\]
with $\PP_1$ and $\PP_2$ in $\K+\Z+\cut_r\,$ then  $\,\K+\Z+\cut_r \vdash \Gamma\cons{\emptyset}\cons{\emptyset}^n\,$.
\end{lemma}
\begin{proof}
  We prove (i) and (ii) simultaneously by induction on
  $|\PP_1|+|\PP_2|$. We perform a case analysis on the two lowermost
  rules in $\PP_1$ and $\PP_2$. If one of the two rules is passive and
  an axiom then $\Gamma\cons{\emptyset}$ is axiomatic as well. If one
  is active and an axiom then we have
\[
\vlderivation{\vliin {\cut_{r+1}}{} {\Gamma\cons{\neg a}}
  {\vlhy{\Gamma\cons{a, \neg a}}}
  {\Proofleaf{\PP_2}{\Gamma\cons{\neg a, \neg a}}}}
\qquad \leadsto \qquad
\vlderivation{\vlin {\ctr}{} {\Gamma\cons{\neg a}}
  {\Proofleaf{\PP_2}{\Gamma\cons{\neg a, \neg a}}}}
\quad .
\]
If one rule is passive then we have
\[
\vlderivation{\vliin {\cut_{r+1}}{} {\Gamma\cons{\emptyset}}
  {\Proofleaf{\PP_1}{\Gamma\cons{A}}}
  {\vlin{\rho}{}{\Gamma\cons{\neg A}}{
  \Proofleaf{\PP_2}{\Gamma'\cons{\neg A}}}}}
\qquad \leadsto \qquad
\vlderivation{\vlin{\rho}{}{\Gamma\cons{\emptyset}}
  {\vliin {\cut_{r+1}}{} {\Gamma'\cons{\emptyset}}
  {\vlin{\neg\rho}{}{\Gamma'\cons{A}}{\Proofleaf{\PP_1}{\Gamma\cons{A}}}}
  {\Proofleaf{\PP_2}{\Gamma'\cons{\neg A}}}}}
\]
for case (i) and similarly for (ii). This leaves the case that
both rules are active and not axioms. For (i) we have:
\[
\vlderivation{\vliin {\cut_{r+1}}{} {\Gamma\cons{\emptyset}}
  {\vliin{\vlan}{}{\Gamma\cons{B \vlan C}}
  {\Proofleaf{\PP_1}{\Gamma\cons{B}}}
  {\Proofleaf{\PP_2}{\Gamma\cons{C}}}}
  {\vlin{\rho}{}{\Gamma\cons{\neg B \vlor \neg C}}{
  \Proofleaf{\PP_3}{\Gamma'\cons{\neg B, \neg C}}}}}
\qquad \leadsto \qquad
\]\[
\vlderivation{\vliin{\cut_r}{}{\Gamma\cons{\emptyset}}
{\Proofleaf{\PP_1}{\Gamma\cons{B}}}
{\vliin{\cut_{r}}{} {\Gamma\cons{\neg B}}
  {\vlin{\wk}{}{\Gamma\cons{\neg B, C}}{\Proofleaf{\PP_2}{\Gamma\cons{C}}}}
  {\Proofleaf{\PP_3}{\Gamma\cons{\neg B, \neg C}}}}} \qquad .
\]
Notice that (i) is a special case of (ii) if $A$ has a modality as
its main connective. The remaining case is thus (ii) with both rules
active and not axioms, and thus on one side the $\Box$-rule and on the
other side either ${\kc, \diatc}$ or ${\diabc}$ (the cases for {\diafourc} and
{\diafivec} are trivial).  The case for the {\kc}-rule is as follows:
\[
\vlderivation{
\vliin{\ycut_{r+1}}{}{\Gamma\cons{\emptyset}\cons{[\Delta]}}
{\vlin{\Box}{}{\Gamma\cons{\Box A}\cons{[\Delta]}}
  {\Proofleaf{\PP_1}{\Gamma\cons{[A]}\cons{[\Delta]}}}}
{\vlin{{\kc}}{}{\Gamma'\cons{\Diamond \neg A}\cons{\Diamond \neg A,
      [\Delta']}} {\Proofleaf{\PP_2}{\Gamma'\cons{\Diamond \neg
        A}\cons{\Diamond \neg A, [\neg A, \Delta']}}}}} \qquad
\leadsto \qquad
\]
\[
\vlderivation{
\vliin{\cut_r}{}{\Gamma\cons{\emptyset}\cons{[\Delta]}}
{\vlin{\ctr}{}
{\Gamma\cons{\emptyset}\cons{[A, \Delta]}}
{\vlin{\wk^2}{}{\Gamma\cons{\emptyset}\cons{[A, \Delta],[A, \Delta]} }
{\vlin{\ystr}{}
{\Gamma\cons{\emptyset}\cons{[A],[\Delta]}}
      {\Proofleaf{\PP_1}{\Gamma\cons{[A]}\cons{[\Delta]}}}}}}
{\vliin{\ycut_{r+1}}{}{\Gamma\cons{\emptyset}\cons{[\neg A, \Delta]}}
  {\vlin{\Box,\wk}{}{\Gamma\cons{\Box
        A}\cons{[\neg A, \Delta]}}
    {\Proofleaf{\PP_1}{\Gamma\cons{[A]}\cons{[\Delta]}}}}
  {\Proofleaf{\PP_2}{\Gamma'\cons{\Diamond \neg A}\cons{\Diamond \neg
        A, [\neg A, \Delta']}}}}}
\qquad ,
\]
where the $\ystr$-rule is applicable since its condition on the
context matches the condition in the $\ycut$-rule. The $\ystr$-rule
can be removed by Lemma~\ref{l:yadm} (Admissibility of {\ystr}),
weakening and contraction can be removed by Lemma~\ref{l:stradmcut}
(Cut-rank preserving admissibility of structural rules) and the
instance of $\ycut$ can be removed by induction hypothesis. The cases
for {\diatc} and {\diabc} are as follows:
\[
\vlderivation{
\vliin{\ycut_{r+1}}{}{\Gamma\cons{\emptyset}\cons{\emptyset}}{
\vlin{\Box}{}{\Gamma\cons{\Box A}\cons{\emptyset}}
{\Proofleaf{\PP_1}{\Gamma\cons{[A]}\cons{\emptyset}}}}
{\vlin{{\diatc}}{}{\Gamma'\cons{\Diamond \neg A}\cons{\Diamond \neg A}}
{\Proofleaf{\PP_2}{\Gamma'\cons{\Diamond \neg A}\cons{\neg A}}}}}
\qquad
\leadsto \qquad
\]
\[
\vlderivation{
\vliin{\cut_r}{}{\Gamma\cons{\emptyset}\cons{\emptyset}}
{\vlin{\strt}{}{\Gamma\cons{\emptyset}\cons{A}}
{\vlin{\ystr}{}{\Gamma\cons{\emptyset}\cons{[A]}}
{\Proofleaf{\PP_1}{\Gamma\cons{[A]}\cons{\emptyset}}}}}
{\vliin{\ycut_{r+1}}{}{\Gamma\cons{\emptyset}\cons{\neg A}}
{\vlin{\Box,\wk}{}{\Gamma\cons{\Box A}\cons{\neg A}}
{\Proofleaf{\PP_1}{\Gamma\cons{[A]}\cons{\emptyset}}}}
{\Proofleaf{\PP_2}{\Gamma'\cons{\Diamond \neg A}\cons{\neg A}}}}}
\]
and
\[
\vlderivation{
\vliin{\ycut_{r+1}}{}{\Gamma\cons{\emptyset}\cons{[\Delta]}}
{\vlin{\Box}{}{\Gamma\cons{\Box A}\cons{[\Delta]}}
  {\Proofleaf{\PP_1}{\Gamma\cons{[A]}\cons{[\Delta]}}}}
{\vlin{\diabc}{}{\Gamma'\cons{\Diamond \neg A}\cons{[\Diamond \neg
A,\Delta']}}
 {\Proofleaf{\PP_2}{\Gamma'\cons{\Diamond \neg
        A}\cons{ \neg A, [\Diamond\neg A, \Delta']}}}}} \qquad
\leadsto \qquad
\]
\[
\vlderivation{
\vliin{\cut_r}{}{\Gamma\cons{\emptyset}\cons{[\Delta]}}
{\vlin{\strb}{}{\Gamma\cons{\emptyset}\cons{A, [\Delta]} }
{\vlin{\ystr}{}
{\Gamma\cons{\emptyset}\cons{[[A],\Delta]}}
      {\Proofleaf{\PP_1}{\Gamma\cons{[A]}\cons{[\Delta]}}}}}
{\vliin{\ycut_{r+1}}{}{\Gamma\cons{\emptyset}\cons{\neg A, [\Delta]}}
  {\vlin{\Box,\wk}{}{\Gamma\cons{\Box
        A}\cons{\neg A, [\Delta]}}
    {\Proofleaf{\PP_1}{\Gamma\cons{[A]}\cons{[\Delta]}}}}
  {\Proofleaf{\PP_2}{\Gamma'\cons{\Diamond \neg A}\cons{\neg
        A, [\Diamond \neg A, \Delta']}}}}}
\qquad ,
\]

In general the $\ycut$, seen upwards, introduces several diamond
formulas. One of them is special in being in the same position as its
dual cut formula in the other premise. In the transformations given
above, the active formula of the diamond-rule above the cut is
different from that special formula.  That is not always the case, of
course, but if the two coincide, then the transformations are simpler.

\end{proof}

\begin{theorem}[Cut-Elimination]\label{t:ce}
  Let $\X$ be a 45-closed subset of $\{{\sf d,t,b,4,5}\}$. Then we have:
\begin{center}
  If $\,\K+\diac{\X}+\cut \vdash \Gamma\,$ then $\,\K+\diac{\X} \vdash
  \Gamma\,$.
\end{center}
\end{theorem}

\begin{proof}
  We first prove the theorem in case that ${\sf d} \notin \X$.  Then
  it follows from a routine induction on the cut-rank of the given
  proof.  The induction step follows by another induction, on the
  depth of the proof. It uses the reduction lemma in the case of a
  maximal-rank cut. In case ${\sf d} \in \X$ we first replace
  instances of the rule $\diadc$ by instances of the rules $\kc$ and
  $\strd$, then proceed as before, and finally apply
  Lemma~\ref{l:strmodadm} (Admissibility of modal structural rules) to
  replace $\strd$ by $\diadc$.
\end{proof}

This finishes the section of sequent systems where modal axioms are
represented as logical rules. The systems cover the entire modal cube
and are systematic in the sense that there is a one-to-one
correspondence between the modal rules and the frame conditions.
However, unlike Hilbert systems and labelled sequent systems, they are
not modular in the sense that each combination of modal rules is
complete for the corresponding class of frames. This forced us to
resort to formulating the condition of 45-closed systems
and proving completeness only for those. It is hard to see how to
achieving modularity using these systems. 

However, during the cut-elimination procedure we discovered the
possibility of forming proof systems not using $\Diamond$-rules but
using the structural rules shown in Figure~\ref{f:admmodstr} on
page~\pageref{f:admmodstr}.  

In particular, the examples from
Fact~\ref{f:45needed}~(Incompleteness) which showed that systems
$\K+\{{\diatc,\diafivec}\}$ and $\K+\{{\diabc,\diafourc}\}$ are
incomplete are provable in systems $\K+\{{\strt,\strfive}\}$ and
$\K+\{{ \strb,\strfour}\}$, respectively:
\[
\vlderivation{\vlin{\Box^2}{}{\Diamond \neg p, \Box\Box
    p}{
\vlin{\strt}{}{\Diamond \neg p, [[p]]}{
\vlin{\strfive}{}{[\Diamond
        \neg p, [[p]]]}{
\vlin{{\sf k}}{}{[\Diamond \neg p,
          [p],[\emptyset]]}{\vlhy{[\Diamond \neg p, [p, \neg
            p],[\emptyset]]}}}}}} \qquad \mbox{and} \qquad
\vlderivation{
\vlin{\Box^2}{}{\Box \neg p,\Box\Diamond p}{
\vlin{\strb}{}{[\neg p],[\Diamond p]}{
\vlin{\strfour}{}{[[[\neg p]],\Diamond p]}{
\vlin{{\sf k}}{}{[[\neg p],\Diamond p]}{\vlhy{[[\neg
p,p],\Diamond p]}}}}}}
\qquad .
\]

We consider such proof systems in the next section.

\section{Modal Axioms as Structural Rules}

The plan of this section is as follows: we first introduce the sequent
systems and state soundness, cut-elimination and completeness, which
we prove by embedding a Hilbert system and using cut-elimination. The
remainder of the section is devoted to proving cut-elimination.
The cut-elimination proof is
interesting: it relies on a decomposition of the contraction rule,
similar to what has been observed in deep inference systems for
propositional logic, where contraction is decomposed into an atomic
version and a local \emph{medial} rule \cite{BruTiu01}.

\subsection{The Sequent Systems}

{\bf System $\Kc+\str{\X}$.} Figure~\ref{fig:kcx} shows the set of
rules from which we form our deductive systems. \emph{System \Kc} is
the set of rules $\{\vlan,\vlor,\Box,{\sf k},\ctr\}$. We will look at
extensions of System $\Kc$ with the structural modal rules $\str{\X}
\subseteq \{\strd,\strt,\strb,\strfour,\strfive\}$ that we have
encountered previously and that are shown in Figure~\ref{fig:kcx} for
convenience. Contrary to systems considered in the last section, the
systems we consider now are not fully invertible and contraction is
not admissible for the contraction-free systems. Of course it is easy
to obtain equivalent systems which are fully invertible and for which
contraction is admissible by using system $\K$ instead of $\Kc$ and by
absorbing contraction into the modal structural rules. However, we
choose not to do this because our cut-elimination technique, which
relies on decomposing contraction, is more natural in a system with an
explicit contraction rule.

\begin{figure}
\centering
\fbox{
\parbox{\ftextwidth}{
  \centering
  \[ \Gamma\cons{p,\neg p} \qquad \cinf{\vlan}{\Gamma\cons{A \vlan
      B}}{\Gamma\cons{A} \quad \Gamma\cons{B}} \qquad
  \cinf{\vlor}{\Gamma\cons{A \vlor B}}{\Gamma\cons{A,B}}
\]
\[
  \kaiinf{\Box}{\Gamma\cons{\Box A}}{\Gamma\cons{[A]}} \qquad
\cinf{{\sf k}}{\Gamma\cons{\Diamond A, [\Delta]}}{\Gamma\cons{[A,
\Delta]}} \qquad 
\cinf{\ctr}{\Gamma\cons{\Delta}}{\Gamma\cons{\Delta,\Delta}}
\]
\medskip 
\[
      \cinf{\strd}{\Gamma\cons{\emptyset}}{\Gamma\cons{[\emptyset]}}
      \qquad \cinf {{\strt}} {\Gamma\cons{\Delta}}
      {\Gamma\cons{[\Delta]}} \qquad \cinf {\strb}
      {\Gamma\cons{[\Delta],\Sigma}} {\Gamma\cons{[\Delta,[\Sigma]]}}
      \] 

      \[ \qquad\qquad\qquad\qquad\cinf {\strfour}
      {\Gamma\cons{[[\Delta],\Sigma]}} {\Gamma\cons{[\Delta],[\Sigma]}} \qquad
      \cinf {\strfive} {\Gamma\cons{\emptyset}\cons{[\Delta]}}
      {\Gamma\cons{[\Delta]}\cons{\emptyset}}{\quad
        \depth{\Gamma\context\cons{\emptyset}}>0}
\]
}}
\caption{System {\Kc+\{\strd,\strt,\strb,\strfour,\strfive
    }\}}
  \label{fig:kcx}
\end{figure}

Soundness of our systems is easily established similarly to soundness
of the systems in the previous section.

\begin{theorem}[Soundness] 
  Let $\X \subseteq \{{\sf d,t,b,4,5}\}$.  If a sequent is provable in
  $\,\Kc+\str{\X}$ then its corresponding formula is provable in a
  Hilbert system for the modal logic $\K$ extended by the axioms in
  $\X$.
\end{theorem}

Our main result is cut-elimination, which we prove in the next
subsection.

\begin{theorem}[Cut-Elimination]\label{t:cestr}
  Let $\X \subseteq \modax$. If $\Kc + \str{\X} + \cut
\vdash \Gamma$ then $\Kc+ \str{\X}\vdash \Gamma$. 
\end{theorem}

By using cut-elimination we obtain the completeness theorem:

\begin{theorem}[Completeness] 
  Let $\X \subseteq \modax$.  If a formula is provable in a
  Hilbert system for the modal logic $\K$ extended by the modal axioms in
  $\X$ then it is provable in system $\Kc+\str{\X}$.
\end{theorem}
\begin{proof}
  Given a proof in the Hilbert system we construct a proof in
  $\Kc+\str{\X}+\cut$ as usual, and then apply Theorem~\ref{t:cestr}
  (Cut-elimination). We show proofs for the modal axioms:
\[
\vlderivation{
\vlin{\vlor}{}{\Box A \imp \Diamond A}{
\vlin{\strd}{}{\Diamond \neg A,\Diamond A}{
\vlin{{\sf k}^2}{}{\Diamond \neg A,\Diamond A,[\emptyset]}{
\vlhy{[\neg A,A]}}}}
}
\qquad
\vlderivation{
\vlin{\vlor}{}{A\imp \Diamond A}{
\vlin{\strt}{}{\neg A, \Diamond A}{
\vlin{{\sf k}}{}{[\neg A], \Diamond A}{
\vlhy{[A,\neg A]}}}}
}
\quad
\vlderivation{
\vlin{\vlor}{}{A \imp \Box\Diamond A }{
\vlin{\Box}{}{\neg A, \Box\Diamond A}{
\vlin{\strb}{}{\neg A, [\Diamond A]}{
\vlin{{\sf k}}{}{[[\neg A],\Diamond A]}{
\vlhy{[[A,\neg A]]}}}}}
}
\quad
\vlderivation{
\vlin{\vlor}{}{\Box A \imp \Box\Box A}{
\vlin{\Box^2}{}{\Diamond\neg A, \Box\Box A }{
\vlin{\strfour}{}{\Diamond\neg A, [[A]]}{
\vlin{{\sf k}}{}{\Diamond\neg A,[A], [\emptyset]}{
\vlhy{[\neg A,A], [\emptyset]}}}}}
}
\quad \vlderivation{
\vlin{\vlor}{}{\Diamond A \imp \Box\Diamond A}{
\vlin{\Box^2}{}{\Box\neg A, \Box\Diamond A}{
\vlin{\strfive}{}{[\neg A], [\Diamond A]}{
\vlin{{\sf k}}{}{[[\neg A],\Diamond A]}{
\vlhy{[[\neg A,A]]}}}}}
}
\qquad .
\]

\end{proof}

\subsection{Syntactic Cut-Elimination}

We first show that weakening and necessitation are admissible.

\begin{lemma}[Weakening and necessitation admissibility]\label{l:wk}
  Let $\X \subseteq \{{\sf d,t,b,4,5}\}$.  The $\wk$-rule and the
  $\nec$-rule are depth- and cut-rank-preserving admissible for
  $\Kc+\str{\X}$.
\end{lemma}
\begin{proof}
  A routine induction shows that a single $\nec$ or $\wk$-rule can be
  eliminated from a given proof, a second induction on the number of
  $\nec$ or $\wk$-rules yields our lemma.
\end{proof}

Similarly to the $\diadc$-rule in the previous section, the
$\strd$-rule is different from the other rules: it trivially permutes
below the cut. So we can get it out of the way and then we need to prove
cut-elimination only for the systems without it.

\begin{lemma}[Push down seriality]\label{l:ddown}
  Let $\X \subseteq \{{\sf d,t,b,4,5}\}$ and $\sf d \in \X$. For each
proof as shown on the left there is a proof as shown on the right:
\[
\vlderivation{
 \vltrnote{}{\Gamma}{
\vlhy {}}{\vlhy {\qquad}}{\vlhy {}}{\Kc + \str{\X} + \cut}}
\qquad\leadsto\qquad
\vlderivation{
 \vlde{}{\strd}{\Gamma}{
 \vltrnote{}{\Gamma'}{
\vlhy {}}{\vlhy {\qquad}}{\vlhy {}}{\Kc + \str{\X}-\strd + \cut}}}
\qquad .
\]  
\end{lemma}
\begin{proof}
  By an easy permutation argument, making use of weakening
  admissibility.
\end{proof}

We also get contraction out of the way in order to eliminate the cut.
First, we decompose contraction into the $\fctr$-rule, which is
contraction on formulas, and the $\med$-rule, shown in
Figure~\ref{f:mrules}. We permute down the $\fctr$-rule. It does not
permute down below the rules $\cut, \Box$ and $\vlan$, so we
generalise these rules as in Figure~\ref{f:mrules}. We define a
contraction-free system $\Km$ as $\Km = \Kc - \ctr + \{ \med,
\mboxx, \mand\}$ and will show cut-elimination for that system. But
first we develop the machinery to show that cut-elimination for
$\Km$ leads to cut-elimination for $\Kc$ (with any $\str{\X}$).

\begin{figure}
\fbox{\parbox{\ftextwidth}{
  \centering
\[
  \vlinf{\mboxx}{}{\Gamma\cons{\Box A}}{\Gamma\cons{[A,\dots,A]}}
\qquad
  \vliinf{\mand}{}{\Gamma\cons{A\vlan B}}{\Gamma\cons{A,\dots,A}}
{\Gamma\cons{B,\dots,B]}}
\]
\[
\vliinf{\mcut}{}{\Gamma\cons{\emptyset}}{\Gamma\cons{A,\dots,A}}
{\Gamma\cons{\neg A, \dots, \neg A}}
\qquad
\vlinf{\med}{}{\Gamma\cons{[\Delta, \Sigma]}}{\Gamma\cons{[\Delta],[
    \Sigma]}}
\qquad
\vlinf{\fctr}{}{\Gamma\cons{A}}{\Gamma\cons{A,A}}
\]
}}
  \caption{Multi-rules, medial, and formula contraction}
  \label{f:mrules}
\end{figure}

\begin{lemma}[Decompose contraction]\label{l:med}
   The $\ctr$-rule is derivable for $\{\fctr,\med\}$.
\end{lemma}
\begin{proof}
  By induction the depth of a sequent which is contracted, we show the
  inductive step:
\[
\vlderivation{
\vlin{\ctr}{}{\Gamma\cons{A_1,\dots,A_m,[\Delta_1],\dots,[\Delta_n]}}{
\vlhy{\Gamma\cons{A_1,\dots,A_m,[\Delta_1],\dots,[\Delta_n], A_1,\dots,A_m,[\Delta_1],\dots,[\Delta_n]}}}
}
\]
\[
\leadsto\qquad
\vlderivation{
\vlin{\fctr^m}{}{\Gamma\cons{A_1,\dots,A_m,[\Delta_1],\dots,[\Delta_n]}}{
\vlin{\ctr^n}{}{\Gamma\cons{A_1,\dots,A_m,A_1,\dots,A_m,[\Delta_1],\dots,[\Delta_n]}}{
\vlin{\med^n }{}{\Gamma\cons{A_1,\dots,A_m,A_1,\dots,A_m,[\Delta_1,\Delta_1],\dots,[\Delta_n,\Delta_n]}}{
\vlhy{\Gamma\cons{A_1,\dots,A_m,[\Delta_1],\dots,[\Delta_n], 
A_1,\dots,A_m,[\Delta_1],\dots,[\Delta_n]}}}}}
}
\]

\end{proof}

\begin{lemma}[Weakening and necessitation admissibility for $\Km$]
\label{l:wkkminus}
\nl
  Let $\X \subseteq \modax$.  The $\wk$-rule and the
  $\nec$-rule are depth- and cut-rank-preserving admissible for
  $\Km+\str{\X}$.
\end{lemma}

\begin{lemma}[From $\mcut$ to $\cut$]\label{l:mcut}
  The rule $\mcut_r$ is derivable for $\{\cut_r,\wk\}$.
\end{lemma}
\begin{proof}
  We define the rule $\mcut_r^{m,n}$ with $m,n>0$ as 
\[
\vliinf{}
{}
{\Gamma\cons{\emptyset}}
{\Gamma\cons{\overbrace{A,\dots,A}^{m-\text{times}}}}
{\Gamma\cons{\overbrace{\neg A,\dots,\neg A}^{n-\text{times}}}}
\qquad ,
\]
and show that rule derivable for $\{\cut_r,\wk\}$ by induction on $m+n$.  The
case for $m=n=1$ is trivial, for $m>1$ and $n=1$ we replace
\[
\vliinf{\mcut_r^{m,1}}
{}
{\Gamma\cons{\emptyset}}
{\Gamma\cons{A,\dots,A}}
{\Gamma\cons{\neg A}}
\]
by
\[
\vlderivation{
\vliin{\cut_r}{}{\Gamma\cons{\emptyset}}
{\vliin{\mcut_r^{m-1,1}}
{}
{\Gamma\cons{A}}
{\vlhy{\Gamma\cons{A,\dots,A}}}
{\vlin{\wk}{}{\Gamma\cons{\neg A, A}}{\vlhy{\Gamma\cons{\neg A}}}}}
{\vlhy{\Gamma\cons{\neg A}}}
}
\]
and apply the induction hypothesis, and for $m,n>1$ we replace
\[
\vliinf{\mcut_r^{m,n}}
{}
{\Gamma\cons{\emptyset}}
{\Gamma\cons{A,\dots,A}}
{\Gamma\cons{\neg A,\dots,\neg A}}
\]
by
\[
\vlderivation{
\vliin{\cut_r}{}{\Gamma\cons{\emptyset}}
{\vliin{\mcut_r^{m-1,n}}
{}
{\Gamma\cons{A}}
{\vlhy{\Gamma\cons{A,\dots,A}}}
{\vlin{\wk}{}{\Gamma\cons{\neg A,\dots,\neg A, A}}{\vlhy{\Gamma\cons{\neg
A,\dots,\neg A}}}}}
{\vliin{\mcut_r^{m,n-1}}
{}
{\Gamma\cons{\neg A}}
{\vlin{\wk}{}{\Gamma\cons{A,\dots, A, \neg A}}{\vlhy{\Gamma\cons{A,\dots, A}}}}
{\vlhy{\Gamma\cons{\neg A,\dots,\neg A}}}}
}
\]
and apply the induction hypothesis twice. 

\end{proof}

\begin{lemma}[Push down contraction]\label{l:cdown}
  Let $\X \subseteq \modaxnod$. Given a proof as shown on the
  left, with $\rho$ a
  single-premise-rule from $\Km + \str{\X}+\wk$, there is a proof as
  shown on the right, with $|\D'|\leq |\D|$:
\[
\vlderivation{
 \vlin{\rho}{}{\Gamma}{
 \vlde{\D}{\fctr}{\Gamma_1}{
 \vltrnote{\PP}{\Gamma_2}{
\vlhy {}}{\vlhy {\qquad}}{\vlhy {}}{\Km+\str{\X}+\mcut+\wk}}}
 }\qquad \leadsto \qquad
\vlderivation{
 \vlde{\D'}{\fctr}{\Gamma}{
 \vltrnote{\PP'}{\Gamma_3}{
\vlhy {}}{\vlhy {\qquad}}{\vlhy {}}{\Km+\str{\X}+\mcut+\wk } }
 }
\qquad .
\]  

\end{lemma}
\begin{proof}
  By induction on the length of $\D$ and a case analysis on $\rho$.
  Most cases are trivial. We show the two interesting ones. For
  $\rho=\vlor$ and $\rho={\sf k}$ we apply the following transformations:
\[
\vlderivation{
\vlin{\vlor}{}{\Gamma\cons{A\vlor B}}{
\vlin{\fctr}{}{\Gamma\cons{A,B}}{
\vlhy{\Gamma\cons{A,A,B}}}}
}
\qquad \leadsto\qquad
\vlderivation{
\vlin{\fctr}{}{\Gamma\cons{A\vlor B}}{
\vlin{\vlor^2}{}{\Gamma\cons{A\vlor B ,A \vlor B}}{
\vlin{\wk}{}{\Gamma\cons{A, B ,A , B}}{
\vlhy{\Gamma\cons{A,A,B}}}}}
}
\]

\[
\vlderivation{
\vlin{{\sf k}}{}{\Gamma\cons{\Diamond A, [\Delta]}}{
\vlin{\fctr}{}{\Gamma\cons{[A, \Delta] }}{
\vlhy{\Gamma\cons{[A,A,\Delta]}}}}
}
\qquad \leadsto\qquad
\vlderivation{
\vlin{\fctr}{}{\Gamma\cons{\Diamond A, [\Delta]}}{
\vlin{{\sf k}^2}{}{\Gamma\cons{\Diamond A, \Diamond A, [\Delta]}}{
\vlhy{\Gamma\cons{[ A,A,\Delta]}}}}
}
\qquad ,
\]
and in each case we apply the induction hypothesis twice.

\end{proof}

\begin{proposition}[Push down contraction]\label{p:cdown}
Given a proof as shown on the left, there is a proof as shown on the right:
\[
\vlderivation{
 \vltrnote{\PP}{\Gamma}{
\vlhy {}}{\vlhy {\qquad}}{\vlhy {}}{\Kc + \str{\X} + \cut}}
\qquad\leadsto\qquad
\vlderivation{
 \vlde{}{\fctr}{\Gamma}{
 \vltrnote{\PP'}{\Gamma'}{
\vlhy {}}{\vlhy {\qquad}}{\vlhy {}}{\Km + \str{\X} + \cut}}}
\qquad .
\]  

\end{proposition}
\begin{proof}
  We first prove the claim that for each proof as shown on the left
  there is a proof as shown on the right:
\[
\vlderivation{
 \vltrnote{\PP_1}{\Gamma}{
\vlhy {}}{\vlhy {\qquad}}{\vlhy {}}{\Km + \str{\X} + \cut + \fctr}}
\qquad\leadsto\qquad
\vlderivation{
 \vlde{}{\fctr}{\Gamma}{
 \vltrnote{\PP_1'}{\Gamma'}{
\vlhy {}}{\vlhy {\qquad}}{\vlhy {}}{\Km + \str{\X} + \mcut + \wk}}}
\qquad ,
\]
The proof of the claim is by induction on the depth of $\PP_1$, using
Lemma~\ref{l:cdown} (Push down contraction).  The proof of our
proposition is as follows: by using Lemma~\ref{l:med} (Decompose
contraction) we obtain a proof in $\Km+\str{\X}+\cut+\fctr$, we
apply our claim, then we use Lemma~\ref{l:mcut} (From $\mcut$ to
$\cut$), to replace $\mcut$, starting with the top-most instances. Finally
we remove weakening using weakening admissibility.  
\end{proof}

It turns out that during the proof of cut-elimination for some system
$\Kc+\str{\X}$ some rules may be introduced that are not in $\str{\X}$
but that logically follow from $\X$. These additional rule instances will
then be removed from the proof after cut-elimination.

\begin{definition}[$\X^+$]
  Given some $\X \subseteq \modax$ we define
  \[
  \X^+ =
  \begin{cases}
    \X \cup \{{\sf 4}\} & \text{if }\{{\sf t,5}\} \subseteq \X \text{
      or
    }\{{\sf b,5}\} \subseteq \X\\
    \X \cup \{{\sf 5}\} &  \text{if }\{{\sf b,4}\} \subseteq \X \\
    \X & \text{otherwise} \quad ,\\
  \end{cases}
  \]
  and likewise for $\dia{\X}$ and $\str{\X}$.
\end{definition}

This definition matches the semantical notion of \emph{45-closed}
that we defined earlier: 

\begin{fact}[$\X^+$ is 45-closure of $\X$]
  If $\X\subseteq\modax$ then $\X^+$ is the least set which contains
  $\X$ and is 45-closed.
\end{fact}

The following lemma ensures that, after we have eliminated cut, we can
indeed remove the additional rules in $\X^+-\X$.

\begin{lemma}[From $\X^+$ to $\X$]\hspace{0mm}\\\label{l:xplustox}
(i)  The $\strfour$-rule is derivable for $\{\strt,\strfive,\nec\}$.\\
(ii)  The $\strfour$-rule is derivable for
$\{\strb,\strfive,\nec\}$.\\
(iii) The $\strfive$-rule is derivable for $\{\strb,\strfour,\wk\}$.\\
\end{lemma}
\begin{proof}
  For (i) notice that the $\strfour$-rule is a special case of the
  $\strfive$-rule unless $\Gamma\context$ has depth zero, and thus
  $\Gamma\context=\Lambda, \context$. In that case we have:
\[
\vlderivation{
\vlin{\strfour}{}{\Lambda, [[\Delta], \Sigma]}{\vlhy{\Lambda, [\Delta], [\Sigma]}}
}
\qquad\leadsto\qquad
\vlderivation{
\vlin{\strt}{}{\Lambda, [[\Delta], \Sigma]}{
\vlin{\strfive}{}{[\Lambda, [[\Delta], \Sigma]]}{
\vlin{\nec}{}{[\Lambda,[\Delta], [ \Sigma]]}{
\vlhy{\Lambda, [\Delta], [\Sigma]}}}}
}
\qquad .
\]
For (ii) we again have to consider only the case where
$\Gamma\context=\Lambda, \context$:
\[
\vlderivation{
\vlin{\strfour}{}{\Lambda, [[\Delta], \Sigma]}{\vlhy{\Lambda, [\Delta], [\Sigma]}}
}
\qquad\leadsto\qquad
\vlderivation{
\vlin{\strb}{}{\Lambda, [[\Delta], \Sigma]}{
\vlin{\strb}{}{ [[\Lambda], [\Delta], \Sigma]}{
\vlin{\strfive}{}{ [[\Lambda, [\Sigma]], [\Delta]]}{
\vlin{\nec^2}{}{ [[\Lambda, [\Delta], [\Sigma]]]}{
\vlhy{ \Lambda, [\Delta], [\Sigma]}}}}}
}
\]
For (iii) notice that a sequent has a tree structure and that, seen
upwards, the $\strfive$-rule allows to move a boxed sequent $[\Delta]$
to any position in that tree, but not to the root. To move a boxed
sequent to any position in the tree it is enough if we are both able
to move it a) from a given node the parent of this node and b) to move
it from a given node to any child of that node. Point a) is just the
$\strfour$-rule and point b) is as follows:
\[
\vlderivation{
\vlin{\strb}{}{\Gamma\cons{[\Lambda],[\Delta]}}{
\vlin{\strfour}{}{\Gamma\cons{[\Lambda,[[\Delta]]]}}{
\vlin{\wk}{}{\Gamma\cons{[\Lambda,[\emptyset],[\Delta]]}}{
\vlhy{\Gamma\cons{[\Lambda,[\Delta]]}}}}}
}\qquad .
\]
\end{proof}

We are now preparing for the reduction lemma, which we prove as usual
by pushing the cut rule upwards. In general we cannot push the cut
above a modal structural rule, so we push it upwards together with the
cut.  The interesting case occurs once this conglomerate of cut and
modal structural rules needs to be pushed above the $\diak$-rule. Then
we have to permute the $\diak$-rule down through the modal structural
rules to meet the cut. In the course of this permutation, the
$\diak$-rule might turn into another $\Diamond$-rule. The following
two lemmas take care of these permutations.

\begin{lemma}[Push down \diafour, \diafive]\label{l:45down}
  Let $\X \subseteq \modaxnod$ and $\rho \in (\dia{\X} \cap \{ \diafour,\diafive\})$.  Given a derivation as shown on the left, where
  $\rho$ applies to $\Diamond A$, there is a derivation as shown on
  the right, where all rules in $\D_3$ apply to the instance
of $\Diamond A$ shown, and where $|\D_2|\leq|\D_1|$:

\[
\vlderivation{
\vlde{\mathcal{D}_1 }{\str{\X}+\med}{\Delta\cons{\Diamond A}}{
\vlin{\rho }{}{\Gamma_1\cons{\Diamond A}}{
\vlhy{\Gamma\cons{\Diamond A}}}}}
\qquad\qquad\leadsto\qquad\qquad
\vlderivation{
\vlde{\mathcal{D}_3 }{(\dia{\X}^+ \cap \{\diafour,\diafive\})}{\Delta\cons{\Diamond A}}{
\vlde{\mathcal{D}_2 }{\str{\X}+\med}{\Gamma_2\cons{\Diamond A}}{
\vlhy{\Gamma\cons{\Diamond A}}}}}
\qquad .
\]
\end{lemma}
\begin{proof}
  The proof is by induction on the length of $\D_1$. We permute the
  instance of $\rho$ down and apply the induction hypothesis, possibly
  several times. We only show the non-trivial permutations. 

 \[
 \vlderivation {
 \vlin{\med}{}{ \Gamma\cons{\Diamond A, [\Delta,\Sigma]} } {
 \vlin{\diafour}{ }{\Gamma\cons{\Diamond A, [\Delta],[\Sigma]}}{
 \vlhy {\Gamma\cons{[\Diamond A, \Delta],[\Sigma]}}}}}
 \qquad\leadsto\qquad
 \vlderivation {
 \vlin{\diafour}{ }{\Gamma\cons{\Diamond A, [\Delta, \Sigma]}}{
 \vlin{\med}{ }{\Gamma\cons{[\Diamond A, \Delta, \Sigma]}}{
 \vlhy {\Gamma\cons{[\Diamond A, \Delta],[\Sigma]}}}}}
 \]

\[
\vlderivation {
\vlin{\strt}{}{ \Gamma\cons{\Diamond A, \Delta} } {
\vlin{\diafour}{ }{\Gamma\cons{\Diamond A, [\Delta]}}{
\vlhy {\Gamma\cons{[\Diamond A, \Delta]}}}}}
\qquad\leadsto\qquad
\vlderivation {
\vlin{\strt}{ }{\Gamma\cons{\Diamond A, \Delta}}{
\vlhy {\Gamma\cons{[\Diamond A, \Delta]}}}}
\]

\[
\vlderivation {
\vlin{\strb}{}{\Gamma\cons{[\Diamond A, \Delta],\Sigma} } {
\vlin{\diafour}{ }{\Gamma\cons{[\Diamond A, \Delta,[\Sigma]]}}{
\vlhy {\Gamma\cons{[\Delta,[\Diamond A, \Sigma]]}}}}}
\qquad\leadsto\qquad
\vlderivation {
\vlin{\diafive}{ }{\Gamma\cons{[\Diamond A, \Delta],\Sigma}}{
\vlin{\strb}{ }{\Gamma\cons{\Diamond A, [\Delta], \Sigma}}{
\vlhy {\Gamma\cons{[\Delta,[\Diamond A, \Sigma]]}}}}}
\]

\[
\vlderivation {
\vlin{\strfour}{}{\Gamma\cons{\Diamond A, [[\Delta], \Sigma]}   } {
\vlin{\diafour}{ }{\Gamma\cons{\Diamond A,[\Delta],[\Sigma] }   }{
\vlhy {\Gamma\cons{[\Diamond A,\Delta],[\Sigma] }}}}}
\qquad\leadsto\qquad
\vlderivation {
\vlin{\diafour}{ }{\Gamma\cons{\Diamond A, [[\Delta], \Sigma]    }}{
\vlin{\diafour}{ }{ \Gamma\cons{[\Diamond A,[\Delta], \Sigma]    } }{
\vlin{\strfour}{ }{ \Gamma\cons{[[\Diamond A,\Delta], \Sigma]}    }{
\vlhy {\Gamma\cons{[\Diamond A,\Delta],[\Sigma]}    }}}}}
\]

\[
\vlderivation { \vlin{\strfive}{}{\Gamma\cons{\Diamond
      A}\cons{[\Delta]}} { \vlin{\diafour}{ }{\Gamma\cons{\Diamond
        A,[\Delta]}\cons{\emptyset} }{ \vlhy {\Gamma\cons{[\Diamond
          A,\Delta]}\cons{\emptyset} }}}} \qquad\leadsto\qquad
\vlderivation { \vlin{\diafive}{ }{\Gamma\cons{\Diamond
      A}\cons{[\Delta]} }{ \vlin{\strfive}{
    }{\Gamma\cons{\emptyset}\cons{[\Diamond A,\Delta]} }{ \vlhy
      {\Gamma\cons{[\Diamond A,\Delta]}\cons{\emptyset} }}}}
\]

Permuting down the $\diafive$-rule is trivial except over the
$\strt$-rule and the $\strb$-rule, and this is also trivial unless
the restriction on the depth of the context in the $\diafive$-rule
becomes relevant:

\[
\vlderivation{
\vlin{\strt}{}{\Gamma_1, \Diamond A, \Delta, \Gamma_2\cons{\emptyset}}{
\vlin{\diafive}{}{\Gamma_1, [\Diamond A, \Delta], \Gamma_2\cons{\emptyset}}{
\vlhy{\Gamma_1, [\Delta], \Gamma_2\cons{\Diamond A}}}}
}
\qquad \leadsto\qquad
\vlderivation{
\vlin{\diafour^*}{}{\Gamma_1, \Diamond A, \Delta, \Gamma_2\cons{\emptyset}}{
\vlin{\strt}{}{\Gamma_1,  \Delta, \Gamma_2\cons{\Diamond A}}{
\vlhy{\Gamma_1, [\Delta], \Gamma_2\cons{\Diamond A}}}}
}
\]

\[
\vlderivation{
\vlin{\strb}{}{[\Delta], \Sigma, \Diamond A, \Gamma\cons{\emptyset}}{
\vlin{\diafive}{}{[\Delta, [\Sigma, \Diamond A]], \Gamma\cons{\emptyset}}{
\vlhy{[\Delta, [\Sigma]], \Gamma\cons{\Diamond A}}}}
}
\qquad \leadsto\qquad
\vlderivation{
\vlin{\diafour^*}{}{[\Delta], \Sigma, \Diamond A, \Gamma\cons{\emptyset} }{
\vlin{\strb}{}{[\Delta], \Sigma, \Gamma\cons{\Diamond A}}{
\vlhy{[\Delta, [\Sigma]], \Gamma\cons{\Diamond A}}}}
}
\]

\end{proof}

\begin{lemma}[Push down \diak, \diat, \diab]\label{l:ktbdown}
  Let $\X \subseteq \modaxnod$ and let $\rho={\diak}$ or $\rho
  \in (\dia{\X} \cap \{\diat ,\diab\})$.  Given a derivation as shown on
  the left, where $\rho$ applies to $\Diamond A$, there is a
  derivation as shown on the right, with $\sigma ={\diak}$ or $\sigma
  \in (\dia{\X} \cap \{\diat, \diab\})$, where all rules in $\D_3$ apply
  to the instance of $\Diamond A$ shown, and where $|\D_2|\leq|\D_1|$:

\[
\vlderivation{
\vlde{\mathcal{D}_1 }{\str{\X}+\med}{\Delta\cons{\Diamond A}}{
\vlin{\rho }{}{\Gamma_1\cons{\Diamond A}}{
\vlhy{\Gamma\cons{A}}}}}
\qquad\qquad\leadsto\qquad\qquad
\vlderivation{
\vlde{\mathcal{D}_3 }{(\dia{\X}^+ \cap \{\diafour,\diafive\})}{\Delta\cons{\Diamond A}}{
\vlin{\sigma}{}{\Gamma_2\cons{\Diamond A}}{
\vlde{\mathcal{D}_2 }{\str{\X}+\med}{\Gamma_3\cons{A}}{
\vlhy{\Gamma\cons{A}}}}}}
\qquad .
\]
\end{lemma}

\begin{proof}
  The proof is by induction on the length of $\D_1$. We permute the
  instance of $\rho$ down and apply Lemma~\ref{l:45down} (Push down
  {\diafour, \diafive}) and/or the induction hypothesis. We only show the
  non-trivial permutations.

\[
\vlderivation {
\vlin{\strt}{}{\Gamma\cons{\Diamond A, \Delta} } {
\vlin{{\diak}}{ }{\Gamma\cons{\Diamond A, [\Delta]}}{
\vlhy {\Gamma\cons{[A, \Delta]}}}}}
\qquad\leadsto\qquad
\vlderivation {
\vlin{\diat}{}{\Gamma\cons{\Diamond A, \Delta}}{
\vlin{\strt}{ }{\Gamma\cons{A, \Delta}}{
\vlhy {\Gamma\cons{[A, \Delta]}}}}}
\]

\[
\vlderivation {
\vlin{\strb}{}{\Gamma\cons{[\Diamond A, \Delta],\Sigma} } {
\vlin{{\diak}}{ }{\Gamma\cons{[\Diamond A, \Delta,[\Sigma]]}}{
\vlhy {\Gamma\cons{[\Delta,[A, \Sigma]]}}}}}
\qquad\leadsto\qquad
\vlderivation {
\vlin{\diab}{ }{\Gamma\cons{[\Diamond A, \Delta],\Sigma}}{
\vlin{\strb}{ }{\Gamma\cons{A, [\Delta], \Sigma}}{
\vlhy {\Gamma\cons{[\Delta,[A, \Sigma]]}}}}}
\]

\[
\vlderivation {
\vlin{\strfour}{}{\Gamma\cons{\Diamond A, [[\Delta], \Sigma] }  } {
\vlin{{\diak}}{ }{\Gamma\cons{\Diamond A,[\Delta],[\Sigma]  }  }{
\vlhy {\Gamma\cons{[A,\Delta],[\Sigma]} }}}}
\qquad\leadsto\qquad
\vlderivation {
\vlin{\diafour}{ }{\Gamma\cons{\Diamond A, [[\Delta], \Sigma]   } }{
\vlin{{\diak}}{ }{\Gamma\cons{ [\Diamond A,[\Delta], \Sigma]   }  }{
\vlin{\strfour}{ }{\Gamma\cons{ [[ A,\Delta], \Sigma]   } }{
\vlhy {\Gamma\cons{[A,\Delta],[\Sigma]  }  }}}}}
\]

\[
\vlderivation { \vlin{\strfive}{}{\Gamma\cons{\Diamond
      A}\cons{[\Delta]}} { 
\vlin{{\diak}}{ }{\Gamma\cons{\Diamond
        A,[\Delta]}\cons{\emptyset} }{ 
\vlhy {\Gamma\cons{[A,\Delta]}\cons{\emptyset} }}}} 
\qquad\leadsto\qquad
\vlderivation { 
\vlin{\diafive}{ }{\Gamma\cons{\Diamond A}\cons{[\Delta]} }{ 
\vlin{{\diak}}{  }{\Gamma\cons{\emptyset}\cons{\Diamond A, [\Delta]} }{
\vlin{\strfive}{  }{\Gamma\cons{\emptyset}\cons{[A,\Delta]} }{ \vlhy
      {\Gamma\cons{[A,\Delta]}\cons{\emptyset} }}}}}
\]

The cases for $\rho=\diat$ are trivial.

\[
\vlderivation{
\vlin{\strt}{}{\Gamma\cons{\Delta, \Diamond A}}{
\vlin{\diab}{}{\Gamma\cons{[\Delta, \Diamond A]}}{
\vlhy{\Gamma\cons{[\Delta], A}}}}
}
\qquad\leadsto\qquad
\vlderivation{
\vlin{\diat}{}{\Gamma\cons{\Delta, \Diamond A}}{
\vlin{\strt}{}{\Gamma\cons{\Delta, A}}{
\vlhy{\Gamma\cons{[\Delta, A]}}}}
}
\]

\[
\vlderivation{
\vlin{\strb}{}{\Gamma\cons{[\Sigma],\Delta,\Diamond A}}{
\vlin{\diab}{}{\Gamma\cons{[\Sigma,[\Delta,\Diamond A]]}}{
\vlhy{\Gamma\cons{[\Sigma,[\Delta], A]}}}}
}
\qquad\leadsto\qquad
\vlderivation{
\vlin{{\diak}}{}{\Gamma\cons{[\Sigma],\Delta,\Diamond A}}{
\vlin{\strb}{}{\Gamma\cons{[\Sigma,A],\Delta}}{
\vlhy{\Gamma\cons{[\Sigma,[\Delta], A]}}}}
}
\]

\[
\vlderivation{
\vlin{\strfour}{}{ \Gamma\cons{[[\Delta, \Diamond A], \Sigma]}}{
\vlin{\diab}{}{\Gamma\cons{[\Delta, \Diamond A],[\Sigma]}}{
\vlhy{\Gamma\cons{[\Delta], A,[\Sigma]}}}}
}
\qquad\leadsto\qquad
\vlderivation{
\vlin{\diafive}{}{\Gamma\cons{[[\Delta, \Diamond A], \Sigma]}}{
\vlin{\diab}{}{\Gamma\cons{[[\Delta], \Diamond A, \Sigma]}}{
\vlin{\strfour}{}{\Gamma\cons{[[\Delta], \Sigma], A } }{
\vlhy{\Gamma\cons{[\Delta], A,[\Sigma] }}}}}
}
\]

For permuting down over the $\strfive$-rule, in the only non-trivial
case, notice that the context has to be of the form shown because of
the restriction of context depth in the $\strfive$-rule:

\[
\vlderivation{
\vlin{\strfive}{}{ \Gamma\cons{[\Delta, \Diamond A]}\cons{[\Sigma,\emptyset]}}{
\vlin{\diab}{}{  \Gamma\cons{\emptyset}\cons{[\Sigma,[\Delta, \Diamond A]]} }{
\vlhy{\Gamma\cons{\emptyset}\cons{[\Sigma,[\Delta], A]}}}}
}
\qquad\leadsto\qquad
 \vlderivation{
 \vlin{\diafive}{}{\Gamma\cons{[\Delta, \Diamond A]}\cons{[\Sigma,\emptyset] }}{
 \vlin{{\diak}}{}{\Gamma\cons{[\Delta]}\cons{\Diamond A, [\Sigma,\emptyset]}   }{
 \vlin{\strfive}{}{\Gamma\cons{[\Delta]}\cons{[A,\Sigma] }   }{
 \vlhy{\Gamma\cons{\emptyset}\cons{[\Sigma,[\Delta], A]} }}}}
 }
\]

\end{proof}

Once a $\Diamond$-rule has been permuted down through the structural
modal rules to meet the cut, we want to build a new derivation
with a lower cut rank. This is not possible when this $\Diamond$-rule
is either $\diafour$ or $\diafive$ since these rules do not decrease
the size of the main formula, when seen upwards. The solution is to
``reflect'' them at the cut and incorporate them in the structural
rules that are pushed up together with the cut.

\begin{lemma}[Reflect \diafour, \diafive]\label{l:45reflect}
  Let $\X \subseteq \{{\sf 4,5}\}$.  Given a derivation as shown on the
  left, where all rules in $\D$ apply to the instance of $\Diamond A$
  shown, then for each sequent $\Delta$ there is a derivation as shown on
  the right:
\[
\vlderivation{
{\vlde{\D}{\dia{\X}}{\Gamma\cons{\emptyset}\cons{\Diamond{A}}}
{\vlhy{\Gamma\cons{\Diamond A}\cons{\emptyset}}}}
}
\qquad\qquad\leadsto\qquad\qquad
\vlderivation{
{\vlde{\D'}{\str{\X}}
{\Gamma\cons{[\Delta]}\cons{\emptyset}}
{\vlhy{\Gamma\cons{\emptyset}\cons{[\Delta]}}}}
}
\quad .
\]
\end{lemma}
\begin{proof}
  By induction on the length of $\D$.
\end{proof}

We are now ready to prove the reduction lemma.

\begin{lemma}[Reduction Lemma]
  Let $\X \subseteq \modaxnod$. Given a proof as shown on the
  left, with $\PP_1$ and $\PP_2$ in $\Km + \str{\X} + \cut_r$,
  then there is a proof $\PP$ in $\Km + \str{\X}^+ + \cut_r$ as
  shown on the right:
\[
\vlderivation{
\vliin{\cut_{r+1}}{}{\Gamma\cons{\emptyset}}
{\vlde{}{\str{\X}+\med}{\Gamma\cons{A}}{\vltr{\PP_1}{\Gamma_1\cons{A} }{
\vlhy {}}{\vlhy {\qquad}}{\vlhy {}}}}
{\vlde{}{\str{\X}+\med}{\Gamma\cons{\neg A}}{\vltr{\PP_2}{\Gamma_2\cons{\neg A} }{
\vlhy {}}{\vlhy {\qquad}}{\vlhy {}}}}
}
\qquad \leadsto\qquad
\vlderivation{
\vltr{\PP}{\Gamma\cons{\emptyset} }{
\vlhy {}}{\vlhy {\qquad}}{\vlhy {}}}
\qquad .
\]  
\end{lemma}
\begin{proof}
  As usual, by an induction on $|\PP_1|+|\PP_2|$ and a case analysis
  on the lowermost rules in $\PP_1$ and $\PP_2$. We only show the most
  complicated case, in which we cut a box introduced by the
  $\mboxx$-rule against a diamond introduced by $\sf k$-rule. All
  other cases are much simpler. We have
\[
\vlderivation{ \vliin{\cut_{r+1}}{}{\Gamma\cons{\emptyset}}
  {\vlde{}{\str{\X}+\med}{\Gamma\cons{\Box B}}
    {\vlin{\mboxx}{}{\Gamma_1\cons{\Box
          B}}{\vlhy{\Gamma_1\cons{[B,\dots,B]}}}}}
  {\vlde{}{\str{\X}+\med}{\Gamma\cons{\Diamond \neg B}} {\vlin{{\sf
          k}}{}{\Gamma_2'\cons{\Diamond \neg B,
          [\Delta]}}{\vlhy{\Gamma_2'\cons{[\neg B, \Delta]}
        }}}} }
\]
In the left subderivation we permute down the instance of $\mboxx$ and
on the right subderivation we apply Lemma~\ref{l:ktbdown} (Push {\sf
  ktb} down) in order to obtain the following derivation, where
$\Gamma\context=\Gamma\context\cons{\emptyset}$. Note that the second
hole in the binary context marks the position to which the $\Diamond
\neg B$ is moved:
\[
\vlderivation{ \vliin{\cut_{r+1}}{}{\Gamma\cons{\emptyset}\cons{\emptyset}}
  {\vlin{\mboxx}{}{\Gamma\cons{\Box B}\cons{\emptyset}}
    {\vlde{}{\str{\X}+\med}{\Gamma\cons{[B,\dots,B]}\cons{\emptyset}}{\vlhy{\Gamma_1\cons{[B,\dots,B]}}}}}
  {\vlde{}{(\X^+ \cap \{{\sf 4,5}\})^\Diamond}{\Gamma\cons{\Diamond \neg
B}\cons{\emptyset}} 
{\vlin{\sigma}{}{\Gamma\cons{\emptyset}\cons{\Diamond \neg B}}{\vlde{}{\str{\X}+\med}
{\Gamma_3\cons{\neg B}}{\vlhy{\Gamma_2'\cons{[\neg B, \Delta]}
        }}}} }}
\]
By using Lemma~\ref{l:45reflect} (Reflect {\sf 45}) we obtain a
derivation $\D$ and build:
\[
\vlderivation{ 
\vliin{\cut_{r+1}}{}{\Gamma\cons{\emptyset}\cons{\emptyset}}
  {\vlin{\mboxx}{}{\Gamma\cons{\emptyset}\cons{\Box B}}{
    \vlde{\D}{(\X^+ \cap \{{\sf 4,5}\})^{\cdot}}{\Gamma\cons{\emptyset}\cons{[B,\dots,B]}}{
    \vlde{}{\str{\X}+\med}{\Gamma\cons{[B,\dots,B]}\cons{\emptyset}}{
\vlhy{\Gamma_1\cons{[B,\dots,B]}}}}}}
{\vlin{\sigma}{}{\Gamma\cons{\emptyset}\cons{\Diamond \neg B}}
{\vlde{}{\str{\X}+\med}
{\Gamma_3\cons{\neg B}}{\vlhy{\Gamma_2'\cons{[\neg B, \Delta]}
        }}}} 
}
\qquad .
\]
We now consider the three possible cases for $\sigma \in \{{\sf
  k},\diat,\diab\}$ and apply one of the following transformations
to the relevant part of the proof:
\[
\vlderivation{ 
\vliin{\cut_{r+1}}{}{\Sigma\cons{[\Delta]}}
  {\vlin{\mboxx}{}{\Sigma\cons{\Box B, [\Delta]}}{
\vlhy{\Sigma\cons{[B,\dots,B],[\Delta]}}}}
{\vlin{{\sf k}}{}{\Sigma\cons{\Diamond \neg B,[\Delta]}}
{\vlhy{\Sigma \cons{[\neg B, \Delta]}
        }}} 
}
\quad\leadsto\quad
\vlderivation{ 
\vliin{\mcut_{r}}{}{\Sigma\cons{[\Delta]}}
  {\vlin{\med}{}{\Sigma\cons{[B,\dots,B,\Delta]}}{
\vlhy{\Sigma\cons{[B,\dots,B],[\Delta]}}}}
{\vlhy{\Sigma \cons{[\neg B, \Delta]}
        }}
}
\]

\[
\vlderivation{ 
\vliin{\cut_{r+1}}{}{\Sigma\cons{\emptyset}}
  {\vlin{\mboxx}{}{\Sigma\cons{\Box B}}{
\vlhy{\Sigma\cons{[B,\dots,B]}}}}
{\vlin{\diat}{}{\Sigma\cons{\Diamond \neg B}}
{\vlhy{\Sigma \cons{\neg B}
        }}} 
}
\quad\leadsto\quad
\vlderivation{ 
\vliin{\mcut_{r}}{}{\Sigma\cons{\emptyset}}
  {\vlin{\strt}{}{\Sigma\cons{B,\dots,B}}{
\vlhy{\Sigma\cons{[B,\dots,B]}}}}
{\vlhy{\Sigma \cons{\neg B}
        }} 
}
\] 

\[
\vlderivation{ 
\vliin{\cut_{r+1}}{}{\Sigma\cons{[\Delta]}}
  {\vlin{\mboxx}{}{\Sigma\cons{[\Box B, \Delta]}}{
\vlhy{\Sigma\cons{[[B,\dots,B],\Delta]}}}}
{\vlin{\diab}{}{\Sigma\cons{[\Diamond \neg B,\Delta]}}
{\vlhy{\Sigma \cons{\neg B, [\Delta]}
        }}} 
}
\quad\leadsto\quad
\vlderivation{ 
\vliin{\mcut_{r}}{}{\Sigma\cons{[\Delta]}}
  {\vlin{\strb}{}{\Sigma\cons{B,\dots,B, \Delta]}}{
\vlhy{\Sigma\cons{[[B,\dots,B],\Delta]}}}}
{\vlhy{\Sigma \cons{\neg B, [\Delta]}
        }}
}
\]
We then eliminate $\mcut$ by using Lemma~\ref{l:mcut} (From $\mcut$ to
$\cut$) and weakening admissibility.  
\end{proof}

\begin{proposition}[Cut-elimination for $\Km$]\label{p:cekminus}
  Let $\X \subseteq \modaxnod$. If $\Km + \str{\X} + \cut
\vdash \Gamma$ then $\Km + \str{\X}^+\vdash \Gamma$. 
\end{proposition}
\begin{proof}
  We first prove the claim: If $\Km + \str{\X} + \cut_{r+1} \vdash
  \Gamma$ then $\Km + \str{\X}^+ + \cut_r \vdash \Gamma$. The
  claim is proved by induction on the depth of the given proof, using
  the reduction lemma. Our proposition then follows from an induction
  on the cut rank of the given proof, using the claim.  
\end{proof}

Finally, we can prove cut-elimination for the systems $\Kc+\str{\X}$.

\begin{proof}[Proof of Theorem~\ref{t:cestr} (Cut-elimination)]
  We first prove the theorem for the cases where ${\sf d} \notin \X$.
  The transformation (i) is by Proposition~\ref{p:cdown} (Push down
  contraction), the transformation (ii) is
  Proposition~\ref{p:cekminus} (Cut-elimination for $\Km$), and
  transformation (iii) is by Lemma~\ref{l:xplustox} (From $\X^+$ to
  $\X$) and weakening and necessitation admissibility.
\[
\vlderivation{
 \vltrnote{\PP_1}{\Gamma}{
\vlhy {}}{\vlhy {\qquad}}{\vlhy {}}{\Kc + \str{\X} + \cut}}
\quad\stackrel{(i)}{\leadsto}\quad
\vlderivation{
 \vlde{}{\fctr}{\Gamma}{
 \vltrnote{\PP_2}{\Gamma'}{
\vlhy {}}{\vlhy {\qquad}}{\vlhy {}}{\Km + \str{\X} + \cut}}}
\quad\stackrel{(ii)}{\leadsto}\quad
\vlderivation{
 \vlde{}{\fctr}{\Gamma}{
 \vltrnote{\PP_3}{\Gamma'}{
\vlhy {}}{\vlhy {\qquad}}{\vlhy {}}{\Km + \str{\X}^+}}}
\]\[
\quad\stackrel{(iii)}{\leadsto}\quad
\vlderivation{
 \vltrnote{\PP_4}{\Gamma}{
\vlhy {}}{\vlhy {\qquad}}{\vlhy {}}{\Kc + \str{\X}}}
\quad .
\]
In the cases where ${\sf d} \in \X$ we first apply Lemma~\ref{l:ddown}
(Push down seriality) and then proceed the same way with the upper part
of the proof.  
\end{proof}

\section{Relation to 
Deep Inference}

\newcommand{\downar}{\mathord{\downarrow}}
\newcommand{\upar}{\mathord{\uparrow}}

\emph{Deep inference} is a proof-theoretic formalism introduced by
Guglielmi \cite{Gugl07} where inference rules are term rewriting rules
which work on formulas and where derivations are just reduction
sequences from one formula to another. Some deep inference systems for
modal logic have been studied by Hein, Stewart and Stouppa
\cite{HeinStew05,StSt,Sto06}.

Stewart and Stouppa give certain deep inference rules for the modal
axioms in their paper \cite{StSt} and conjecture that all combinations
yield cut-free systems that are complete for the corresponding frame
conditions (Conjecture~11 in \cite{StSt}). They prove their conjecture
just for some modal logics, namely K, KD, KT, S4 and S5, and in all
cases their method is embedding a cut-free (hyper-)sequent system.
They do not provide cut-free deep inference systems for the other 10
logics of the cube. Also, their method does not extend to logics for
which there is no known cut-free (hyper-)sequent system, such as $\sf
KB$ and $\sf K5$.

In this section we see cut-free deep inference systems for these modal logics.
Nested sequent systems can be easily embedded into corresponding deep
inference systems and via this embedding we get complete and cut-free
deep inference systems for all the modal logics considered in this
chapter. In fact, we get two sets, one based on the nested sequent
systems with logical rules, and one based on the ones with structural
rules. However, this does not settle Stewart and Stouppa's
Conjecture~11, since our rules are different.

The embedding of nested sequent systems into corresponding deep
inference systems is trivial: essentially, all derivations on nested
sequents are special deep inference derivations where rules do not
apply deeply with respect to all connectives, but only with respect to
the comma (structural disjunction) and structural box. The reverse
direction, embedding deep inference into nested sequent calculus is
also easy, but requires cut.

In this section we extend our language of formulas by the constants
$\true$ for \emph{true} and $\false$ for \emph{false}. 

A deep inference rule is just a labelled rewrite rule as used in term
rewriting. An example is the following \emph{switch-down}-rule:
\[
\vlinf{{\sf s}\downar}{}{S\cons{(A \vlan B) \vlor C}}{S\cons{A \vlan (B
\vlor C)}} \rlap{\qquad ,}
\]
which in term rewriting would be written as
\[
{\sf s}\downar:\quad(A \vlan B) \vlor C \rightarrow A \vlan (B\vlor C) \rlap{\qquad .}
\]
There is a notational difference: in the deep inference rule the
context in which it can be applied is made explicit, in this case any
formula context $S\context$. A \emph{proof} of a formula is a
rewriting sequence starting from the constant $\true$ and ending with
that formula. For more explicit definitions and more discussion of
deep inference systems, see \cite{BruTh}.

\newcommand{\KS}{{\sf KS}}
\newcommand{\SKS}{{\sf SKS}}

A deep inference system for propositional logic is shown in
Figure~\ref{fig:propcos}. This particular system is similar to the
one given by Stra{\ss}burger in \cite{Stra09a} and slightly
weaker than the one originally given in \cite{BruTh} because it
replaces the equivalence rule by several explicit rules for for
commutativity, associativity and units (which together are weaker than
the equivalence rule). Let us call it system $\KS$ for the purpose of
this section. Systems for modal logics can be obtained from it by
adding rules from Figure~\ref{f:cos}. The \emph{cut} in deep inference
has the form 
\[
\vlinf{{\sf i}\upar}{}{S\cons{\false}}{S\cons{A \vlan \neg A}}
\rlap{\qquad .}
\]

Let an instance of ${\sf 5}\downar$ be an instance of either ${\sf
5a}\downar, {\sf 5b}\downar$ or ${\sf 5c}\downar$.  For a set $\X$ of rule
names append the symbol $\downar$ to each name to obtain $\X\downar$.
Let system ${\sf KSk}$ be system $\KS+\{\nec\downar,{\sf k}\downar, {\sf
  r}\downar\}$.  

\begin{proposition}[Nested sequent calculus into deep inference]\nl  For all
  $\X \subseteq \modax$ and
  sequents $\Gamma$ we have that:\\
  If $\,\K+\diac{\X} \,\vdash\, \Gamma\,$ then $\,{\sf KSk} +
  \X\downar
  \;\vdash\,\formula{\Gamma}$.\\
\end{proposition}
\begin{proof}
A routine induction on the depth of the proof and a straightforward
extension of a corresponding embedding for the propositional system as
given in \cite{BruTh}. Note that embedding the $\vlan$-rule requires
the ${\sf r}\downar$-rule.
\end{proof}

\begin{proposition}[Deep inference into nested sequent calculus]\nl For all
  $\X \subseteq \modax$ and formulas $A$ we have that:\\
  If $\,{\sf KSk} + \X\downar+ {\sf i}\upar\; \vdash A\,$ then
  $\,\K+\diac{\X}+\cut \,\vdash\, A$.
\end{proposition}
\begin{proof}
A routine induction on the length of the proof and a straightforward
extension of a corresponding embedding for the propositional system as
given in \cite{BruTh}. 
\end{proof}

These propositions, together with cut-elimination for our nested
sequent systems, trivially yields cut-elimination for the
corresponding deep inference systems. By the second proposition we
translate a deep inference proof with cuts into a nested sequent
calculus proof with cuts, eliminate the cuts, and translate back to
deep inference by the first proposition.

\begin{corollary}[Cut elimination for deep inference]\nl
  For all 45-closed $\X \subseteq \modax$ we have that if a formula is
  provable in system ${\sf KSk} + \X\downar + {\sf i}\upar$ then it
  is also provable in system ${\sf KSk} + \X\downar$.
\end{corollary}

\begin{figure}
\fbox{
\parbox{\ftextwidth}{
  \centering
\[
\vlinf{{\sf as}\downar}{}{S\cons{(A \vlor B) \vlor C}}{S\cons{A \vlor (B
\vlor C)}}
\qquad
\vlinf{{\sf co}\downar}{}{S\cons{B \vlor A}}{S\cons{A \vlor B}}
\qquad
\vlinf{\false\downar}{}{S\cons{A \vlor \false}}{S\cons{A}}
\qquad
\vlinf{\true\downar}{}{S\cons{A \vlan \true}}{S\cons{A}}
\]
\[
\vlinf{{\sf ai}\downar}{}{S\cons{a \vlor \neg a}}{S\cons{\sf t}}
\qquad
\vlinf{{\sf s}\downar}{}{S\cons{(A \vlan B) \vlor C}}{S\cons{A \vlan (B
\vlor C)}}
\qquad
\vlinf{{\sf c}\downar}{}{S\cons{A}}{S\cons{A \vlor A}}
\qquad
\vlinf{{\sf w}\downar}{}{S\cons{A}}{S\cons{\sf f}}
\]  
}}
  \caption{A deep inference system for propositional logic}
\label{fig:propcos}
\end{figure}

\begin{figure}
\fbox{
\parbox{\ftextwidth}{
  \centering
\[
\vlinf{\nec\downar}{}{S\cons{\Box \sf t}}{S\cons{\sf t}}
\qquad
\vlinf{{\sf k}\downar}{}{S\cons{\Box
A \vlor \Diamond B}}{S\cons{\Box(A\vlor B)}}
\qquad
\vlinf{{\sf r}\downar}{}{S\cons{\Box (A \vlan B)}}
{S\cons{\Box A \vlan  \Box B}}
\]
\[
\vlinf{{\sf d}\downar}{}{S\cons{\Diamond A}}{S\cons{\Box A}} \qquad
\vlinf{{\sf t}\downar}{}{S\cons{\Diamond A}}{S\cons{A}} \qquad
\vlinf{{\sf b}\downar}{}{S\cons{\Box(\Diamond A \vlor B)}}{S\cons{A \vlor
    \Box B}} \qquad 
\vlinf{{\sf 4}\downar}{}{S\cons{\Box A \vlor \Diamond
    B}}{S\cons{\Box(A \vlor \Diamond B)}}
\]
\[
\vlinf{{\sf 5a}\downar}{}{S\cons{\Box(\Diamond A \vlor
      B)}}{S\cons{\Diamond A \vlor \Box B}} \qquad 
\vlinf{{\sf 5b}\downar}{}{S\cons{\Box(\Diamond A\vlor B) \vlor \Box
    C}}{S\cons{\Box B \vlor \Box(\Diamond A \vlor C)}} \qquad 
\vlinf{{\sf 5c}\downar}{}{S\cons{\Box(A \vlor \Diamond B \vlor \Box
      C)}}{S\cons{\Box(A \vlor \Box ( \Diamond B \vlor C))}}
\]}}
\caption{Deep inference rules for modal logic}
\label{f:cos}
\end{figure}

A similar exercise will obtain cut-free and complete deep inference
systems from the nested sequent systems with structural modal rules.

\begin{remark}[for some systems the ${\sf r}\downar$-rule is admissible]
  Some of the deep inference systems are not minimal: for example in
  system {\sf KSk} the ${\sf r}\downar$-rule is admissible for ${\sf
    KSk}-{\sf r}\downar$. This can be seen by embedding the usual
  sequent system for $\K$, of which we show the case for the
  $\Box$-rule:
\[
\vlderivation{
\vlin{\Box}{}{\Box A, \Diamond B_1, \dots ,\Diamond B_n}{
 \vltrnote{\PP}{A, B_1, \dots, B_n}{
\vlhy {}}{\vlhy {\qquad}}{\vlhy {}}{}}}
\qquad\leadsto\qquad
\vlderivation{
\vlin{{\sf k}^n}{}{\Box A \vlor \Diamond B_1 \vlor \dots \vlor\Diamond B_n}{
 \vlde{\Box\PP'}{}{\Box(A \vlor B_1 \vlor \dots \vlor B_n)}{
\vlin{\nec\downar}{}{\Box\true}{\vlhy{\true}}}}}
\qquad ,
\]  
where $\PP'$ is the translation of $\PP$, $\Box\PP'$ is obtained by
adding a box to every formula in $\PP'$ and ${\sf k}^n$ denotes $n$
instances of the $\sf k$-rule. For some systems, however, the ${\sf
  r}\downar$-rule is not admissible. For example in system ${\sf
  KSk}+{\sf b}\downar$ the formula $\Box(a \vlor \Diamond\Diamond\neg
a) \vlan (b \vlor \Diamond\Diamond\neg b))$ is provable, but it is not
provable without ${\sf r}\downar$-rule.
\end{remark}

\section{Discussion}

We have seen how nested sequents allow us to give a systematic proof
theory for the modal logics of the cube. In fact, we have seen two
distinct proof-theories, one based on formalising modal axioms as
logical rules and one based on formalising modal axioms as structural
rules. The first option is closer to the ordinary sequent calculus and
allows for a straightforward terminating proof-search procedure, but
fails to be modular: not every possible combination of rules yields a
complete system for the corresponding logic. The second option yields
a modular set of systems, but the presence of structural rules
devalues the subformula property. In any case, we have seen that
generalising hypersequents to nested sequents yields cut-free systems
for more modal logics, so it leads to greater expressivity. It is
particularly pleasant that this extra generality does not come at the
cost of extra complexity, but in fact simplifies hypersequent systems:
the two kinds of context in hypersequent inference rules (sequent
context and hypersequent context) are merged into one. Our systems
with logical rules enjoy invertibility of all rules. This property
does not seem to be achievable in an ordinary sequent system for modal
logic.  In hypersequent systems it also does not seem to be achievable
in a non-trivial way (although one could of course trivially make
rules invertible by copying a component whenever a rule applies in
it).

{\bf Relation to the display calculus.} Nested sequents and display
sequents share the idea of simply allowing the connective $\Box$ as a
structural connective.  There are two crucial differences.  First,
display sequents also contain a structural connective for the
backward-looking modality. This is crucial for the \emph{display
  property} to hold, a central property of display calculi which
allows to single out a formula in order to apply a logical rule to it.
Since in our proof systems logical rules apply deeply inside nested
sequents, there is no need for a display property, and thus no need
for the backward-looking modality, so we can stay inside the modal
language.  The second difference is that in the display calculus one
has to use structural rules called \emph{display postulates} to move a
formula to the top in order to apply a logical rule to it. In nested
sequent systems one can apply the rule on the spot and thus has no
need for such structural rules. Nested sequents thus allow for
deductive systems with fewer rules and shorter derivations.

{\bf Relation to labelled systems.}  The main conceptual advantage of
a nested sequent over a labelled sequent is that it can be read as a
modal formula. Labelled sequents are more general than nested
sequents: they can form an arbitrary graph, while nested sequents are
always trees. A cut-free proof in nested sequents is thus in general a
more restricted, simpler object than a cut-free proof in labelled
sequents. I hope that this fact will help in using nested sequent
systems for interpolation proofs, for which labelled systems do not
seem to be well-suited. It should also be easy to embed cut-free
nested sequent systems into corresponding cut-free labelled sequent
systems, while the opposite is not true in general. I thus think of the
completeness of a nested sequent system as a stronger result than the
completeness of a corresponding labelled sequent system.  To get this
stronger result we had to work harder, for example in our completeness
proof for the systems with logical rules: we had to establish certain
properties of, say, the euclidean closure of a relation, which is not
needed for labelled systems. There, that relation is part of the proof
system and it is being closed under euclideanness by the appropriate
inference rule.  The extra work also shows in our cut-elimination
procedure: we had to show admissibility of certain rules in order to
push the cut over the rules for the frame properties.  This, again, is
not needed for labelled systems. There the rules for the frame
conditions do not affect the cut-elimination procedure at all.

{\bf Relation to tableau systems.} While the focus of tableau systems
is on giving decision procedures, our focus is on giving proof systems
which support proof-transformations, in particular cut-elimination.
This is more easily and more commonly done with local rules, so
in sequent systems instead of tableau systems. Nevertheless, there is
correspondence between tableau systems and sequent systems. For an
overview of modal tableau systems see the survey by Gor{\'e}
\cite{Gore99}. The tableau formalism which corresponds most closely to
nested sequents is the prefixed tableau formalism, due to Fitting
\cite{Fitting1983}. In particular, prefixes impose the same tree structure on
formulas that is imposed in a nested sequent. However, prefixed
tableaux are closer to the semantics. In particular they have
rules which are parametrised by an accessibility relation, which is a
marked difference from our inference rules.

Specific tableau rules which correspond to our inference rules have
also been studied before, namely by Castilho et.~al.~\cite{Cast97}.
Their systems are based on graphs rather than trees, but they have structural
rules which closely correspond to (some of) ours and \emph{propagation
  rules} which correspond to our $\Diamond$-rules.  A difference is
that propagation rules and structural rules are mixed in
\cite{Cast97}, while here we first treat systems purely consisting of
propagation- or $\Diamond$-rules in Chapter 2 and systems purely
consisting of structural rules in Chapter 3.

{\bf Future work.} Of course we would like to extend the range of
logics for which there are cut-free nested sequent systems. Candidates
are the set of modal logics formalised by so-called \emph{primitive}
axioms, which have been captured in the display calculus \cite{Wan98}.
At the same time, it is interesting to generate such systems
automatically, so it is our goal to devise 1) easily checkable criteria
on rules, which guarantee cut-elimination, and 2) a procedure which
turns modal axioms into rules which satisfy these criteria. Such a
generic cut-elimination procedure exists already for the display
calculus \cite{Wan98}. Recently, such a procedure has also been
proposed by Ciabattoni et al.\ for certain hypersequent systems
\cite{Ciab08}.

On the other hand, we would like to use nested sequent systems to
obtain results which are harder or cannot be obtained with other
proof-theoretic formalisms. Neither display calculus nor labelled
sequent calculus seem to allow us to prove interpolation results, for
example. Conservativity results are another interesting field. Here
the property of staying inside the modal language is useful. The
conservativity of tense logic over modal logic is an immediate
consequence of the completeness of a cut-free nested sequent system
for tense logic, as noted by Gor{\'e} et al.~\cite{GorPosTiu09}. This
conservativity result is not an immediate consequence of
cut-elimination in the display calculus, precisely because of the
presence of (rules affecting) backward-looking structural connectives. 

Another area to explore is the one of explicit modal logics
\cite{Art94}. Here the modality in modal logic which can be read as
provability or as knowledge is replaced by specific terms which can be
read as individual proofs or as pieces of evidence. Researchers study
\emph{realisation}-procedures which turn a proof in modal logic into a
proof in explicit modal logic. Such procedures rely on cut-free
systems for modal logics.  Nested sequent systems may provide such
realisation procedures.

\chapter{Systems for Common Knowledge}

\def\mathllap{\mathpalette\mathllapinternal}
\def\mathrlap{\mathpalette\mathrlapinternal}
\def\mathclap{\mathpalette\mathclapinternal}
\def\mathllapinternal#1#2{\llap{$\mathsurround=0pt#1{#2}$}} 
\def\mathrlapinternal#1#2{\rlap{$\mathsurround=0pt#1{#2}$}} 
\def\mathclapinternal#1#2{\clap{$\mathsurround=0pt#1{#2}$}}

\newcommand{\com}{{\sf com}}
\newcommand{\assoc}{{\sf assoc}}
\newcommand{\ac}{{\sf ac}}

\newcommand{\rr}{{\rho}}


\def\calchilb{\ensuremath{\mathsf{H_C}}}
\def\calcdeep{\ensuremath{\mathsf{D_C}}}

\def\tautax{(\mathrm{TAUT})}
\def\kax{(\mathrm{K})}
\def\cocloax{(\mathrm{CCL})}
\def\mprule{(\mathrm{MP})}
\def\necrule{(\mathrm{NEC})}
\def\indrule{(\mathrm{IND})}

\def\rk{\mathit{rk}}
\def\max{\mathit{max}}
\def\nsum{\mathop{\#}}
\newcommand{\veblen}[2]{\varphi_{#1}#2}

\def\prov#1#2{\mathrel{\vrule width 0.3pt height 8pt depth
2pt\mkern-2mu
\textstyle{\frac{\hskip0.5ex\raise2pt\hbox{$\scriptstyle#1$}\hskip0.5ex}
{\hskip0.5ex#2\hskip0.5ex}}}}

\newcommand{\provlen}[2]{\prov{#1}{} #2}
\newcommand{\provable}[2]{\prov{#1}{#2}}

\newcommand{\diaa}{\mathord{\Diamond}}
\newcommand{\boxx}{\mathord{\Box}}
\newcommand{\mystar}{\mathord{\scriptstyle *}}

\newcommand{\ckdiaa}[2]{{%
    \declareslashed{\mathord}{\mystar}{#1}{#2}{\diaa}%
    \slashed{\diaa}%
  }} \newcommand{\ckboxx}[2]{{%
    \declareslashed{\mathord}{\mystar}{#1}{#2}{\boxx}%
    \slashed{\boxx}%
  }}

\newcommand{\ckbox}{\mathord{\ckboxx{0}{-0.05}}}%
\newcommand{\ckdia}{\mathord{\ckdiaa{0}{0.22}}}%

\newcommand{\calcgentzen}{\ensuremath{\mathsf{G_C}}}

The notion of {\em common knowledge} is well-studied in epistemic
logic, where modalities express knowledge of agents. Two standard
textbooks on epistemic logic and common knowledge in particular, are
\cite{FHMV} by Fagin, Halpern, Moses, and Vardi and \cite{MevdHoe95}
by Meyer and van der Hoek.

The fact that a proposition $A$ is common knowledge can be expressed
by the infinite conjunction ``all agents know $A$ and all agents know
that all agents know $A$ and so on''. In order to express this in a
finite way we can use fixpoints: common knowledge of $A$ is then
defined to be the greatest fixpoint of the function 
\[
X \mapsto \mbox{ everybody knows } A \mbox{ and everybody knows } X.
\]
Such a definition was introduced by Halpern and Moses \cite{hm90} and
further studied in \cite{FHMV}.

The traditional way to formalise common knowledge is to use a
Hilbert-style axiom system. Such a system has a fixpoint axiom, which
states that common knowledge is a fixpoint, and an induction rule,
which states that this fixpoint is the greatest fixpoint. However,
this approach does not work well for designing a Gentzen-style sequent
calculus. In particular, Alberucci and J\"ager show in \cite{alj05}
that a cut-free sequent system designed in this way is not complete. 

To obtain a complete cut-free system Alberucci and J\"ager replace the
induction rule by an infinitary $\omega$-rule. This results in a
system in which proofs have transfinite depth and in which common
knowledge is the greatest fixpoint of the function described above.
Although this system has been further studied in \cite{KS06,JKS07}, no
syntactic cut-elimination procedure has been found. Cut-elimination
was proved only indirectly by showing completeness of the cut-free
system. No non-trivial bound on the depth of proofs in this system is
known.

In this chapter, we give a syntactic cut-elimination procedure for an
infinitary system of common knowledge based on nested sequents. Since
its inference rules apply deeply inside of the nested sequents we call
this system ``deep'' while we call the system by Alberucci and J\"ager
``shallow''. The deep system allows to straightforwardly apply the
method of {\em predicative cut-elimination}, which is a standard tool
for the proof-theoretic analysis of systems of set theory and second
order number theory, see Pohlers \cite{pohlers89,pohlers98} and
Sch\"utte \cite{schutte:77}. Since the shallow and the deep system can
be embedded into each other, this also yields a syntactic
cut-elimination procedure for the shallow system. For both systems we
thus obtain an upper bound of $\veblen{2}{0}$ on the depth of proofs,
where $\veblen{}{}$ is the Veblen function.

Please note that, like Alberucci and J\"ager, our term \emph{logic of
  common knowledge} refers to the least normal modal logic {\K}, with
an added fixpoint modality. Some people might prefer to call that the
logic of common \emph{belief}. The methods introduced here should
transfer easily to cases where rules for the modal axioms are added
that were studied in the previous chapter.  The combination of the
techniques presented here and the ones in the previous chapter should
suffice to get cut-elimination for modal logics with additional modal
axioms and common knowledge.

Several cut-free systems for logics with common knowledge exist
already. The one that is closest to our system was introduced by
Tanaka in \cite{Tan03} for predicate common knowledge logic and is
based on Kashima's ideas. It essentially also uses nested sequents,
but uses explicit labels to name the nodes of the tree. In fact, if
one disregards the rather different notation and some choices in the
formulation of rules, then one could say that our system is the
propositional part of Tanaka's system. There are also finitary
systems. Abate, Gor\'e and Widmann, for example, introduce a cut-free
tableau system for common knowledge in \cite{AGW07}. Cut-free system
have also been studied in the context of \emph{explicit modal logic}
by Artemov \cite{artemov06} and by Antonakos
\cite{DBLP:conf/lfcs/Antonakos07}.

However, we do not know of syntactic cut-elimination procedures for
any of the systems mentioned.  Typically, cut-elimination is
established only indirectly.  There are cut-elimination procedures for
similar logics, for example by Pliu{\v s}kevi{\v c}ius for an infinitary system
for linear time temporal logic in \cite{Pliu91}. For linear temporal
logic there is no need for nested sequents. For this logic it is enough to
use indexed formulas of the form $A^i$ which denotes $A$ at the $i$-th
moment in time.

This chapter is organised as follows. We first review the shallow
sequent system by Alberucci and J\"ager and show the obstacle to
cut-elimination. We then present our nested sequent system, prove the
invertibility of its rules, the admissibility of the structural rules
and finally cut-elimination. Then we embed the shallow system into the
deep system and vice versa, thus establishing cut-elimination for the
shallow system. Then, by embedding the Hilbert system into our deep
sequent system, we obtain an upper bound for the depth of proofs in
both the shallow and the deep system. Some discussion about future
work ends this chapter.

\section{The Shallow Sequent 
System}

{\bf Formulas and sequents.} We are considering a language with $h$ agents for some
$h>0$. Propositions $p$ and their negations $\neg p$ are \emph{atoms},
with $\neg{\neg p}$ defined to be $p$. \emph{Formulas} are denoted by
$A,B,C,D$. They are given by the following grammar:
\[
A \grammareq p \mid \neg p \mid (A \vlor A) \mid (A \vlan A) \mid
\Diamond_i A \mid \Box_i A \mid \ckdia A \mid \ckbox A \quad ,
\]
where $1\leq i\leq h$. The formula $\Box_i A$ is read as ``agent $i$
knows $A$'' and the formula $\ckbox A$ is read as ``$A$ is common
knowledge''. The connectives $\Box_i$ and $\ckbox$ have $\Diamond_i$
and $\ckdia$ as their respective De Morgan duals. Binary connectives
are left-associative: $A \vlor B \vlor C$ denotes $((A \vlor B)
\vlor C)$, for example.

Given a formula $A$, its \emph{negation} $\neg A$ is defined as usual
using the De Morgan laws, $A \imp B$ is defined as $\neg A \vlor B$
and $\bot$ is defined as $p \vlan \neg p$ for some proposition $p$.
The formula $\Box A$ is an abbreviation for ``everybody knows $A$'':
\[
\Box A = \Box_1 A \vlan \ldots \vlan \Box_h A \qquad \text{ and } \qquad 
\Diamond A = \Diamond_1 A \vlor \ldots \vlor \Diamond_h A.
\]
A sequence of $n\geq 0$ modal connectives can be abbreviated, for example 
\[
\Box^n A = \underbrace{\Box\dots\Box}_{n-\text{times}}A
\rlap{\quad .}
\]
A \emph{(shallow) sequent} is a finite multiset of formulas.
Sequents are denoted by $\Gamma,\Delta,\Lambda,\Pi,\Sigma$.

{\bf Inference rules.} In an instance of the inference rule $\rr$
\[
\cinf{\rr}{\Delta}{\Gamma_1 \quad \Gamma_2  \quad \dots}
\]
the sequents $\Gamma_1,\Gamma_2\dots$ are its \emph{premises} and the
sequent $\Delta$ is its \emph{conclusion}.  An \emph{axiom} is a rule
without premises. We will not distinguish between an axiom and its
conclusion.  A \emph{system}, denoted by $\Sys$, is a set of rules.
Figure~\ref{f:gc} shows system $\calcgentzen$, a shallow sequent
calculus for the logic of common knowledge. Its only axiom is called
\emph{identity axiom}. Notice that the $\ckbox$-rule has infinitely
many premises. If $\Gamma$ is a sequent then $\Diamond_i \Gamma$ is
obtained from $\Gamma$ by prefixing the connective $\Diamond_i$ to
each formula occurrence in $\Gamma$, and similarly for other
connectives.

\begin{figure}
\fbox{
\parbox{\ftextwidth}{
  \centering
  \[ \Gamma, p,\neg p \qquad \cinf{\vlan}{\Gamma, A \vlan
      B}{\Gamma, A \quad \Gamma, B} \qquad
  \cinf{\vlor}{\Gamma,A \vlor B}{\Gamma,A,B}
\]
\[
  \cinf{\Box_i}{\Diamond_i \Gamma,  \ckdia \Delta,\Box_i A, \Sigma}
{\Gamma, \ckdia \Delta, A} 
\]
 \vspace{1ex}
\[
\cinf{\ckbox}{\Gamma,\ckbox A}{\Gamma,\Box^k A 
\quad \text{for all $k \geq 1$}}
\qquad\quad 
\cinf{\ckdia}{\Gamma, \ckdia A}{\Gamma, \ckdia A,\Diamond A}
\]
}}
  \caption{System \calcgentzen}
  \label{f:gc}
\end{figure}

\begin{figure}
\centering
\[
\cinf{\wk}{\Gamma, A}{\Gamma} \qquad
\cinf{\ctr}{\Gamma,A}{\Gamma,A,A} \qquad
\cinf{\cut}{\Gamma,\Delta}{\Gamma,A \quad \Delta,\neg A}
\]
  \caption{Weakening, contraction and cut for system \calcgentzen}
  \label{f:admwcc}
\end{figure}

{\bf Derivations and proofs.} In the following, a \emph{tree} is a
tree in the graph-theoretic sense, and may be infinite.  A tree is
\emph{well-founded} if it does not have an infinite path. A
\emph{derivation} in a system $\Sys$ is a directed, rooted, ordered
and well-founded tree whose nodes are labelled with sequents and which
is built according to the inference rules from $\Sys$. Derivations are
visualised as upward-growing trees, so the root is at the bottom. The
sequent at the root is the \emph{conclusion} and the sequents at the
leaves are the \emph{premises} of the derivation. A \emph{proof} of a
sequent $\Gamma$ in a system is a derivation in this system with
conclusion $\Gamma$ where all leaves are axioms. Proofs are denoted by
$\PP$. We write $\Sys \vdash \Gamma$ if there is a proof of $\Gamma$
in system $\Sys$. Given a proof $\PP$ we denote its depth by
$|\PP|$. Notice that derivations here are in general infinitely branching,
thus their depth can be infinite even though each branch has to be finite.

{\bf Formula rank.}  Notice that formulas in the premises of the
$\ckbox$-rule are generally larger than formulas in its conclusion.
This is typically a problem for cut-elimination, but we
can easily solve this by defining an appropriate measure. For a
formula $A$ 
we define its \emph{rank} $\rk(A)$ as follows:\\
\[
\begin{array}{l}
  \rk(p) = \rk(\neg{p}) = 0\\
  \rk(A \vlan B) = \rk(A \vlor B) = \max(\rk(A),\rk(B)) + 1\\
  \rk(\Box_i A) = \rk(\Diamond_i A) = \rk(A) + 1\\
  \rk(\ckbox A) = \rk(\ckdia A) = \omega+{\rk(A)}\\
\end{array}
\]

\begin{lemma}[Some properties of the rank] For all formulas $A$ we
\label{l:rk}
  have that\\
  (i)  $\rk(A)=\rk(\neg A)$,\\
  (ii) $\rk(A)<\omega^2$,\\
  (iii) for all $k<\omega$ we have $\rk(\Box^k A)<\rk(\ckbox A)$.
\end{lemma}
\begin{proof}
  Statements (i) and (ii) are immediate. For (iii), an induction on
  $k$ yields that $\rk(\Box^k A) = \rk(A)+k\cdot h$. By (ii) it is
  then enough to check that for all $k$ and all $\alpha<\omega^2$ we
  have $\alpha+k\cdot h < \omega + \alpha$.
\end{proof}

{\bf Cut rank.}  The \emph{cut rank} of an instance of $\cut$ as shown
in Figure~\ref{f:admwcc} is the rank of its \emph{cut formula} $A$.
For an ordinal $\gamma$ we define the rule $\cut_\gamma$ which is cut
with at most rank $\gamma$ and the rule $\cut_{<\gamma}$ which is cut
with a rank strictly smaller than $\gamma$. For a system $\Sys$ and
ordinals $\alpha$ and $\gamma$ and a sequent $\Gamma$ we write
$\Sys\prov{\alpha}{\gamma}\Gamma$ to say that there is a proof of
$\Gamma$ in system $\Sys + \cut_{<\gamma}$ with depth bounded by
$\alpha$. We write $\Sys\prov{<\alpha}{\gamma}\Gamma$ to say that
there is an ordinal $\alpha_0 < \alpha$ such that
$\Sys\prov{\alpha_0}{\gamma}\Gamma$.

{\bf Admissibility and invertibility.} An inference rule $\rr$ is
\emph{depth- and cut-rank-preserving admissible} or, for short,
\emph{perfectly admissible} for a system $\Sys$ if for each instance
of $\rr$ with premises $\Gamma_1,\Gamma_2\dots$ and conclusion
$\Delta$, whenever $\Sys\prov{\alpha}{\gamma}\Gamma_i$ for each
premise $\Gamma_i$ then $\Sys\prov{\alpha}{\gamma}\Delta$.  For each
rule $\rr$ there is its \emph{inverse}, denoted by $\inverse{\rr}$,
which has the conclusion of $\rr$ as its only premise and any premise
of $\rr$ as its conclusion.  An inference rule $\rr$ is
\emph{perfectly invertible} for a system $\Sys$ if $\inverse{\rr}$
is perfectly admissible for $\Sys$.

We omit the proof of the following lemma, which is standard.

\begin{lemma}[Admissibility of the structural rules and invertibility]\nl
  (i) The rules weakening and contraction from Figure~\ref{f:admwcc} are perfectly admissible for system $\calcgentzen$.\\
  (ii) All rules of $\,\calcgentzen$ except for the $\Box_i$-rule are
  perfectly invertible for system $\calcgentzen$.
\end{lemma}

\subsection{The Problem for 
Cut-Elimination}

Let us look at the problem of cut-elimination in system
$\calcgentzen$.   Consider the following proof:
\[
\vlderivation{
\vliin{\cut}{}{\Box_i A, \Diamond_i \Gamma, \Sigma, \Delta}  {
\vlin{\Box_i}{}{\Box_i A, \Diamond_i \Gamma, \Sigma, \ckdia \neg
      B}{\Proofleaf{\PP_1}{A, \Gamma, \ckdia \neg B}}}  {
\vliiin{\ckbox}{\;^{1\leq k <\omega}}
{\ckbox B,\Delta}
{\vlhy{\vdots}}
{\Proofleaf{\PP_{2k}}{\Box^k B,\Delta}}
{\vlhy{\vdots}}}}
\]
Here the inference rule above the cut on the left does not apply to
the cut formula while the inference rule on the right does. The
typical transformation would push the left rule instance below the
cut, as follows:
\[
\vlderivation{
\vlin{\Box_i}{}{\Box_i A, \Diamond_i \Gamma, \Sigma,{{\Diamond_i}
  \Delta}} {
\vliin{\cut}{}{A, \Gamma,{\Delta}}{\Proofleaf{\PP_1}
{A, \Gamma,{\ckdia \neg
        B}}} {{
\vliiin{\ckbox}{\;^{1\leq k <\omega}}{{\ckbox
          B},
        {\Delta}}{\vlhy{\vdots}}{\Proofleaf{\PP_{2k}}{{\Box^k
            B},{\Delta}}}{\vlhy{\vdots}}}}}}
\]

However, this transformation introduces the $\Diamond_i$ in
$\Diamond_i \Delta$, and thus it does not yield a proof of the
original conclusion. This problem is caused by the context restriction
in the $\Box_i$-rule.

Such a context restriction also occurs in the standard sequent
calculus for the modal logic {\K}. While it destroys invertibility, at
least it does not cause any difficulties for syntactic cut-elimination
for {\K}.  However, we see that the context restriction poses a genuine
problem for logics with more modalities like in the logic of common
knowledge. In the next section we will see how a more general format
for sequents and inference rules solves the problem since it does not
require context restrictions.

\section{The Nested Sequent 
System}

\textbf{Nested sequents.} A \emph{nested sequent} is a finite multiset of
formulas and boxed sequents.  A \emph{boxed sequent} is an expression
$[\Gamma]_i$ where $\Gamma$ is a nested sequent and $1 \leq i \leq h$.
The letters $\Gamma,\Delta,\Lambda,\Pi,\Sigma$ from now on denote nested
sequents and the word sequent from now on refers to nested sequent,
except when it is clear from the context that a sequent is shallow,
such as a sequent appearing in a derivation in $\calcgentzen$.  A
sequent is always of the form
\[
A_1,\dots,A_m,[\Delta_1]_{i_1},\dots,[\Delta_n]_{i_n} \rlap{\quad ,}
\]
where the $i_j$ denote agents and thus range from $1$ to $h$. As
usual, the comma denotes multiset union and there is no distinction
between a singleton multiset and its element.

Fix an arbitrary linear order on formulas. Fix an arbitrary linear
order on boxed sequents. The \emph{corresponding formula} of a
non-empty sequent $\Gamma$, denoted $\formula{\Gamma}$, is defined as
follows:
\[
\formula{A_1,\dots,A_m,[\Delta_1]_{i_1},\dots,[\Delta_n]_{i_n}} = 
A_1 \vlor \dots \vlor A_m \vlor \Box_{i_1}\formula{\Delta_1} \vlor \dots
\vlor 
\Box_{i_n}\formula{\Delta_n} \mathrlap{\; ,}
\]
where formulas and boxed sequents are listed according to the fixed
orders. The \emph{corresponding formula} of the empty sequent is
$\bot$.  A sequent has a \emph{corresponding tree} whose
nodes are marked with multisets of formulas and whose edges are marked
with agents. The corresponding tree of the above sequent is
\[
\vcenter{\xymatrix@!C=1cm{{}  & {} & {\{A_1,\dots,A_m\}} 
\ar[dll]_{i_1} \ar[dl]^{i_2} \ar[dr]_{i_{n-1}} \ar[drr]^{i_n}\\
    {\tree{\Delta_1}} & {\tree{\Delta_2}} & {\dots} &
    {\tree{\Delta_{n-1}}}& {\tree{\Delta_n}}}} \qquad ,
\]
where $\tree{\Delta_1}\dots\tree{\Delta_n}$ are the corresponding
trees of $\Delta_1\dots\Delta_n$.  Often we do not distinguish between
a sequent and its corresponding tree, for example the \emph{root} of a
sequent is the root of its corresponding tree.

{\bf Formula contexts and sequent contexts.} A \emph{formula context}
is a formula with exactly one occurrence of the special atom
$\context$, which is called \emph{the hole} or \emph{the empty
  context}. A \emph{sequent context} is a sequent with exactly one
occurrence of the hole, which does not occur inside formulas. Formula
contexts are denoted by $A\context, B\context$, and so on. Sequent
contexts are denoted by $\Gamma\context$, $\Delta\context$, and so on.
Formally, sequent contexts are generated inductively as follows: if
$\Delta$ is a sequent then $\Delta, \context$ is a sequent context,
and if $\Delta$ is a sequent and $\Gamma\context$ is a sequent
context, then $\Delta, [\Gamma\context]$ is a sequent context.

The formula $A\cons{B}$ is obtained by replacing $\context$ inside
$A\context$ by $B$ and the sequent $\Gamma\cons{\Delta}$ is obtained
by replacing $\context$ inside $\Gamma\context$ by $\Delta$. For
example, if $\Gamma\context=A,[[B]_1,\context]_2$ and $\Delta=C,[D]_3$
then
\[
\Gamma\cons{\Delta}=A,[[B]_1,C,[D]_3]_2 \quad .
\]
Formally, given a sequent $\Delta$ and a sequent context
$\Gamma\context$ then $\Gamma\cons{\Delta}$ is defined inductively as
follows: if $\Gamma\context=\Gamma_1,\context$ then
$\Gamma\cons{\Delta}=\Gamma_1,\Delta$ and if
$\Gamma\context=\Gamma_1,[\Gamma_2\context]$ then
$\Gamma\cons{\Delta}=\Gamma_1,[\Gamma_2\cons{\Delta}]$.

The \emph{corresponding formula context} of a sequent context
$\Gamma\context$, denoted $\formula{\Gamma\context}$ is defined as
follows:
\[
\begin{array}{l}
  \formula{\Gamma,\context} = \formula{\Gamma} \vlor \context\\
  \formula{\Gamma, [\Delta\context]_i} = \formula{\Gamma} \vlor \Box_i
\formula{\Delta\context}  
\end{array}
\]

\begin{figure}
\fbox{
\parbox{\ftextwidth}{
  \centering
  \[ \Gamma\cons{p,\neg p} \qquad \cinf{\vlan}{\Gamma\cons{A \vlan
      B}}{\Gamma\cons{A} \quad \Gamma\cons{B}} \qquad
  \cinf{\vlor}{\Gamma\cons{A \vlor B}}{\Gamma\cons{A,B}}
\]
\[
  \cinf{\Box_i}{\Gamma\cons{\Box_i A}}{\Gamma\cons{[A]_i}} \qquad\quad
\cinf{\Diamond_i}{\Gamma\cons{\Diamond_i A, [\Delta]_i}}{\Gamma\cons{\Diamond_i A, [\Delta, A]_i}}
\]
 \vspace{1ex}
\[
\cinf{\ckbox}{\Gamma\cons{\ckbox A}}{\Gamma\cons{\Box^k
    A} \quad \text{for all $k \geq 1$}}
\qquad\quad 
\cinf{\ckdia}{\Gamma\cons{\ckdia A}}{\Gamma\cons{\ckdia A,\Diamond^k A}}
\]
}}
  \caption{System \calcdeep}
  \label{f:dc}
\end{figure}

\begin{figure}
\centering
\[
\cinf{\nec}{[\Gamma]_i}{\Gamma} \qquad
\cinf{\wk}{\Gamma\cons{\Delta}}{\Gamma\cons{\emptyset}} \qquad
\cinf{\ctr}{\Gamma\cons{\Delta}}{\Gamma\cons{\Delta,\Delta}} \qquad
\cinf{{\cut}}{\Gamma\cons{\emptyset}}{\Gamma\cons{A} \qquad
  \Gamma\cons{\neg A}}
\]

  \caption{Necessitation, weakening, contraction and cut for system
\calcdeep}
  \label{f:adm}
\end{figure}

Figure~\ref{f:dc} shows our nested sequent system $\calcdeep$.
Figure~\ref{f:adm} shows the structural rules \emph{necessitation},
\emph{weakening} and \emph{contraction} as well as the rule
\emph{cut}, which are associated to system $\calcdeep$. Notice that
the rules of system $\calcdeep$ and the associated rules are different
from the corresponding rules in system $\calcgentzen$ but have the
same names. If we refer to a rule only by its name then it will be
clear from the context which rule is meant. For example the cut in
$\calcgentzen+\cut$ is the one associated to system $\calcgentzen$ and
the one in $\calcdeep+\cut$ is the one associated with system
$\calcdeep$.

\begin{lemma}[Admissibility of the structural rules and invertibility]
  \nl (i) The rules necessitation, weakening and contraction from
  Figure~\ref{f:adm} are
  perfectly admissible for system $\calcdeep$.\\
  (ii) All rules in $\calcdeep$ are perfectly
  invertible for $\calcdeep$.\\
\end{lemma}
\begin{proof}
  Admissibility of necessitation and weakening follow from a routine
  induction on the depth of the proof. The same works for the
  invertibility of the $\vlan, \vlor, \Box_i$ and $\ckbox$-rules in
  (ii). The inverses of all other rules are just weakenings. For
  admissibility of contraction we also proceed by induction on the
  depth of the proof tree, using invertibility of the rules.  The
  cases for the propositional rules and for the
  $\Box_i,\ckbox,\ckdia$-rules are trivial. For the $\Diamond_i$-rule
  we consider the formula $\Diamond_i A$ from its conclusion
  $\Gamma\cons{\Diamond_i A, [\Delta]_i}$ and its position inside the
  premise of contraction $\Lambda\cons{\Sigma,\Sigma}$. We have the
  cases 1) $\Diamond_i A$ is inside $\Sigma$ or 2) $\Diamond_i A$ is
  inside $\Lambda\context$.  We have two subcases for case 1: 1.1)
  $[\Delta]_i$ inside $\Lambda\context$, 1.2) $[\Delta]_i$ inside
  $\Sigma$. There are three subcases of case 2: 2.1) $[\Delta]_i$
  inside $\Lambda\context$ and 2.2) $[\Delta]_i$ inside $\Sigma$, 2.3)
  $\Sigma,\Sigma$ inside $[\Delta]_i$.  All cases are either simpler
  than or similar to case 2.2, which is as follows:
  \[
  \vlderivation{
\vlin{{\ctr}}{}{\Lambda'\cons{\Diamond_i A, \Sigma', [\Delta]_i}}{
\vlin{{\Diamond_i}}{}{\Lambda'\cons{\Diamond_i A, \Sigma',
          [\Delta]_i,\Sigma',
          [\Delta]_i}}{
\vlhy{\Lambda'\cons{\Diamond_i A, \Sigma', 
[\Delta,A]_i, \Sigma', [\Delta]_i}}  
}}} \quad
  \leadsto \quad 
\vlderivation{ 
\vlin{{\Diamond_i}}{}{\Lambda'\cons{\Diamond_i A,
        \Sigma', [\Delta]_i}}{ 
\vlin{\ctr}{}{\Lambda'\cons{\Diamond_i A,
          \Sigma', [\Delta,A]_i}}{ 
\vlin{{\inverse{\Diamond}_i}}{}
{\Lambda'\cons{\Diamond_i A, \Sigma', [\Delta,A]_i, \Sigma',
            [\Delta,A]_i}}{\vlhy{\Lambda'\cons{\Diamond_i
              A, \Sigma', [\Delta,A]_i, \Sigma', [\Delta]_i}}}}} } \quad
  ,
  \]
  where the instance of $\inverse{\Diamond}_i$ in the proof on the right
  is removed because it is perfectly admissible and the
  instance of contraction is removed by the induction hypothesis.
\end{proof}

\begin{lemma}[Derivability of the general identity axiom]\label{l:gid}
  For all contexts $\Gamma\context$ and all formulas $A$ we have
$\calcdeep \prov{2\cdot\rk(A)}{0} \Gamma\cons{A,\neg{A}}$.
\end{lemma}
\begin{proof}
  We perform an induction on $\rk(A)$ and a case analysis on the main
  connective of $A$. The cases for atoms and for the propositional
  connectives are obvious. For $A=\Box_i B$ and $A=\ckbox B$ we
  respectively have
\[
\vlderivation{
\vlin{\Box_i}{}{\Gamma\cons{\Box_i B, \Diamond_i \neg B}}
{\vlin{\Diamond_i}{}{\Gamma\cons{[B]_i, \Diamond_i \neg B}}
{\vlhy{\Gamma\cons{[B,\neg B]_i, \Diamond_i \neg B}}}}
}
\qquad\mbox{and}\qquad
\vlderivation{
\vliiin{\ckbox}{\;^{1\leq k < \omega}}{\Gamma\cons{\ckbox B, \ckdia
\neg B}}
{\vlhy{\vdots}}
{\vlin{\wk,\ckdia}{}{\Gamma\cons{\Box^k B, \ckdia\neg B}}
  {\vlhy{\Gamma\cons{\Box^k B, \Diamond^k \neg B}}}}
{\vlhy{\vdots}}
} \rlap{\quad .}
\]
On the left by induction hypothesis we get a proof of the premise of
depth $2\cdot\rk(B)$ and thus a proof of the conclusion of depth
$2\cdot\rk(B)+2=2\cdot(\rk(B)+1)=2\cdot\rk(\Box_i B)$. On the right by
Lemma~\ref{l:rk} we can apply the induction hypothesis for each
premise to get a proof of depth $2\cdot \rk(\Box^k B)=2\cdot
(\rk(B)+k\cdot h)$ and thus a proof of the conclusion of depth $2\cdot
(\rk(B)+\omega) \leq 2\cdot (\omega + \rk(B))=2\cdot\rk(\ckbox B)$.
\end{proof}

\section{Cut-Elimination for the Nested System}

We first need some notions concerning ordinals. For an introduction to
ordinals we refer to Sch\"utte \cite{schutte:77}. We write $\alpha
\nsum \beta$ for the {\em natural sum of $\alpha$ and $\beta$} which,
in contrast to the ordinary ordinal sum, does not cancel additive
components. In particular, the natural sum is commutative. To give
names to the ordinals which measure our proofs we also need the
following definition.

\begin{definition}[Veblen function]
  The \emph{binary Veblen function} $\varphi$ is generated inductively
  as follows:
  \begin{enumerate}
  \item $\veblen{0}{\beta} := \omega^\beta$,
  \item if $\alpha>0$, then $\veblen{\alpha}{\beta}$ is the
    $(\beta+1)$th common fixpoint of the functions
    $\xi\mapsto\veblen{\gamma}{\xi}$ for all $\gamma < \alpha$.
  \end{enumerate}
\end{definition}

The Veblen function just generates an increasing sequence of ordinals,
as follows: $\veblen{0}{0}=\omega^0=1$,
$\veblen{0}{1}=\omega^1=\omega$, $\dots$,
$\veblen{0}{\omega}=\omega^\omega$, $\dots$,\\
$\veblen{1}{0}=$ first fixpoint of the function
($\xi\mapsto\veblen{0}{\xi}=\xi\mapsto\omega^\xi)=\varepsilon_0$,
$\dots$,\\
$\veblen{2}{0}=$ first fixpoint of the function
($\xi\mapsto\veblen{1}{\xi}=\xi\mapsto\varepsilon_\xi)$,
$\dots$\quad .

Here we will only need ordinals up to $\veblen{2}{0}$. In this
subsection we write $\prov{\alpha}{\beta}\Gamma$ for
$\calcdeep\prov{\alpha}{\beta}\Gamma$. We now prove the central lemma.

\begin{lemma}[Reduction Lemma]
  If there is a proof
\[
\vlderivation{
\vliin{\cut_{\gamma}}{} {\Gamma\cons{\emptyset}}
  {\Proofleaf{\PP_1}{\Gamma\cons{A}}}
  {\Proofleaf{\PP_2}{\Gamma\cons{\neg A}}}}
\]
with $\PP_1$ and $\PP_2$ in $\calcdeep+\cut_{<\gamma}\,$ then
$\prov{|\PP_1|\nsum|\PP_2|}{\gamma}\Gamma\cons{\emptyset}\,$.\\
\end{lemma}

\begin{proof}
  By induction on $|\PP_1| \nsum |\PP_2|$. We perform a case analysis
  on the two lowermost rules in the given proofs. If one of the two
  rules is passive and an axiom then $\Gamma\cons{\emptyset}$ is
  axiomatic as well. If one is active and an axiom then we have
\[
\vlderivation{
\vliin {\cut_{0}}{} {\Gamma\cons{\neg p}}
  {\vlhy{\Gamma\cons{p, \neg p}}}
  {\Proofleaf{\PP_{2}}{\Gamma\cons{\neg p, \neg p}}}}
\qquad \leadsto \qquad
\vlderivation{
\vlin {\ctr}{} {\Gamma\cons{\neg p}}
  {\Proofleaf{\PP_{2}}{\Gamma\cons{\neg p, \neg p}}}}
\quad ,
\]
and by contraction admissibility we have
$\prov{|\PP_2|}{0}\Gamma\cons{\neg p}$ and thus
$\prov{|\PP_1|\nsum|\PP_2|}{0}\Gamma\cons{\neg p}$. If some rule
$\rr$ is passive then we have
\[
\vlderivation{
\vliin {\cut_{\gamma}}{} {\Gamma\cons{\emptyset}}
  {\Proofleaf{\PP_1}{\Gamma\cons{A}}}
  {\vliiin{\rr}{}{\Gamma\cons{\neg A}}{\vlhy{\vdots}}{
  \Proofleaf{\PP_{2i}}{\Gamma_i\cons{\neg A}}}{\vlhy{\vdots}}}}
\quad \leadsto \quad
\]\[
\vlderivation{\vliiin{\rr}{}{\Gamma\cons{\emptyset}}{\vlhy{\vdots}}
  {\vliin {\cut_{\gamma}}{} {\Gamma_i\cons{\emptyset}}
  {\vlin{\neg \rr}{}{\Gamma_i\cons{A}}{\Proofleaf{\PP_1}{\Gamma\cons{A}}}}
  {\Proofleaf{\PP_{2i}}{\Gamma_i\cons{\neg A}}}}{\vlhy{\vdots}}}
\quad ,
\]
where $i$ ranges from 1 to the number of premises of $\rr$. By
invertibility of $\rr$ we get
$\prov{|\PP_1|}{\gamma}\Gamma_i\cons{A}$, thus by induction hypothesis
$\prov{|\PP_1|\nsum |\PP_{2i}|}{\gamma}\Gamma_i\cons{\emptyset}$ for
all $i$ and by $\rr$ we get $\prov{|\PP_1|\nsum
  |\PP_{2}|}{\gamma}\Gamma\cons{\emptyset}$.

This leaves the case that both rules are active and neither is an
axiom. We have:

$(\vlan - \vlor)$:
\[
\vlderivation{
\vliin {\cut_{\sigma+1}}{} {\Gamma\cons{\emptyset}}
  {\vliin{\vlan}{}{\Gamma\cons{B \vlan C}}
  {\Proofleaf{\PP_{11}}{\Gamma\cons{B}}}
  {\Proofleaf{\PP_{12}}{\Gamma\cons{C}}}}
  {\vlin{\vlor}{}{\Gamma\cons{\neg B \vlor \neg C}}{
  \Proofleaf{\PP_{21}}{\Gamma\cons{\neg B, \neg C}}}}}
\qquad \leadsto \qquad
\]\[
\vlderivation{
\vliin{\cut_\sigma}{}{\Gamma\cons{\emptyset}}
{\Proofleaf{\PP_{11}}{\Gamma\cons{B}}}
{\vliin{\cut_{\sigma}}{} {\Gamma\cons{\neg B}}
  {\vlin{\wk}{}{\Gamma\cons{\neg B, C}}
{\Proofleaf{\PP_{12}}{\Gamma\cons{C}}}}
  {\Proofleaf{\PP_{21}}{\Gamma\cons{\neg B, \neg C}}}}} \qquad ,
\]
where by weakening admissibility we get
$\prov{|\PP_{12}|}{\gamma}{\Gamma\cons{\neg B, C}}$, and since
$\sigma<\sigma+1=\gamma$ we get
$\prov{\alpha}{\gamma}\Gamma\cons{\emptyset}$ for
$\alpha=\max(|\PP_{11}|,\max(|\PP_{12}|,|\PP_{21}|)+1)+1$. It is easy
to check that $\alpha \leq |\PP_1|\nsum|\PP_2|$.

$(\Box_i - \Diamond_i)$:
\[
\vlderivation{
\vliin {\cut_{\sigma+1}}{} {\Gamma\cons{[\Delta]_i}}
 {\vlin{\Box_i}{}{\Gamma\cons{[\Delta]_i,\Box_i A}}
   {\Proofleaf{\PP_{11}}{\Gamma\cons{[\Delta]_i,[A]_i}}}} 
{\vlin{{\Diamond_i}}{}
   {\Gamma\cons{[\Delta]_i, \Diamond_i \neg A}}
   {\Proofleaf{\PP_{21}}{\Gamma\cons{[\Delta,\neg A]_i, \Diamond_i
\neg A}}}}
}
\qquad \leadsto \qquad
 \] \[
 \vlderivation{
\vliin {\cut_{\sigma}}{} {\Gamma\cons{[\Delta]_i}} {
\vlin {\ctr}{}
   {\Gamma\cons{[\Delta,A]_i}}
   {\vlin{\wk^2}{}{\Gamma\cons{[\Delta,A]_i,[\Delta,A]_i}}{\Proofleaf{\PP_{11}}
{\Gamma\cons{[\Delta]_i,[A]_i}}}}}
 {\vliin {\cut_{\sigma+1}}{} {\Gamma\cons{[\Delta, \neg A]_i}} {
\vlin {\wk,\Box_i}{}
     {\Gamma\cons{[\Delta, \neg A]_i,\Box_i A}}
     {\Proofleaf{\PP_{11}}{\Gamma\cons{[\Delta]_i,[A]_i}}}}
   {\Proofleaf{\PP_{21}}{\Gamma\cons{[\Delta, \neg A]_i, \Diamond_i
\neg A}}}}}
\quad ,
\] 
where the premises of the upper cut have been derived by use of
weakening admissibility with depth $|\PP_{11}|+1$ and $|\PP_{21}|$,
the natural sum of which is smaller than $|\PP_1|\nsum|\PP_2|$. The
induction hypothesis thus yields
$\prov{(|\PP_{11}|+1)\nsum|\PP_{21}|}{\gamma}\Gamma\cons{[\Delta,\neg
  A]_i}$ and since $\sigma<\sigma+1=\gamma$ we get
$\prov{|\PP_{1}|\nsum|\PP_{2}|}{\gamma}\Gamma\cons{[\Delta]_i}$ by the
lower cut.

$(\ckbox - \ckdia)$:
\[
\vlderivation{
\vliin{\cut_{\omega+\sigma}}{}{\Gamma\cons{\emptyset}}{
\vliiin{\ckbox}
{\;^{1 \leq k}}{\Gamma\cons{\ckbox A}}
{\vlhy{\vdots}}
{\Proofleaf{\PP_{1k}}{\Gamma\cons{\Box^k A}}}
{\vlhy{\vdots}}}
{\vlin{\ckdia}{}{\Gamma\cons{\ckdia \neg A}}
{\Proofleaf{\PP_{21}}{\Gamma\cons{\ckdia\neg A, \Diamond^j \neg A}}}}
}\qquad \leadsto
\]
\[
\vlderivation{
\vliin{\cut_{\sigma+(j\cdot h)}}{}{\Gamma\cons{\emptyset}}
{\Proofleaf{\PP_{1j}}{\Gamma\cons{\Box^j A}}}
{\vliin{\cut_{\omega+\sigma}}{}{\Gamma\cons{\Diamond^j \neg A}}
{\vliiin{\ckbox}{\;^{1 \leq k}}
{\Gamma\cons{\ckbox A, \Diamond^j \neg A}}
{\vlhy{\vdots}}{
\vlin{\wk}{}{\Gamma\cons{\Box^k A, \Diamond^j \neg A}}
{\Proofleaf{\PP_{1k}}{\Gamma\cons{\Box^k A}}}}{\vlhy{\vdots}}}
{\Proofleaf{\PP_{21}}{\Gamma\cons{\ckdia \neg A, \Diamond^j \neg A}}}}}
\quad ,
\]
where the induction hypothesis applied on the upper cut gives us
$\prov{|\PP_1|\nsum|\PP_{21}|}{\gamma}\Gamma\cons{\Diamond^j \neg A}$
and since by Lemma~\ref{l:rk} we have $\sigma+j\cdot h < \omega +
\sigma = \gamma$ the lower cut yields
$\prov{|\PP_1|\nsum|\PP_{2}|}{\gamma}\Gamma\cons{\emptyset}$.
\end{proof}

From the reduction lemma we obtain the first and the second
elimination lemma as usual, see for instance Pohlers
\cite{pohlers89,pohlers98} or Sch\"utte \cite{schutte:77}.

\begin{lemma}[First Elimination Lemma]
If\/ $\prov{\alpha}{\gamma+1}\Gamma$ then  $\;\prov{2^\alpha}{\gamma}\Gamma$.
\end{lemma}
\begin{proof}
  By induction on $\alpha$ and a case analysis on the last rule
  applied. Most cases are trivial, in case of a cut with rank $\gamma$
  we apply the induction hypothesis to both proofs of the premises of
  the cut and then apply the reduction lemma to obtain
  $\prov{2^{\alpha_0} \nsum 2^{\alpha_0}}{\gamma}\Gamma$ for some
  $\alpha_0<\alpha$ and thus $\prov{2^\alpha}{\gamma}\Gamma$.
\end{proof}

\begin{lemma}[Second Elimination Lemma]
  If\/ $\prov{\alpha}{\beta + \omega^\gamma}\Gamma$ then\/
  $\prov{\veblen{\gamma}{\alpha}}{\beta}\Gamma$.
\end{lemma}
\begin{proof}
  By induction on $\gamma$ with a subinduction on $\alpha$. For
  $\gamma=0$ this trivially follows from the first elimination lemma.
  Assume $\gamma>0$. The non-trivial case is where the last rule in
  the given proof of $\Gamma$ is a cut with a rank of $\beta$ or
  greater. With $\Gamma=\Gamma\cons{\emptyset}$ the proof is of the
  following form:
\[
\vlderivation{
\vliin{\cut_{<\beta+\omega^\gamma}}{}{\Gamma\cons{\emptyset}}
{\Proofleaf{\PP_1}{\Gamma\cons{A}}}
{\Proofleaf{\PP_2}{\Gamma\cons{\neg A}}}}
\rlap{\qquad .}
\]
Let $\alpha_0=\max(|\PP_1|,|\PP_2|)$. We apply the subinduction
hypothesis on the subproofs of the cut and obtain
$\prov{\veblen{\gamma}(\alpha_0)}{\beta}\Gamma\cons{A}$ and
$\prov{\veblen{\gamma}(\alpha_0)}{\beta}\Gamma\cons{\neg A}$. Since
$\rk(A) < \beta+\omega^\gamma$ a quick calculation by case analysis on
$\gamma$ yields the existence of $\sigma$ with $\sigma<\gamma$ and of
$n$ such that $\rk(A) < \beta+\omega^\sigma\cdot n$. Thus, by a cut we
obtain $\prov{\veblen{\gamma}(\alpha_0)+1}{\beta+\omega^\sigma\cdot
  n}\Gamma$. We apply the induction hypothesis $n$ times to obtain
$\prov{\veblen{\sigma}^n(\veblen{\gamma}(\alpha_0)+1)} {\beta}\Gamma$,
where $\veblen{\sigma}^n$ means $\veblen{\sigma}$ applied $n$ times.
Since
$\veblen{\sigma}^n(\veblen{\gamma}(\alpha_0)+1)<\veblen{\gamma}(\alpha)$
we have $\prov{\veblen{\gamma}(\alpha)}{\beta}\Gamma$.
\end{proof}

The cut-elimination theorem follows by iterated application of the
second elimination lemma.

\begin{theorem}[Cut-elimination for the deep system]\nl
  If\/ $\calcdeep\prov{\alpha}{\omega\cdot n}\Gamma$ 
then
  $\calcdeep\prov{\veblen{1}^n(\alpha)}
{0}\Gamma$.
\label{t:cutel}
\end{theorem}

\section{Cut-Elimination 
for the Shallow System
}

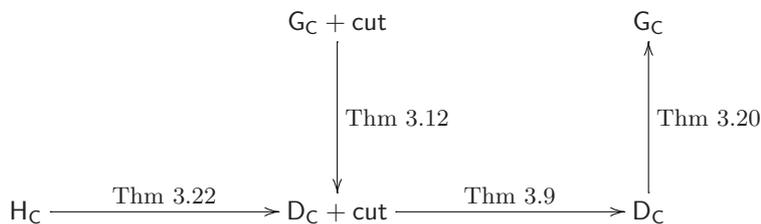
\begin{figure}
  \centering
\[
  \vcenter{ \xymatrix@R=2cm@C=3cm{{} & {\calcgentzen+\cut}
      \ar[d]^-{\txt{\small Thm \ref{t:shallowdeep}}} &
      {\calcgentzen} \\
      {\calchilb} \ar[r]^-{\txt{\small Thm \ref{t:hilb}}} &
      {\calcdeep+\cut} \ar[r]^-{\txt{\small Thm \ref{t:cutel}}} &
      {\calcdeep} \ar[u]_-{\txt{\small Thm \ref{t:deepshallow}}}\\
    } }
  \]
  \caption{Overview of the various embeddings}
\label{fig:emb}
\end{figure}

In this section we give a cut-elimination procedure for the shallow
system. To do so, we first embed the shallow system with cut into the
deep system with cut, eliminate the cut there, and embed the cut-free
deep system into the cut-free shallow system. Figure~\ref{fig:emb}
gives an overview of the embeddings. We have seen the horizontal arrow
on the right in the last section. Now we are going to see the vertical
arrows. System $\calchilb$ is a Hilbert system which we will see in
the last section, together with the horizontal arrow on the left.

\subsection{Embedding Shallow 
into Deep}

This is the easy direction. We first define a notion of admissibility
which is weaker than ``depth-preserving'': it allows the proof
to grow by a finite amount.

\begin{definition}
  A rule $\rr$ is \emph{finitely admissible} for a system $\Sys$ if
  for each instance of $\rr$ with premises $\Gamma_1,\Gamma_2\dots$
  and conclusion $\Delta$ there exists a finite ordinal $n$ such that
  whenever $\Sys\prov{\alpha}{\gamma}\Gamma_i$ for all $i$ then
  $\Sys\prov{\alpha+n}{\gamma}\Delta$.
\end{definition}

Note that every perfectly admissible (that is, depth- and
cut-rank-preserving admissible) rule is also finitely admissible: in
that case the $n$ in the above definition is zero. A finitary rule
which is contained in a system is also finitely admissible for that
system: in that case the $n$ in the above definition is one. The cut
rule, on the other hand, is generally not finitely admissible for
(cut-free) infinitary systems.

\begin{lemma}
  The rule $\vlinf{{\sf d}}{}{\Gamma\cons{\ckdia A,
      [\Delta]_i}}{\Gamma\cons{[\ckdia A, \Delta]_i}}$ is finitely
  admissible for system $\calcdeep$.
\end{lemma}
\begin{proof}
  By induction on the depth of the proof of the premise. The only
interesting case is the one with a $\ckdia$-rule:
\[
\vlderivation{\vlin{{\sf d}}{}
{\Gamma\cons{\ckdia A, [\Delta]_i}}
{\vlin{\ckdia}{}
{\Gamma\cons{[\ckdia A, \Delta]_i}}
{\Proofleaf{\PP}
{\Gamma\cons{[\ckdia A, \Diamond^k A, \Delta]_i}}
}}}
\qquad\leadsto\qquad
\vlderivation{
\vlin{\ckdia}{}
{\Gamma\cons{\ckdia A, [\Delta]_i} }
{\vlin{\wk,\vlor^*}{}
{\Gamma\cons{\ckdia A, \Diamond^{k+1} A, [\Delta]_i} }
{\vlin{\wk, \Diamond_i}{}
{\Gamma\cons{\ckdia A, \Diamond_i \Diamond^k A, [\Delta]_i} }
{\vlin{{\sf d}}{}
{\Gamma\cons{\ckdia A, [\Diamond^k A, \Delta]_i} }
{\Proofleaf{\PP}
{\Gamma\cons{[\ckdia A, \Diamond^k A, \Delta]_i} }
}}}}
} \rlap{\qquad ,}
\]
where the instance of $\sf d$ shown on the right is removed by
induction hypothesis.
\end{proof}

\begin{theorem}[Shallow into deep] If $\;\calcgentzen
  \provable{\alpha}{\gamma}{\Gamma}\;$ then
  $\;\calcdeep \provable{\omega \cdot \alpha}{\gamma}\Gamma\;$ .\\
\label{t:shallowdeep}  
\end{theorem}

\begin{proof}
  By induction on $\alpha$ and a case analysis on the last rule in the
  proof. Each rule of $\calcgentzen$ except for the $\Box_i$-rule is a
  special case of its respective rule in $\calcdeep$. For the
$\Box_i$-rule we have the following transformation:
\[
\vlderivation{
\vlin{\Box_i}{}
{\Diamond_i \Gamma, \ckdia \Delta, \Box_i A, \Sigma}
{\Proofleaf{\PP}
{\Gamma,\ckdia\Delta,A}}}
\qquad\leadsto\qquad
\vlderivation{
\vlin{\Box_i,\wk}{}
{\Diamond_i \Gamma, \ckdia\Delta, \Box_i A, \Sigma}
{\vlin{{\sf d}^*}{}
{\Diamond_i \Gamma, \ckdia\Delta, [A]_i  }
{\vlin{\wk^*,\Diamond_i^*}{}
{\Diamond_i \Gamma, [\ckdia\Delta,A]_i }
{\vlin{\nec}{}
{[\Gamma,\ckdia\Delta,A]_i }
{\Proofleaf{\PP'}
{\Gamma,\ckdia\Delta,A}
}}}}}\rlap{\qquad ,}
\]
where $\PP'$ is obtained by induction hypothesis.
\end{proof}

\subsection{Embedding Deep 
into Shallow}

This is the harder direction, since we need to simulate deep
applicability of rules in the shallow system. We use the invertibility
of rules in the shallow system in order to do so. The $\Box_i$-rule is
the only rule in $\calcgentzen$ which is not invertible. However, a
somewhat weaker property than invertibility holds, which is sufficient
for our purposes, and which is stated in the upcoming lemma. 

\begin{example}
  To motivate the following definition consider the following three
  provable sequents to which the $\Box_i$-rule cannot be
  applied (upwards) in an invertible way:
\[
\Box_i(a \vlan b), \Diamond_i \neg a \vlor \Diamond_i \neg b
\qquad 
\Box_i(a \vlan b), \ckdia \neg a \vlor \ckdia \neg b
\qquad
\Box_i a, \ckdia \neg a
\qquad .
\]

\end{example}

\begin{definition}[hiding formula, $\ckdia$-saturated sequent]
  A formula is \emph{essentially $\Diamond_i$} if 1) it is of the form
  $\Diamond_i A$ for any formula $A$ or 2) it is of the form $A \vlor
  B, B \vlor A, A \vlan B$ or $B\vlan A$ where $A$ is any formula and
  $B$ is a formula which is essentially $\Diamond_i$.  A formula is
  \emph{hiding $\Diamond_i$} in case 2). We define \emph{essentially
    $\ckdia$} and \emph{hiding $\ckdia$} formulas likewise. A formula
  is just \emph{hiding} if it is either \emph{hiding $\Diamond_i$} for
  some $i$ or \emph{hiding $\ckdia$}.  A sequent $\Gamma$ is
  \emph{$\ckdia$-saturated} if $\ckdia A \in\Gamma$ implies $\Diamond_i A \in
  \Gamma$, for each formula $A$ and each $i$ with
  $1\leq i\leq h$.
\end{definition}

\begin{definition}[canonical $\Box_i$-instance]
  An instance of the rule
\[
  \cinf{\Box_i}{\Diamond_i \Gamma,  \ckdia \Delta,\Box_i A, \Sigma}
{\Gamma, \ckdia \Delta, A} 
\]
is \emph{canonical} if no formulas of the form $\Diamond_i B$ or
$\ckdia B$ are in $\Sigma$.
\end{definition}

\begin{lemma}[Quasi-invertibility of the $\Box_i$-rule]\label{l:boxquinv}
  Let $\Gamma$ be a $\ckdia$-saturated sequent without hiding formulas
  and let there be a proof of the sequent $\Box_i A, \Gamma$ in
  $\calcgentzen$. Then there is a proof of the same depth in
  $\calcgentzen$ either 1) of the sequent $\Gamma$ or 2) of the
  sequent $\Box_i A, \Gamma$ where the last rule instance is a
  canonical instance of the $\Box_i$-rule applying to the shown
  formula $\Box_i A$.
\end{lemma}
\begin{proof}
  By induction on the depth of the given proof and a case analysis on the
  last rule.  If the endsequent is axiomatic then $\Gamma$ is
  axiomatic and the first disjunct of our lemma applies. If the last
  rule is the $\ckbox$-rule then the proof is of the form
\[
\vlderivation{
\vliiin{\ckbox}{\;^{1 \leq k}}{\Box_i A,\Gamma_1,\ckbox B}
  {\vlhy{\vdots}}{\Proofleaf{\PP_k}{\Box_i A, \Gamma_1,\Box^k B}}
{\vlhy{\vdots}}}
\]

We apply the induction hypothesis to each premise, with $\Gamma =
\Gamma_1, \Box^k B$. Notice that $\Gamma$ is $\ckdia$-saturated and
does not contain hiding formulas.  There are two cases. First, if for all
premises the first disjunct of the induction hypothesis is true then
for each $k$ we have a proof $\PP'_k$ such that the following shows
the first disjunct of our lemma:
\[
\vlderivation{\vliiin{\ckbox}{\;^{1 \leq k}}{\Gamma_1,\ckbox B}
  {\vlhy{\vdots}}{\Proofleaf{\PP'_k}{\Gamma_1,\Box^k B}}
{\vlhy{\vdots}}}\rlap{\qquad .}
\]

Second, if for some premise the second disjunct of the induction hypothesis is
true then for some $k$ we have a proof of the form
\[
\vlderivation{
\vlin{\Box_i}{}{\Box_i A, \Gamma_1, \Box^k B}
{\Proofleaf{}{A, \Gamma'}}}\rlap{\qquad .}
\]
Notice that the $\Box_i$-rule can only introduce a formula of the form
$\Box^k B$ in $\Sigma$, so we can easily turn this into a proof 
\[
\vlderivation{
\vlin{\Box_i}{}{\Box_i A, \Gamma_1, \ckbox B}
{\Proofleaf{}{A, \Gamma'}}}\rlap{\qquad ,}
\]
and we have shown the second disjunct of our lemma.  The cases for
$\vlor$ and $\vlan$ are similar.

If the last rule is the $\ckdia$-rule then the following
transformation yields a shorter proof:
\[
\vlderivation{
\vlin{\ckdia}{}{\Box_i A, \Gamma_1, \ckdia B}
{\Proofleaf{\PP}{\Box_i A, \Gamma_1, \ckdia B, \Diamond B}}}
\qquad\leadsto\qquad
\vlderivation{
\vlin{\ctr^*}{}{\Box_i A, \Gamma_1, \ckdia B}
{\vlin{\neg{\vlor}^*}{}{\Box_i A, \Gamma_1, \ckdia B, \Diamond_1
B,\dots,\Diamond_n B }
{\Proofleaf{\PP}{\Box_i A, \Gamma_1, \ckdia B, \Diamond B}}}}
\rlap{\qquad ,}
\]
where by assumption of $\ckdia$-saturation all the $\Diamond_i B$ are
in $\Gamma_1$. To this proof we can now apply the induction hypothesis
which yields our lemma.

If the last rule in the given proof is the $\Box_j$-rule, then we
distinguish two cases. First, if $\Box_i A$ is the active formula then the
second disjunct of our lemma is either immediate or obtained via
weakening admissibility if the rule instance is not canonical.

Second, if $\Box_i A$ is not the active formula then the proof is of
the form
\[
\vlderivation{
\vlin{\Box_j}{}{\Box_i A, \Gamma_1, \Box_j B}
{\Proofleaf{\PP}{\Gamma'}}}
\rlap{\qquad ,}
\]
where the formula $\Box_i A$ has been introduced inside $\Sigma$. We
can thus change it into a proof 
\[
\vlderivation{
\vlin{\Box_j}{}{\Gamma_1, \Box_j B}
{\Proofleaf{\PP}{\Gamma'}}}
\rlap{\qquad ,}
\]
which shows the first disjunct of our lemma.
\end{proof}

In order to translate a derivation with deep rule applications into a
derivation where only shallow rules are allowed we need a way of
simulating the deep applicability. It turns out that, for certain
shallow rules, if they are admissible for the shallow system, then
their ``deep version'' is also admissible.

\begin{definition}[Make a shallow rule deep]
  Let $C\context$ be a formula context. Given a rule $\rho$ we define
  a rule \emph{rule $C\cons{\rho}$} as follows: an instance of the
  rule $\rho$ is shown on the left iff an instance of the \emph{rule
    $C\cons{\rho}$} is shown on the right:
  \[
  \vlinf{\rho}{}{\Gamma,A}{\Gamma, A_1\,\dots\,\Gamma,A_i\,\dots}
\qquad\qquad
  \vlinf{C\cons{\rho}}{}{\Gamma, C\cons{A}}
    {\Gamma,C\cons{A_1}\,\dots\,\Gamma, C\cons{A_i}\,\dots} 
 \rlap{\qquad .}
  \]
  We define a \emph{restricted context} as a formula context in which
  the hole is in the scope of at most the connectives from $\{\vlor,
  \Box_1,\dots,\Box_h\}$. Given a rule $\rho$ we define the rule
  \emph{rule $\check{\rho}$} as follows: its set of instances is the
  union of all sets of instances of $C\cons{\rho}$ where $C\context$
  is a restricted context.
\end{definition}

\begin{lemma}[Deep applicability preserves finite admissibility]\label{l:depth}
  Let $C\context$ be a restricted context.\\
  (i) There is an $n$ such that for all $\Gamma$ we have
  $\;\calcgentzen \provable{n}{0} \Gamma, C\cons{p
    \vlor \neg p}\;$ . \\
  (ii) If a rule $\rho$ is finitely admissible for $\calcgentzen$ then
  $C\cons{\rho}$ is also finitely admissible for system
  $\calcgentzen$.\\
  (iii) If a rule $\rho$ is finitely admissible for $\calcgentzen$ then
  $\check{\rho}$ is also finitely admissible for system
  $\calcgentzen$.\\
\end{lemma}
\begin{proof}
  Statement (iii) is immediate from (ii). Both (i) and (ii) are proved
  by induction on $C\context$. The case with
  $\;C\context=C_1\context\vlor C_2\;$ is of course analogous to the
  case with $\;C\context=C_1\vlor C_2\context\;$ and is omitted. We
  first prove (i). The case that $C\context$ is empty is handled by an
  application of the $\vlor$-rule. If $\;C\context=C_1\vlor
  C_2\context\;$ or $\;C\context=\Box_i C_1\context\;$ then we obtain
  a proof respectively as follows:
\[
\vlderivation{
\vlin{\vlor}{}{\Gamma, C_1\vlor C_2\cons{p\vlor\neg p}}
{\Proofleaf{\PP}{\Gamma, C_1,C_2\cons{p\vlor\neg p}}}}
\qquad \mbox{or} \qquad
\vlderivation{
\vlin{\Box_i}{}{\Gamma, \Box_i C_1\cons{p\vlor\neg p}}
{\Proofleaf{\PP}{C_1\cons{p\vlor\neg p}}}}
\]
where in both cases $\PP$ exists by induction hypothesis. For
statement (ii) the case that $C\context$ is empty is clear, so we
assume that it is non-empty. If $\;C\context=C_1\vlor C_2\context\;$
then the following transformation proves our claim:
\[
\vlderivation{
\vliiin{C_1\vlor C_2\cons{\rho}}{}{\Gamma, C_1\vlor
  C_2\cons{A}}{\vlhy{\vdots}}{\Proofleaf{\PP_k}{\Gamma, C_1\vlor
    C_2\cons{A_k}}}{\vlhy{\vdots}}}
\qquad \leadsto \qquad
\vlderivation{
\vlin{\vlor}{}{\Gamma, C_1 \vlor C_2\cons{A}}{
\vliiin{C_2\cons{\rho}}{}{\Gamma, C_1,
  C_2\cons{A}}{\vlhy{\vdots}}{
\vlin{\neg \vlor}{}{\Gamma, C_1,
    C_2\cons{A_k}}{\Proofleaf{\PP_k}{\Gamma, C_1\vlor
    C_2\cons{A_k}}}}{\vlhy{\vdots}}}}
\]

If $\;C\context=\Box_i C_1\context\;$ then we have the following situation:
\[
\vlderivation{ 
\vliiin{\Box_i C_1\cons{\rho}}{}{\Gamma, \Box_i
    C_1\cons{A}}{\vlhy{\vdots}}{\Proofleaf{\PP_k}{\Gamma, \Box_i
      C_1\cons{A_k}}}{\vlhy{\vdots}}}
\mathrlap{\qquad .}
\]

In order to apply quasi-invertibility of $\Box_i$,
Lemma~\ref{l:boxquinv}, we first need to replace the shown instance of
the rule $\Box_i C_1\cons{\rho}$ by several instances of it which are
applied in a context which is $\ckdia$-saturated and free of hiding
formulas. We apply conjunction invertibility, disjunction
invertibility and weakening admissibility to each $\PP_k$ to obtain a
sequence of proofs $\PP_{k1}\dots\PP_{km}$ such that for each $k$
there is a proof of the form
\[
\toks0={0.5}
\vlderivation{
\vltrf{\vlan,\vlor,\ckdia}{\Gamma, \Box_i C_1\cons{A_k} }
{\Proofleaf{\PP_{k1}}{\Gamma_1, \Box_i C_1\cons{A_k}}}
{\vlhy{\dots}}
{\Proofleaf{\PP_{km}}{\Gamma_m, \Box_i C_1\cons{A_k}}}
{\the\toks0}
}
\mathrlap{\qquad ,}
\]
where each $\Gamma_j$
is $\ckdia$-saturated and free of hiding formulas.

Fix some $j$. For all $k$ apply quasi-invertibility of $\Box_i$,
Lemma~\ref{l:boxquinv}, to the proof $\PP_{kj}$.   Either this yields some proof $\PP$ of
$\Gamma_j$ or for each $k$ it yields a proof $\PP'_{kj}$ of some sequent
$\Gamma_j', C_1\cons{A_k}$. Then we can build either
\[
\vlderivation{ 
\vlin{\wk}{}{\Gamma_j, \Box_i
    C_1\cons{A}}{\Proofleaf{\PP}{\Gamma}}} \qquad \mbox{or} \qquad
\vlderivation{ 
\vlin{\Box_i}{}{\Gamma_j, \Box_i
    C_1\cons{A}}{
\vliiin{C_1\cons{\rho}}{}{\Gamma_j',
    C_1\cons{A}}{\vlhy{\vdots}}{\Proofleaf{\PP'_{kj}}{\Gamma_j',
      C_1\cons{A_k}}}{\vlhy{\vdots}}}}
\mathrlap{\qquad ,}
\]
where in the second case $C_1\cons{\rho}$ is finitely admissible by
induction hypothesis. Repeat this argument for each $j$ with $1\leq j
\leq m$, which for each $j$ yields a proof $\PP''_j$ in
$\calcgentzen$. From those we build
\[
\toks0={0.5}
\vlderivation{
\vltrf{\vlan,\vlor,\ckdia}{\Gamma, \Box_i C_1\cons{A} }
{\Proofleaf{\PP''_1}{\Gamma_1, \Box_i C_1\cons{A}}}
{\vlhy{\dots}}
{\Proofleaf{\PP''_m}{\Gamma_m, \Box_i C_1\cons{A}}}
{\the\toks0}
}
\mathrlap{\qquad ,}
\]
which shows our lemma.
\end{proof}

\begin{figure}
  \centering
  \fbox{
\parbox{\ftextwidth}{
\medskip
\[
\cinf{{\sf g_c}}{\Gamma, {B \vlor A}}{\Gamma, {A \vlor B}}
\qquad\qquad
\cinf{{\sf g_a}}{\Gamma, {A \vlor (B \vlor C)}}{\Gamma, {(A
    \vlor B) \vlor C}}
\]
\smallskip
\[
\cinf{{\sf g}_{\ctr}}{\Gamma,A}{\Gamma,A\vlor A}
\qquad 
\cinf{{\sf g}_{\Diamond}}{\Gamma, \Diamond_i A ,\Box_i B}
{\Gamma, \Box_i (A \vlor B)}
\qquad 
\cinf{{\sf g}_{\ckdia}}{\Gamma, \ckdia A}{\Gamma, \Diamond^k A}
\text{ where }{k\geq 1}
\]
}}
  \caption{Some glue}
  \label{fig:glue}
\end{figure}

\begin{lemma}[Some glue]\label{l:glue}
  The rules in Figure~\ref{fig:glue} are finitely admissible for
  system $\calcgentzen$.
\end{lemma}
\begin{proof}
  The rules $\sf g_c,g_a$ and $\sf g_{ctr}$ are easily seen to be
  finitely admissible by using invertibility of the $\vlor$-rule. For
  the ${\sf g_{\Diamond}}$-rule we proceed by induction on the given
  proof of the premise and make a case analysis on the last rule in
  this proof.  All cases are trivial except when this is the
  $\Box_i$-rule. We distinguish two cases: either 1) $\Box_i(A \vlor
  B)$ is the active formula or 2) it is not. In the first case we
  have:
\[
\vlderivation{
\vlin{{\sf g_{\Diamond}}}{}
{\Sigma, \ckdia \Delta, \Diamond_i \Lambda, \Diamond_i A, \Box_i B}
{\vlin{\Box_i}{}
{\Sigma, \ckdia \Delta, \Diamond_i \Lambda, \Box_i (A
\vlor B)}{\Proofleaf{\PP}
{\ckdia \Delta, \Lambda, A\vlor B}}}
}
\qquad \leadsto \qquad
\vlderivation{
\vlin{\Box_i}{}
{\Sigma, \ckdia \Delta, \Diamond_i \Lambda, \Diamond_i A, \Box_i B}
{\vlin{\inverse{\vlor}}{}
{\ckdia \Delta, \Lambda, A, B}
{\Proofleaf{\PP}{\ckdia \Delta, \Lambda, A\vlor B}}}
}
\]
and in the second case we have the following:
\[
\vlderivation{
\vlin{{\sf g_{\Diamond}}}{}
{\Box_i C, \Gamma',\Diamond_i A, \Box_i B}
{\vlin{\Box_i}{}
{\Box_i C, \Gamma', \Box_i(A\vlor B)}
{\Proofleaf{\PP}{C, \Gamma''}}}}
\qquad \leadsto \qquad
\vlderivation{
\vlin{\Box_i}{}{\Box_i C, \Gamma',\Diamond_i A, \Box_i B}
{\Proofleaf{\PP}{C, \Gamma''}}
}
\mathrlap{\qquad .}
\]
For the ${\sf g_{\ckdia}}$-rule we proceed by induction on $k$ and a
subinduction on the depth of the given proof of the premise. For $k=1$
the ${\sf g_{\ckdia}}$-rule coincides with the $\ckdia$-rule plus a
weakening, so we assume that we have a proof of $\Gamma,
\Diamond^{k+1} A$. By invertibility of the $\vlor$-rule we obtain a
proof
\[
\Proof{\PP}{\Gamma, \Diamond_1\Diamond^k A, \dots ,
  \Diamond_h\Diamond^k A}
\]
of the same depth. By induction on the depth of $\PP$ and a case analysis
on the last rule in $\PP$ we now show that we have a proof of the same
depth of $\Gamma, \ckdia A$. All cases are trivial except when the
last rule is $\Box_i$. Then the following transformation:
\[
\vlderivation{
\vlin{\Box_i}{}
{\Box_i B, \Diamond_i \Delta, \ckdia \Lambda, \Sigma,
\Diamond_1 \Diamond^k A, \dots, \Diamond_h \Diamond^k A }
{\Proofleaf{\PP'}
{B, \Delta, \ckdia \Lambda, \Diamond^k A}
}}
\quad\leadsto\quad
\vlderivation{
\vlin{\Box_i}{}
{\Box_i B, \Diamond_i \Delta, \ckdia \Lambda, \Sigma, \ckdia A}
{\vlin{{\sf g_{\ckdia}}}{}
{B, \Delta, \ckdia \Lambda, \ckdia A}
{\Proofleaf{\PP'}
{B, \Delta, \ckdia \Lambda, \Diamond^k A}
}}}
\]
proves our claim, where the instance of the ${\sf g_{\ckdia}}$-rule on the
right is finitely admissible by the outer induction hypothesis.
\end{proof}

For our translation from deep into shallow we translate nested
sequents into formulas and thus fix an arbitrary order and association
among elements of a sequent. The arbitrariness of this translation
gets in the way, and we work around it as follows: we write
\[
\cinf{{\sf ac}}{B}{A}
\]
 if the formula $B$ can be derived
from the formula $A$ in $\{{\sf \check{g}_c},{\sf \check{g}_a}\}$.
Clearly, in that case $A$ and $B$ are equal modulo commutativity and
associativity of disjunction. The converse is not the case. For
example $\ckdia (C \vlor D)$ can not be derived from $\ckdia (D \vlor
C)$ by {\sf ac}, in general. Note that since ${\sf \check{g}_c}$ and
${\sf \check{g}_a}$ are finitely admissible for system {\calcgentzen},
so is the rule {\sf ac}.

\begin{theorem}[Deep into shallow]\nl
If $\;\calcdeep \provable{\alpha}{0}\Gamma\;$ then we
have $\;\calcgentzen
  \provable{\,\omega\cdot(\alpha+1)}{0}\formula{\Gamma}\;$.\\
\label{t:deepshallow}
\end{theorem}

\begin{proof} By induction on $\alpha$. If the endsequent of the given
  proof is of the form $\Gamma\cons{p,\neg p}$, then we have
\[
\Gamma\cons{p,\neg p}
\qquad\leadsto\qquad
\vlderivation{
\vlin{\ac}{}{\formula{\Gamma\cons{p,\neg p}}}
{\Proofleaf{\PP}{\formula{\Gamma}\cons{p\vlor \neg p}}}}
\]
where $\PP$ is of finite depth by Lemma~\ref{l:depth} and $\ac$ is
finitely admissible by Lemma~\ref{l:glue} and Lemma~\ref{l:depth}. If
the last rule is the $\vlor$-rule then an application of $\ac$ proves
our claim. The case of the $\Box_i$-rule is trivial since the
corresponding formula for the premise is the corresponding formula of
the conclusion. For the $\ckbox$-rule we apply the following
transformation, where the $\PP'_k$ are obtained by induction
hypothesis:
\[
\vlderivation{
\vliiin{\ckbox}{\;^{1\leq k < \omega}}
{\Gamma\cons{\ckbox A}}
{\vlhy{\vdots}}
{\Proofleaf{\PP_k}{\Gamma\cons{\Box^k A}}}
{\vlhy{\vdots}}
}
\qquad\leadsto\qquad
\vlderivation{
\vlin{\ac}{}{\formula{\Gamma\cons{\ckbox A}}}{
\vliiin{\formula{\Gamma}\cons{\ckbox}}
{\;^{1\leq k < \omega}}
{\formula{\Gamma}\cons{\ckbox A}}
{\vlhy{\vdots}}
{\vlin{\ac}{}{\formula{\Gamma}\cons{\Box^k A}}
{\Proofleaf{\PP'_k}{\formula{\Gamma\cons{\Box^k A}}}}}
{\vlhy{\vdots}}}
}
\]

Let the depth of the proof on the left be $\beta$ with
$\beta\leq\alpha$ and the depth of a proof $\PP_k$ be $\beta_k$. Note
that the depth of the {\sf ac}-derivations both below and above the
infinitary rule is bounded by a finite ordinal $m$ because the context
$\Gamma\context$ is finite. Then, by finite admissibility of the rule
$\formula{\Gamma}\cons{\ckbox}$ (Lemma~\ref{l:depth}) there is a
finite ordinal $n$ such that the proof on the right has the depth
\[
\begin{array}{ll}
  & \sup_k(|\PP'_k|+m+1)+n+m \;<\; 
\sup_k(|\PP'_k|)+\omega\\
  \leq & \sup_k(\omega \cdot (\beta_k+1))+\omega \;=\; 
\omega\cdot \sup_k(\beta_k+1)+\omega\\
=    & \omega\cdot \beta +\omega \;=\; \omega\cdot(\beta+1) \;\leq\;
\omega\cdot(\alpha+1) \rlap{\qquad .}\\
\end{array}
\]
The case for the $\vlan$-rule is similar. For the $\Diamond_i$-rule
we apply the following transformation, where $\PP'$ is obtained by
induction hypothesis and the bound on the depth is easy to check:
\[
\vlderivation{
\vlin{\Diamond_i}{}{\Gamma\cons{\Diamond_i A, [\Delta]_i}}
{\Proofleaf{\PP}{\Gamma\cons{\Diamond_i A, [A, \Delta]_i}}}}
\qquad\leadsto\qquad
\vlderivation{
\vlin{\ac}{}{\formula{\Gamma
\cons{\Diamond_i A, [\Delta]_i}}}{
\vlin{\formula{\Gamma}\cons{{\sf g_{ctr}} \vlor \Box_i\formula{\Delta}}}{}
{\formula{\Gamma}
\cons{\Diamond_i A \vlor\Box_i \formula{\Delta}}}{
\vlin{\ac}{}
{\formula{\Gamma}
\cons{(\Diamond_i A \vlor \Diamond_i A) \vlor\Box_i \formula{\Delta}}}{
\vlin{\formula{\Gamma}\cons{\Diamond_i A \vlor {\sf g_{\Diamond}}}}{}
{\formula{\Gamma}
\cons{\Diamond_i A \vlor (\Diamond_i A \vlor\Box_i \formula{\Delta})}}{
\vlin{\ac}{}{\formula{\Gamma}
\cons{\Diamond_i A \vlor \Box_i (A\vlor \formula{\Delta})}}{
\Proofleaf{\PP'}
{\formula{\Gamma\cons{\Diamond_i A, [A, \Delta]_i}} }}}}}}}
\rlap{\quad .}
\]
Note that here a rule like $C\cons{\rho \vlor A}$ means the rule $\rho$
applied in the context $C\cons{\context \vlor A}$, and is finitely
admissible for $\calcgentzen$ if is $\rho$ is finitely admissible for
$\calcgentzen$, by Lemma~\ref{l:depth}. 

The case for the $\ckdia$-rule
is similar.
\end{proof}

We can now state the cut-elimination theorem for the shallow system.

\begin{theorem}[Cut-elimination for the shallow system]\nl
  If\/ $\calcgentzen \prov{\alpha}{\omega\cdot n}\Gamma$ then\/ $\calcgentzen
  \prov{\omega\cdot(\veblen{1}^n(\omega\cdot\alpha)+1)}{0}\Gamma$
\end{theorem}

\section{An Upper Bound on 
the Depth of Proofs}

The Hilbert system $\calchilb$ is obtained from some Hilbert system
for classical propositional logic by adding the axioms and rules shown
in Figure~\ref{fig:hc}. It is essentially the same as system $K^C_h$
from the book \cite{FHMV}, where also soundness and completeness are
shown. We will now embed $\calchilb$ into $\calcdeep+\cut$, keeping
track of the proof depth and thus, via cut-elimination for
$\calcdeep$, establish an upper bound for proofs in $\calcdeep$. Via
the embedding of the deep system into the shallow system, this bound also
holds for the shallow system.

\begin{figure}
\fbox{
\parbox{\ftextwidth}{
  \centering
\[
\kax\quad \Box_i A \vlan \Box_i(A \imp B) \imp \Box_i B 
\qquad
\cocloax\quad \ckbox A \imp (\Box A \vlan \Box \ckbox A) 
\]
\[
\cinf{\indrule}{B \imp \ckbox A}{B \imp (\Box A \vlan \Box B)} \qquad 
\cinf{\mprule}{B}{A \quad A \imp B} \qquad
\cinf{\necrule}{\Box_i A}{A}
\]
}}
  \caption{System \calchilb}
  \label{fig:hc}
\end{figure}

\begin{theorem}[$\calchilb$ into $\calcdeep+\cut$] If\/ $\calchilb \vdash A$ then 
$\calcdeep\prov{<\omega^2}{\omega^2} A$.
\label{t:hilb}
\end{theorem}
\begin{proof}
  The proof is by induction on the depth of the derivation in
  $\calchilb$.  If $A$ is a propositional axiom of $\calchilb$ then
  there is a finite derivation of $A$ in the propositional part of
  system $\calcdeep$ such that all premises are instances of the
  general identity axiom. Thus we obtain $\calcdeep\prov{\omega\cdot
    m}{0} A$ for some $m<\omega$ by admissibility of the general
  identity axiom (Lemma \ref{l:gid}).

  If $A$ is an instance of $\kax$, then we obtain
  $\calcdeep\prov{\omega\cdot m}{0} A$ for some $m<\omega$ from the
  following derivation and admissibility of the general identity axiom
  to take care of the premises.
\[
\vlderivation{
\vlin{\vlor^2}{}{\Box_i A \vlan \Box_i(A \imp B) \imp \Box_i B }
{\vlin{\Box_i}{}{\Diamond_i \neg A,\Diamond_i(A \vlan \neg B),\Box_i B}
{\vlin{\Diamond_i}{}{\Diamond_i \neg A,\Diamond_i(A \vlan \neg
    B),[B]_i}{ 
\vliin{\vlan}{}{\Diamond_i \neg A,\Diamond_i(A \vlan
      \neg B),[B, A \vlan \neg B]_i}{
\vlin{\Diamond_i}{}{\Diamond_i \neg
        A,\Diamond_i(A \vlan \neg B),[B, A]_i}{\vlhy{\Diamond_i \neg
          A,\Diamond_i(A \vlan \neg B),[B, A, \neg
          A]_i}}}{\vlhy{\Diamond_i \neg A,\Diamond_i(A \vlan \neg
        B),[B, \neg B]_i}}}}}}
\]

If $A$ is an instance of $\cocloax$, then we obtain
$\calcdeep\prov{\omega\cdot m}{0} A$ for some $m<\omega$ from the
following derivation and again admissibility of the general identity
axiom to take care of the premises. An argument similar to the one
used to derive the general identity axiom guarantees that all premises
of the $\ckbox$ rule are derivable with depth smaller than $\rk(\ckbox
A)$.

\[
\vlderivation{
\vlin{\vlor}{}{\ckbox A \imp (\Box A \vlan \Box \ckbox A) }{
\vliin{\vlan}{}{\ckdia \neg A, \Box A \vlan \Box\ckbox A}{
\vlin{\ckdia,\wk}{}{\ckdia \neg A, \Box A}{
\vlhy{\Diamond \neg A, \Box A}}}
{\vliiin{\vlan}{\;^{1\leq i \leq h}}{\ckdia \neg A, \Box\ckbox A}
  {\vlhy{\vdots}}
  {\vlin{\Box_i}{}{\ckdia \neg A, \Box_i\ckbox A}{
    \vliiin{\ckbox}{\;^{1\leq k < \omega}}{\ckdia \neg A, [\ckbox A]_i}
       {\vlhy{\vdots}} 
       {\vlin{\ckdia, \wk}{}{\ckdia \neg A, [\Box^k A]_i}
         {\vlin{\vlor, \wk}{}{\Diamond^{k+1} \neg A, [\Box^k A]_i }
           {\vlin{\Diamond_i, \wk}{}{\Diamond_i\Diamond^k \neg A,[\Box^k A]_i}
             {\vlhy{ [\Diamond^k \neg A, \Box^k A]_i}}}}}
       {\vlhy{\vdots}}}}
  {\vlhy{\vdots}}}
}}
\]

If the last rule in the derivation is an instance of $\mprule$, then
by the induction hypothesis there are $m_1,m_2,n_1,n_2<\omega$ such
that $\calcdeep\prov{\omega\cdot m_1}{\omega\cdot n_1} A$ and
$\calcdeep\prov{\omega\cdot m_2}{\omega\cdot n_2} A \imp B$.  Thus we
get $\calcdeep\prov{\omega\cdot m_1}{\omega\cdot n_1} A, B$ by
weakening admissibility and $\calcdeep\prov{\omega\cdot
  m_2}{\omega\cdot n_2} \neg A, B$ by invertibility.  An application
of $\cut$ yields $\calcdeep\prov{\omega\cdot m}{\omega\cdot n} B$ for
$m = \max(m_1,m_2)+1$ and $n = \max(n_1,n_2,\rk(A)+1)$.

If the last rule in the derivation is an instance of  $\necrule$, then the claim follows from the induction hypothesis, the fact that $\nec$ is
cut-rank- and depth-preserving admissible, and an application of $\Box_i$.

If the last rule in the derivation is an instance of  $\indrule$, then by the induction hypothesis there are $m_1,n_1<\omega$ such that
$\calcdeep\prov{\omega\cdot m_1}{\omega\cdot n_1} B \imp (\Box A \vlan \Box B)$.
Then by invertibility of the $\vlan$- and $\vlor$-rules we obtain 
\[
1)\; \calcdeep\prov{\omega\cdot m_1}{\omega\cdot n_1} \neg B, \Box B 
\qquad \text{ and }\qquad
2)\; \calcdeep\prov{\omega\cdot m_1}{\omega\cdot n_1} \neg B, \Box A.
\]
Let $n_2$ be such that $\rk(\Box B) < \omega \cdot n_2$. We set $n = \max(n_1,n_2)$.
By induction on $k$ we show that
for all $k\geq 1$ there is an $m_2<\omega$ such that
$\calcdeep\prov{\omega\cdot m_1 + m_2}{\omega\cdot n} \neg B, \Box^k A$.
The case $k=1$
is given by 2) and the induction step is as follows:
\[
\vlderivation{
\vliin{\cut}{}{\neg B, \Box^{k+1} A}
{\vlhy{\neg B, \Box B}}
{\vliiin{\vlan}{\;^{1\leq i\leq h}}{\Diamond \neg B, \Box^{k+1}A}
{\vlhy{\vdots}}
{\vlin{\vlor,\wk}{}{\Diamond \neg B, \Box_i\Box^{k}A}
{\vlin{\Box_i}{}{\Diamond_i \neg B, \Box_i\Box^{k}A}
{\vlin{\Diamond_i, \wk}{}{\Diamond_i \neg B, [\Box^{k}A]_i}
{\vlin{\nec}{}{[\neg B, \Box^{k}A]_i}
{\vlhy{\neg B, \Box^{k}A}}}}}}
{\vlhy{\vdots}}}
} \qquad ,
\]
where the premise on the left is 1) and the premise on the right
follows by induction hypothesis. The claim follows by applications of
$\ckbox$ and $\vlor$.
\end{proof}

The embedding of the Hilbert system into the nested sequent system
together with the cut-elimination theorem for the deep system gives us
the following upper bounds on the depth of proofs in the cut-free
systems.

\begin{theorem}[Upper bounds] If $A$ is a valid formula then\nl
(i) $\calcdeep \prov{<\veblen{2}{0}}{0} A$, and\\
(ii) $\calcgentzen \prov{<\veblen{2}{0}}{0} A$.
\end{theorem}

\begin{proof}
  If $A$ is valid then by completeness of $\calchilb$ we have
  $\calchilb \vdash A$ and by the embedding of the Hilbert system into
  the nested sequent system there are natural numbers $m,n$ such that
  $\calcdeep \prov{\omega\cdot m}{\omega\cdot n} A$.  By the cut
  elimination theorem for the nested sequent system we obtain $\calcdeep
  \prov{\veblen{1}^n(\omega\cdot m)}{0} A$.  We know
  $\veblen{\beta_1}{\gamma_1} < \veblen{\beta_2}{\gamma_2}$ if
  $\beta_1 < \beta_2$ and $\gamma_1 < \veblen{\beta_2}{\gamma_2}$.
  Thus $\calcdeep \prov{<\veblen{2}{0}}{0} A$. For (ii) by the
  embedding of the deep system into the shallow system it suffices to
  check that for $\alpha<\veblen{2}0$ we have
  $\omega\cdot(\alpha+1)<\veblen{2}0$.
\end{proof}

\section{Discussion}

We have introduced a nested sequent system for common knowledge which,
in contrast to the ordinary sequent system by Alberucci and J\"ager,
admits a syntactic cut-elimination procedure. We have shown this
cut-elimination procedure, and, via embedding the two systems into
each other, have also provided a cut-elimination procedure for the
shallow system. We embedded a Hilbert style system and obtained
$\veblen{2}{0}$ as upper bound on the depth of cut-free proofs for
both sequent systems.

Notice in particular how we used the nested sequent system as a tool
in order to prove a result about an existing, ordinary sequent system.
Designing some kind of infinitary proof system with a syntactic
cut-elimination for the logic of common knowledge was the less
interesting part: we could have used both the display calculus and
labelled sequents to do so. However, it is hard to imagine how a
cut-free display calculus or cut-free labelled sequent calculus could
have been translated back into our cut-free ordinary sequent calculus.
The fact that nested sequents stay inside the modal language allowed
us to do so.

{\bf Other modal logics.} We have looked at common knowledge based on
the least normal modal logic. In a sense, \emph{common belief} would
be a better name. Given the previous chapter, it seems that any other
modal logic of the cube could be used instead and thus our approach is
independent of the particular underlying axiomatisation of knowledge.
The modal logic {\sf S5} is often proposed as an adequate logic for
knowledge.  As we have seen in the previous chapter, contrary to
shallow sequents, nested sequents can easily handle {\sf S5}. So it is
easy to design a system for {\sf S5}-based common knowledge.  We just
need to add a single rule to system $\calcdeep$:
\[
\cinf{{\sf S5}}
{\Gamma\cons{\Diamond A}\cons{\emptyset}}
{\Gamma\cons{\Diamond A}\cons{A}} \qquad .
\]
However, hypersequents are already sufficient to capture {\sf S5} and a system
for {\sf S5}-based common knowledge based on the hypersequent system
{\sf LS5} by Mints \cite{MintsHS} seems to admit a cut-elimination
procedure similar to the one given here for system $\calcdeep$. So for
{\sf S5}-based common knowledge there is no need for nested sequents.

{\bf Future work.} Of course there are also more speculative questions. What is the
mathematical meaning of the upper bound on the depth of cut-free
proofs? Is there a kind of boundedness lemma in modal logic similar to
the one used in the analysis of set theories and second order
arithmetic?  What would be the equivalent of a well-ordering proof in
modal logic? Is $\veblen{2}{0}$ the best possible upper bound on the
depth of proofs?  Can our analysis be extended to more powerful modal
fixpoint logics? Work under preparation suggests that a cut-free
infinitary nested sequent system for the $\mu$-calculus can be bounded
by $\veblen{\omega}{0}$. 

A more interesting, and harder problem is the design of cut-free
finitary sequent systems for modal fixpoint logics. While such systems
exist, for example for temporal logics \cite{Gaintz07,BL08}, their
rules are context-dependent in a way which makes it hard to study them
proof-theoretically, in particular it seems hard to design a syntactic
cut-elimination procedure.

\bibliographystyle{plain}
\bibliography{../../kai}

\end{document}